\documentclass[11pt,pdfa,letterpaper]{article}
\newif\ifcomments

\commentsfalse

\usepackage[in]{fullpage}

\usepackage{iftex}
\ifPDFTeX
  \usepackage[utf8]{inputenc}
  \usepackage[noTeX]{mmap}
  \usepackage[T1]{fontenc}
\fi
\ifLuaTeX
  \usepackage{luatex85}
  \usepackage[noTeX]{mmap}
\fi

\usepackage{comment}
\usepackage{amsfonts}
\usepackage{amsmath}
\usepackage{mathtools,amsthm,amssymb} %
\usepackage{xcolor,xstring,xparse}
\usepackage{graphicx}
\usepackage{comment}
\usepackage{qtree}
\usepackage{tree-dvips}
\usepackage{float}
\definecolor{linkblue}{HTML}{001487}
\usepackage[colorlinks=true,allcolors=linkblue]{hyperref}
\usepackage[nameinlink,capitalize,noabbrev]{cleveref}
\usepackage{braket}
\usepackage{mathrsfs}
\usepackage{tikz}
\usepackage{qcircuit}
\usepackage{xspace}
\usepackage{dsfont}
\usepackage{enumitem} %
\usepackage{longfbox}
\usepackage{csquotes}
\usepackage{authblk}
\usepackage{longfbox}

\usepackage{adjustbox}
\usepackage{gensymb}
\usepackage{siunitx}
\usepackage{stmaryrd}
\usepackage{float}
\usepackage{amsthm}
\usepackage{tikz-cd}
\usepackage{tikz}
\usetikzlibrary{backgrounds}
\usepackage{hyperref}

\hypersetup{
    colorlinks = true,
    linkcolor = violet,
    linkbordercolor = {white}
}

\usepackage{booktabs}    
\usepackage{tabularx}    

\usepackage{geometry}
\usepackage[style=alphabetic,minalphanames=3,maxalphanames=4,maxnames=99,backref=true]{biblatex}

\usepackage{url}
\usepackage{bbm}
\usepackage{hyperref} 
\hypersetup{breaklinks=true}

\usetikzlibrary{fadings}
\usetikzlibrary{patterns}
\usetikzlibrary{shadows.blur}
\usetikzlibrary{shapes}

\usepackage{framed}
\usepackage{longtable}
\usepackage{array}
\usepackage{ltablex}
\keepXColumns

\usepackage{caption}
\ifcomments
  \newcommand{\andrey}[1]{{\color{red}\bf (Andrey: #1)}} 
  \newcommand{\jon}[1]{{\color{blue}\bf (Jonathan: #1)}}
\else
  \newcommand{\andrey}[1]{}
  \newcommand{\jon}[1]{}
\fi

\theoremstyle{plain}
\newtheorem{theorem}{Theorem}[section]
\newtheorem*{theorem*}{Theorem}
\newtheorem{lemma}[theorem]{Lemma}

\Crefname{claim}{Claim}{Claims}

\newtheorem*{lemma*}{Lemma}

\newtheorem*{corollary*}{Corollary}
\newtheorem{proposition}[theorem]{Proposition}
\newtheorem{conjecture}[theorem]{Conjecture}

\theoremstyle{definition}
\newtheorem{definition}[theorem]{Definition}

\newtheorem{construction}{Construction}

\theoremstyle{remark}

\numberwithin{equation}{section}

\newcommand{\F}{\mathbb{F}}
\newcommand{\R}{\mathbb{R}}

\newcommand{\Z}{\mathbb{Z}}

\newcommand{\ketbra}[2]{|#1\rangle\langle#2|}

\renewcommand{\vec}[1]{\mathbf{#1}}

\newcommand{\bounds}[2]{\bigg\rvert_{#1}^{#2}}
\newcommand{\boudns}[2]{\bounds} 

\renewcommand{\a}{\alpha}
\renewcommand{\b}{\beta}
\newcommand{\g}{\gamma}
\newcommand{\G}{\Gamma}

\newcommand{\D}{\Delta}

\renewcommand{\k}{\kappa}
\newcommand{\m}{\mu}

\newcommand{\s}{\sigma}
\renewcommand{\t}{\tau}

\newcommand{\CE}{\mathcal{E}}

\newcommand{\CG}{\mathcal{G}}

\newcommand{\CL}{\mathcal{L}}

\newcommand{\CP}{\mathcal{P}}

\newcommand{\CS}{\mathcal{S}}

\DeclareMathOperator{\wt}{wt}

\newcommand{\Cliff}{\mathsf{Cliff}}

\makeatletter
\renewcommand{\paragraph}{%
  \@startsection{paragraph}{4}%
  {\z@}{2.25ex \@plus 1ex \@minus .2ex}{-1em}%
  {\normalfont\normalsize\bfseries}%
}
\makeatother
\title{Mirror codes: High-threshold quantum LDPC codes beyond the CSS regime}

\author[1]{\quad\quad Andrey Boris Khesin\footnote{\url{andrey.khesin@cs.ox.ac.uk}}}
\author[2]{Jonathan Z. Lu\footnote{\url{lujz@mit.edu}}}
\affil[1]{Department of Computer Science, University of Oxford}
\affil[2]{Department of Mathematics, Massachusetts Institute of Technology}

\begin{document}

\pagenumbering{gobble}

\clearpage\maketitle
\thispagestyle{empty}

\begin{abstract}
The realization of a fault-tolerant quantum memory rests on our ability to implement quantum error correction protocols whose logical error rates are suppressed far below physical error rates.
Such protocols rely in turn on an intricate combination: the error-correcting code's efficiency, the syndrome extraction circuit's fault tolerance and overhead, the decoder's quality, and the device's constraints, such as physical qubit count and connectivity.

This work makes two contributions towards error-corrected quantum devices. 
First, we introduce mirror codes, a simple yet flexible construction of LDPC stabilizer codes parameterized by a group $G$ and two subsets of $G$ whose total size bounds the check weight.
Up to permutation and local Cliffords, these codes contain all abelian two-block group algebra codes, such as bivariate bicycle (BB) codes.
At the same time, they are manifestly not CSS in general, thus deviating substantially from most prior constructions.
Fixing a check weight of 6, we find $\llbracket 60, 4, 10 \rrbracket, \llbracket 36, 6, 6 \rrbracket, \llbracket 48, 8, 6 \rrbracket$, and $\llbracket 85, 8, 9 \rrbracket$ codes, all of which are not CSS; we also find several weight-7 codes with $kd > n$.

Next, we construct a series of syndrome extraction circuits that trade overhead for provable fault tolerance, which may be of independent interest.
These circuits use 1-2, 3, and 6 ancillae per check, and respectively are partially fault-tolerant (FT), provably FT on weight-6 CSS codes, and provably FT on \emph{all} weight-6 stabilizer codes.
Using our constructions, we perform end-to-end quantum memory experiments on several representative mirror codes under circuit-level noise.
We achieve an error pseudothreshold on the order of $0.2\%$, approximately matching that of the $\llbracket 144, 12, 12 \rrbracket$ BB code under the same model.
These findings position mirror codes as a versatile candidate for fault-tolerant quantum memory, especially on smaller-scale devices in the near term.
\end{abstract}

\newpage
\tableofcontents

\newpage
\pagenumbering{arabic} 
\setcounter{page}{1}   


\section{Introduction}
\label{sec:introduction}

Quantum computers hold the potential to efficiently solve computational problems of broad interest which are otherwise intractable to classical computers~\cite{algo1,algo2,algo3,algo4,algo5,algo6,algo7,algo8,algo9,algo10,algo11,algo12,algo13,algo14,algo15,algo16,algo17}.
Almost all quantum algorithms are, however, highly sensitive to noise, requiring very low error rates on the order of at most $10^{-9}$ in order to produce meaningful output.
The fragile nature of qubits suggests that directly achieving physical error rates low enough for algorithmic implementation may be infeasible.
Consequently, quantum error correction has become a leading proposal for the realization of large-scale quantum algorithm execution~\cite{ec5,ec6,ec7}.
At a high level, quantum error correction encodes $k$ logical qubits into a certain subspace of a larger $n$-qubit space. 
By choosing this subspace carefully, one can protect the logical qubits from noise to a far greater extent compared to direct usage of the logical qubits without quantum error correction.
In fact, if the physical error rate is below a certain threshold value $p_{\text{th}}$, then a quantum error-correcting code can suppress the error rate on its logical qubits to a value far below the physical error rate~\cite{ec1,ec2,ec3,ec4}.
We can hope, therefore, that if we can design quantum error correction protocols for which $p_{\text{th}}$ is below achievable physical error rates, then we may use these protocols to ultimately realize quantum algorithms.
Note that even a successful quantum error correction protocol is itself insufficient to implement quantum algorithms fault-tolerantly, as it is not necessarily clear how to perform \emph{computation} on the logical space.
Such a protocol instead only implies a fault-tolerant quantum \emph{memory}, i.e. error-resistant storage of quantum information.
Nevertheless, achieving fault-tolerant memory on quantum devices would mark a significant first step towards algorithmic execution.

The demonstration of quantum error correction below threshold is, however, a challenging task rife with subtleties.
In the classical setting, in order to gain some tractable structure, we typically restrict our attention to linear codes, i.e. linear subspaces of $\set{0, 1}^n$.
The quantum analog of linear codes is known as \emph{stabilizer codes}, and we likewise restrict ourselves only to such codes~\cite{gottesman1997stabilizer}.
To decode a stabilizer code, we first measure certain Pauli operators which yield a classical bitstring (known as the \emph{syndrome}) that encodes information about which error occurred.
We then use a classical decoding algorithm to infer and undo the error.
A first obstacle is that the quantum error-correcting code must be chosen to both protect against many errors and to be efficiently decodable.
These requirements generally restrict potential codes to either algebraically structured codes~\cite{alg1,alg2,alg3}, which have provable decoders, or codes with a property known as low density parity checks (LDPC), which have general-purpose decoders that demonstrate excellent practical performance~\cite{breuckmann2021quantum,panteleev2021degenerate}.
A second obstacle, however, is that the process of obtaining the syndrome from the code is itself subject to significant noise.
To reliably extract the syndrome, one must employ a number of ancillary qubits that scales with the check weight (weight of the Pauli operators measured) of the stabilizer code~\cite{ec1}.
Thus, at least in the near term during which qubit overhead is significantly costly, quantum error correction protocols are often restricted to LDPC codes, which have very small check weight.

Even among LDPC codes, the realization of quantum error correction (QEC) protocols remains difficult because the threshold $p_{\text{th}}$ depends sensitively on several objects: the robustness of the code (as measured by, e.g., its distance), the fault tolerance and overhead of the circuit performing syndrome extraction, and the quality and efficiency of the decoder.
The device's native constraints moreover play a vital role, as they determine the amount of acceptable qubit overhead and the physical error rate, which in turn determine how large of a code and how fault-tolerant of a syndrome extraction circuit may be used.
In particular, relatively small changes in device requirements may affect multiple components determining the threshold, rendering flexible QEC protocol constructions challenging.

Currently, the leading proposal for QEC utilizes the surface code, a quantum LDPC code with several useful properties, including geometrically local parity checks of weight 4, low-overhead fault-tolerant syndrome extraction, and a highly reliable decoder~\cite{google2025quantum}.
This proposal is moreover quite flexible to device requirements and scaling; the primary drawback is that the surface code encodes very few logical qubits relative to physical qubits, and thus requires large quantum devices to meaningfully encode a substantial amount of information.
To improve coding efficiency, recent proposals have constructed more general families of quantum LDPC codes that can achieve a threshold even with the simplest, lowest-overhead syndrome extraction circuits, the most well-known of which is the $\llbracket 144, 12, 12 \rrbracket$ bivariate bicycle code~\cite{BB}.
These proposals also achieve a threshold, and have a much higher coding efficiency than the surface code, which comes at the cost of having weight-6 checks that are less geometrically local, as well as being relatively inflexible to scaling.
In particular, smaller-scale near-term devices which cannot support several hundred physical qubits will be unable to realize QEC based on the $\llbracket 144, 12, 12 \rrbracket$ code.

\subsection{Contributions}

\begin{figure}
    \centering
    \includegraphics[width=0.75\linewidth]{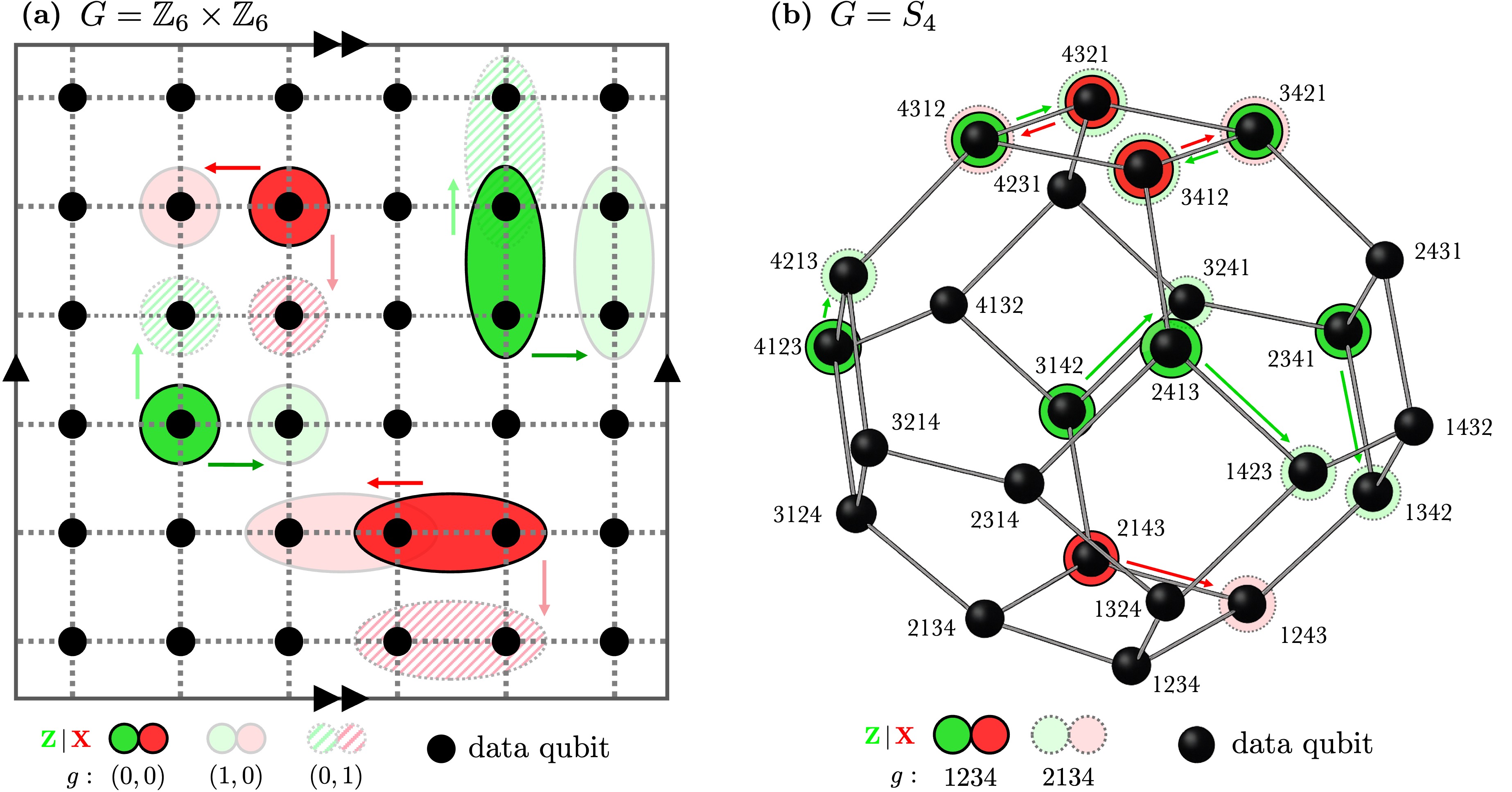}
    \caption{
    Visualization of two mirror codes. \textbf{(a)} An abelian mirror code with $G = \Z_6 \times Z_6$, i.e. a $6 \times 6$ square lattice with periodic boundary conditions.
    Each black point is a data qubit.
    Qubits highlighted in solid green (red) comprise $A=\{(1,2),(4,3),(4,4)\}$ ($B=\{(2,4),(3,1),(4,1)\}$), and a stabilizer is given by $\vec S((0,0)) = \vec Z(A) \vec X(B)$.
    Two other stabilizers are shown, one in light green/red corresponding to $\vec S((1,0)) = \vec Z(A + (1,0)) \vec X(B -(1,0))$, and one in striped green/red corresponding to $\vec S((0,1)) = \vec Z(A + (0,1)) \vec X(B - (0,1))$.
    In general, there are $n$ stabilizers, one for each element in $G$, and $\vec Z$ and $\vec X$ components translate in opposite directions.
    This code has parameters $\llbracket 36, 6, 6 \rrbracket$ with check weight 6.
    \textbf{(b)} An example of a non-abelian construction of a mirror code, with $G = S_4$ the group of permutations on $4$ elements, visualized on its Cayley graph which forms a permutohedron.
    Once again, $A, B$ are highlighted in solid green/red, forming a stabilizer $\vec S(e) = \vec Z(A) \vec X(B)$.
    Another stabilizer, shown in light green/red, is given by $\vec S((12)) = \vec Z(A (12)) \vec X(B(12))$.
    This purpose of this non-abelian construction is largely illustrative, as the code shown, while well-defined, has logical dimension 0.
    }
    \label{fig:mirror_main_figure}
\end{figure}

In this work, we propose a QEC protocol which is more flexible both in terms of code construction and syndrome extraction.
To the former, almost all prior code constructions for near-term quantum memory are not only LDPC, but are specifically CSS codes---stabilizer codes which are ``maximally classical'' in the sense that they can be expressed as the ``product'' of two classical linear codes.
While we continue to construct only LDPC codes, we argue that new state-of-the-art codes---especially in the $n \leq 144$ range---can be achieved by studying stabilizer codes outside the CSS regime.
This argument is supported in part by the history of quantum error correction.
For example, the most compact code that corrects one error is not a CSS code~\cite{laflamme1996perfect}.

Our code construction, like most quantum LDPC codes, do not generally have geometrically local checks.
We therefore envision that these codes would be well-suited for hardware architectures based on neutral atoms and trapped ions, as these offer more general qubit connectivity than superconducting-type architectures.
In turn, these offer their own advantages, with most neutral atom platforms being able to perform many multi-qubit gate simultaneously and trapped ions generally having more reliable two-qubit gates.
At the same time, as we discuss below, our code family includes all bivariate bicycle codes.
Hence, many instances of the codes we introduce are also well-suited for superconducting architectures.

We call our construction \emph{mirror codes}, which produce a quantum stabilizer code from a finite group $G$ and two subsets $A, B \subseteq G$.
The construction of mirror codes is remarkably simple, and readily seen to be well-defined when $G$ is abelian.
\begin{construction}[Mirror codes, informal] \label{def:abelian_mirror_codes}
    Let $G$ be a group with $|G| = n$ and $A, B \subseteq G$.
    A mirror code is defined by $n$ stabilizers labeled by $g \in G$, namely 
    \begin{align}
        \mathbf{S}(g) := \mathbf{Z}(A g) \mathbf{X}(B g^{-1}) .
    \end{align}
\end{construction}
Here, by $Ag$, we mean $\set{ag \,|\, a \in A}$.
The stabilizers are readily seen to commute when $G$ is abelian.
For any two stabilizers $\mathbf{S}(g)$ and $\mathbf{S}(h)$, the anti-commutations occur on the overlaps between $Ag$ and $Bh^{-1}$, as well as between $Ah$ and $B g^{-1}$.
But if $ag = b h^{-1}$ for $a \in A, b \in B$, then $ah = b g^{-1}$ since $G$ is abelian.
Thus, the overlaps may be exactly paired, ensuring that all anti-commutations cancel.
As in the case of bivariate bicycle codes, the rate and distance are not generally known analytically.
However, the check weight is always at most $w := |A| + |B|$.
Mirror codes are manifestly not CSS, as every stabilizer has a mixture of $\mathbf{Z}$ and $\mathbf{X}$ operators.
However, there are important special cases in which the code is equivalent (up to single qubit Clifford operators and qubit permutations) to a CSS code, which we explore in \Cref{sec:mirror_codes}.

In general, mirror codes can also be defined on non-abelian groups $G$, though the choice of $A, B$ can no longer be arbitrary.
We illustrate two examples of mirror codes in \Cref{fig:mirror_main_figure}, one on an abelian group $\Z_6 \times \Z_6$ and one on a non-abelian group $S_4$ the symmetric group on 4 elements.

To the latter---syndrome extraction---we construct three new syndrome extraction circuits which have increasing amounts of provable fault tolerance at the cost of increasing qubit overhead.
Our baseline comparison is the well-known bare syndrome extraction circuit, which uses 1 qubit per check but has no fault tolerance guarantees.
An improved known folklore circuit, known as the loop circuit, enables some partial fault tolerance at the cost of 2 qubits per check.
We construct a circuit which is provably fault-tolerant on weight-6 stabilizer codes, at the cost of 6 qubits per check, as well as a circuit which is provably fault-tolerant on weight-6 CSS codes at the cost of 3 qubits per check.
Finally, we construct a heuristic ``superdense'' circuit which, like the bare circuit, has only 1 qubit per check, but may be more fault tolerant than the bare circuit in some cases.

Using a nearly exhaustive search, we find several novel codes, such as $\llbracket 36, 6, 6 \rrbracket$, $\llbracket 48, 8, 6\rrbracket$, $\llbracket 60, 4, 10 \rrbracket$, and $\llbracket 85, 8, 9 \rrbracket$ codes with check weight 6, all of which are not CSS.
Combining our code construction and fault-tolerant syndrome extraction circuits, we perform end-to-end numerical quantum memory experiments under standard circuit-level noise using a general-purpose quantum stabilizer decoder.
We show that several of our codes with excellent parameters achieve a threshold on the order of $0.2\%$ under this model, approximately matching the performance of the $\llbracket 144, 12, 12 \rrbracket$ code under the same model, which considers circuit-level noise but does not optimize syndrome extraction circuits for their circuit distance, only their depth.
Moreover, we find that the logical error rates drop faster as we use more fault-tolerant syndrome extraction circuits, at the cost of having a lower threshold due to the increased space of possible errors with more ancillary qubits.
As a consequence, we find mirror codes in combination with our various circuits to be a potentially promising path towards near-term fault-tolerant quantum memory, across a range of devices which may vary in physical qubit count.
As devices improve in size and physical error rate, larger mirror codes may be used, as well as more fault-tolerant circuits that incur more ancillary overhead.

The code used for our numerical experiments is \href{https://github.com/jz-lu/mirrorcodes}{freely available}.

\subsection{Related constructions}

\begin{figure}
    \centering
    \includegraphics[width=0.6\linewidth]{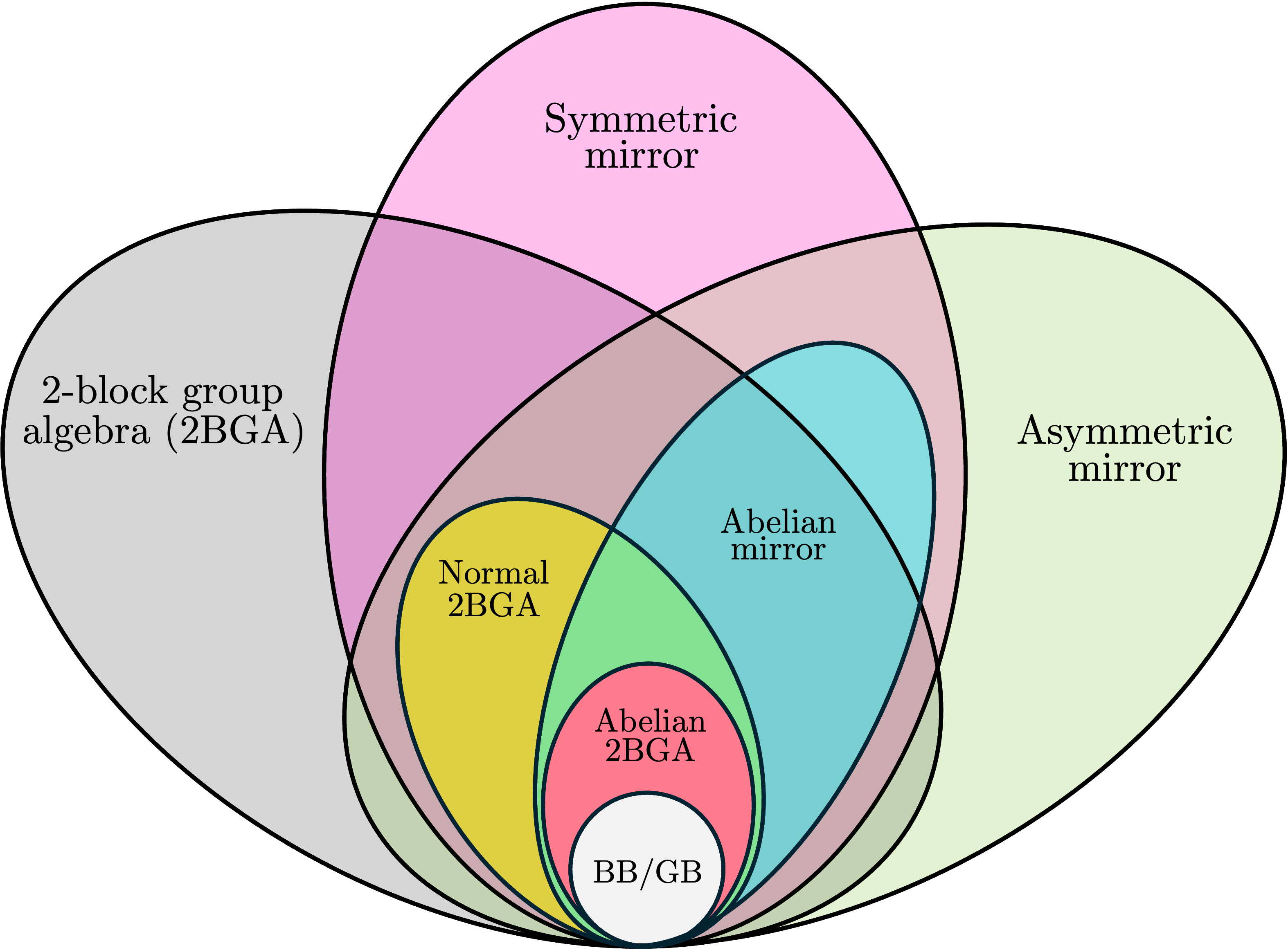}
    \caption{
    Comparison of mirror codes with its closest relatives, including two-block group algebra (2BGA), bivariate bicycle (BB), and generalized bicycle (GB) codes.
    ``Normal'' denotes a construction with group $G$ and subsets $A, B \subseteq G$ which satisfy $g A g^{-1} = A,\; gBg^{-1} = B$ for all $g \in G$; ``abelian'' denotes a construction with an abelian group $G$.
    }
    \label{fig:mirror_classification}
\end{figure}

The closest family of known codes to mirror codes are known at two-block group algebra (2BGA) codes, a family of quantum LDPC CSS codes which are also constructed from a group $G$ and two subsets $A, B \subseteq G$~\cite{2BGA}.
These codes contain all bivariate-bicycle and generalized-bicycle codes.
In the most general case, mirror codes have two forms: symmetric and asymmetric versions, which are not generally the same but are identical when $G$ is abelian. 
We show that mirror codes contain all normal 2BGA codes, a slight generalization of abelian 2BGA codes.
However, in general, there exists non-abelian 2BGA codes which are not mirror codes, and vice versa.
We illustrate the relation of our codes with 2BGA codes and its children in \Cref{fig:mirror_classification}
Note that our containments are up to qubit permutation and local Clifford operations; without this extra freedom, mirror codes would have zero intersection with 2BGA codes.

The remainder of this paper is organized as follows.
In \Cref{sec:preliminaries}, we recall notation and some preliminary facts which will be useful in our construction.
We then construct mirror codes and study them analytically in \Cref{sec:mirror_codes}; this includes classifying the symmetries of mirror codes, determining when mirror codes can be transformed into CSS codes, determining conditions under which mirror codes have poor properties, and formally relating mirror codes to other known code families.
Next, in \Cref{sec:FT_SEC}, we construct our syndrome extraction circuits and prove their fault tolerance.
We then discuss the details of our numerical code search and end-to-end quantum memory experiments in \Cref{sec:code_search_benchmarking}.
Finally, we conclude with some open questions in \Cref{sec:outlook}.

\section{Preliminaries} \label{sec:preliminaries}
We here discuss the requisite notation and background on elementary group theory, quantum stabilizer codes, and practical circuit-level benchmarking.

\subsection{Groups}
Generally, when we refer to a group $G$, we will use standard multiplication notation to refer to the group operation.
However, when $G$ is abelian, we may instead use $+$.
Likewise, we generically refer to the identity element as $e$, but in the abelian case we may instead refer to the identity as $0$.
We define the \emph{commutator} of group elements as $[g, h] := ghg^{-1}h^{-1}$, so that $gh = [g, h] hg$ and $[g, h] = e$ if and only if $g$ and $h$ commute.
Given a group $G$, the \emph{center} of $G$, denoted $Z(G)$, is the set of elements in $G$ which commute with all elements of $G$; that is, \begin{align}
    Z(G) := \set{h \in G \,|\, [g, h] = e} .
\end{align}
Note that $Z(G)$ is always a normal subgroup of $G$.

Given a finite group $G$, $|G|$ denotes the order of $G$, i.e. the number of elements in the group.
Given a subgroup $H \leq G$, we say that $H$ has \emph{index} $t = [G : H]$ if there are $t$ left cosets of $H$ in $G$.
By Lagrange's theorem, $|G| = t |H|$.
We recall that for $h \in G$, the \emph{conjugacy class} of $h$ is given by $\set{ghg^{-1} \,|\, g \in G}$.

Given a subset $A \subseteq G$ of group elements and some $g \in G$, we define set-element multiplication element-wise.
That is, \begin{align}
    Ag := \set{ag \,|\, a \in A} ,
\end{align}
and analogously for $gA$.
We say a subset $A$ is \emph{normal} if $Ag = gA$ for all $g \in G$.
When $A$ is a subgroup of $G$, this definition coincides with the definition of normal subgroups.
However, even if a group $G$ has no normal subgroups other than $\set{e}$ and $G$ itself, it will typically have many normal subsets.
In particular, a subset of $G$ is normal if and only if it is a union of conjugacy classes of elements in $G$.

The fundamental theorem of finite abelian groups gives a convenient representation of any finite abelian group, which we will use to give simple parameterizations of abelian mirror codes.

\begin{theorem}[Fundamental theorem of finite abelian groups] \label{thm:fundamental_abelian_gps}
    Let $\CG$ be a finite abelian group of order at least 2. Then 
    \begin{align}
        \CG = \bigoplus_{i=1}^m \Z_{q_i} ,
    \end{align}
    where the $q_i$ are (not necessarily distinct) prime powers $p_i^{a_i}$ which can without loss of generality be ordered lexicographically $(p_1, a_1) \leq \cdots \leq (p_m, a_m)$.
\end{theorem}
Therefore, when we refer to a finite abelian group of order $n$, we will associate it uniquely with a tuple $(q_1, \dots, q_m)$, such that $\prod_{i=1}^m q_i = n$.

\subsection{Quantum stabilizer and CSS codes}
We denote the Pauli operators as \begin{align}
    \mathbf{I} := \begin{pmatrix}
        1 & 0 \\ 0 & 1
    \end{pmatrix} , \; \mathbf{Z} := \begin{pmatrix}
        1 & 0 \\ 0 & -1
    \end{pmatrix} ,\;
    \mathbf{X} := \begin{pmatrix}
        0 & 1 \\ 1 & 0
    \end{pmatrix} ,\;
    \mathbf{Y} := \begin{pmatrix}
        0 & -i \\ i & 0
    \end{pmatrix} .
\end{align}
The $n$-qubit Pauli group $\CP_n$ is the set of $n$-fold tensor products of Paulis, i.e. \begin{align}
    \CP_n := \set{\alpha \mathbf{P}_1 \otimes \cdots \otimes \mathbf{P}_n \,|\, \alpha \in \set{\pm 1, \pm i}, \, \mathbf{P}_i \in \set{\mathbf{I}, \mathbf{Z}, \mathbf{X}, \mathbf{Y}}} .
\end{align}

Among the group of $d \times d$ unitary matrices $\mathrm{U}(d)$, the Clifford group is a subgroup which maps any Pauli to any other Pauli via conjugation.
That is, \begin{align}
    \Cliff_n := \set{\mathbf{C} \in \mathrm{U}(2^n) \,|\,\forall \mathbf{P} \in \CP_n, \, \mathbf{C P C}^\dag \in \CP_n} .
\end{align}
We refer to $\Cliff_1$ as single-qubit Cliffords, and operators of the form $\bigotimes_{i=1}^n \vec C_i$ where $\vec C_i \in \Cliff_1$ as \emph{local Cliffords}.
In general, $\Cliff_n$ is generated by the Hadamard, phase, and controlled-NOT gates, given by \begin{align}
    \mathbf{H} := \frac{1}{\sqrt{2}} \begin{pmatrix}
        1 & 1 \\ 1 & -1
    \end{pmatrix} ,\; \mathbf{S} := \begin{pmatrix}
        1 & 0 \\ 0 & i 
    \end{pmatrix} , 
\end{align}
and $\mathbf{CX}_{i \to j} := \ketbra{0}{0}_i \otimes \mathbf{I}_j + \ketbra{1}{1}_i \otimes \mathbf{X}_j$.
Local Cliffords are generated solely by $\mathbf{H}$ and $\mathbf{S}$.

A \emph{stabilizer subgroup} $\CS \leq \CP_n$ is an abelian group of $\CP_n$ which does not contain $-\mathbf{I}$.
Any such subgroup can be specified by a set of generators which, for purposes of this work, we will not insist be independent.
Such a set of generators is known as the stabilizer tableau $S = \set{\mathbf{S}_1, \dots, \mathbf{S}_r}$.
A stabilizer tableau defines a subspace of quantum states by the set of states which are fixed points of every operator in the subgroup.
In other words, we define the \emph{stabilizer code} associated with $S$ as $C(S) := \set{\ket{\psi} \,|\,\forall i, \, \mathbf{S}_i \ket{\psi} = \ket{\psi} }$.
The logical dimension of the code is given by $k := \log \dim C(S)$, and the rate is $R := k/n$.
The \emph{check weight} of the tableau is the maximum weight over the $\mathbf{S}_i$, where the weight of a $n$-qubit Pauli is the maximum number of non-identity Paulis in the tensor product.
We denote this operation $\wt(\mathbf{P})$.

Given a stabilizer subgroup $\CS \leq \CP_n$, the \emph{logical subspace} $\CL(\CS) := C_{\CP_n}(\CS)$ is the centralizer $C_{\CP_n}(\CS)$ of $\CS$.
That is, the logical subspace is the set of Paulis which commute with all stabilizers.
By definition, the logical subspace includes all stabilizers.
Such operators act on the logical qubits encoded by the stabilizer code.
The \emph{distance} of the stabilizer code is the minimum weight element of the logical subspace \emph{excluding} the stabilizer subgroup, i.e. \begin{align}
    d := \min_{\mathbf{P} \in \CL(\CS) \setminus \CS} \wt(\mathbf{P}) .
\end{align}
We refer to a $n$-qubit code with logical dimension $k$, distance $d$, and check weight $w$ as a $\llbracket n, k, d, w \rrbracket$ code.
On occasion, we may omit the weight, and refer to the code as a $\llbracket n, k, d \rrbracket$ code.

Let $\mathbf{O}$ be a $1$-qubit operator.
For a $m$-dimensional binary vector $\mathbf{v}$, denote $\mathbf{O}^{\mathbf{v}} := \bigotimes_{i=1}^m \mathbf{O}^{v_i}$, where $\mathbf{O}^{0} := \mathbf{I}$.
In the case every $\mathbf{S}_i$ is either of the form $\mathbf{Z}^{\mathbf{v}}$ or $\mathbf{X}^{\mathbf{v}}$, we say that the stabilizer tableau is \emph{CSS}.
In many cases, the tableau may not be CSS, but there will be a simple local transformation to make the code CSS.
Therefore, we say that a code is \emph{equivalently CSS} if there exists $\mathbf{C} = \bigotimes_{i=1}^n \mathbf{C}_i$, where $\mathbf{C}_i \in \Cliff_1$, such that $\set{\mathbf{C} \mathbf{S}_1 \mathbf{C}^\dag, \,\dots,\, \mathbf{C} \mathbf{S}_r \mathbf{C}^\dag}$ is CSS.
We further say that a code is \emph{equivalently CSS via Hadamards} if each $\mathbf{C}_i \in \set{\mathbf{I}, \,\mathbf{H}}$.
More generally, we say that two stabilizer codes are permutation-equivalent if there exists a permutation $\pi$ such that one code can be transformed into the other by permuting the qubits via $\pi$.
Two codes are LC-equivalent if there exists a local Clifford $\mathbf{C}$ such that applying $\mathbf{C}$ by conjugation onto the stabilizers yields new stabilizers which equal the original stabilizers up to phases.
Note that equivalent codes have the same check weight, rate, and distance.

Let $A \subseteq [n] := \set{1,\, \dots,\, n}$ be a set of qubit indices, and let $\mathbf{1}_{A} \in \F_2^n$ be the indicator vector of presence in $A$.
Then we define $\mathbf{O}(A) := \mathbf{O}^{\mathbf{1}_A}$, i.e. take the tensor product of $\mathbf{O}$ on each qubit in $A$.

\subsection{Circuit-level noise models}

The distance of a code is the smallest number of Pauli errors that can occur on the data qubits to perform a logical operator.
While the distance is an excellent simple proxy for a code's practical performance, it does not fully capture the landscape of noise in a memory experiment. 
This is because to measure syndromes, one must implement a \textit{syndrome-extraction circuit}, which iteratively applies elements of a small gate set between data qubits and ancillary qubits to measure stabilizers.
Not only is each gate application inherently noisy, but qubits which idle while we compute on other qubits experience additional noise as well.
In general, the \emph{circuit noise model} represents the set of all possible faults that can occur during the operation of a circuit.
Every qubit experiences noise at every moment in time, which we model as occurring with a probability proportional to some global physical error rate $p$.
The actual probabilities of failure will depend on the hardware architecture being used, but generally these are modeled as being constant multiples of $p$, depending on the type of operation being applied.
Specifically, qubits can experience noise through all of the following:
\begin{itemize}
    \item A single-qubit Clifford acting on a qubit might be followed by a random Pauli. 
    We denote this \textit{single-qubit depolarizing noise} by $p_1$.
    \item A two-qubit Clifford gate might be followed by a random 2-qubit Pauli. 
    We denote this two-qubit depolarizing noise by $p_2$.
    \item A measurement might report the wrong outcome. 
    We denote this measurement error by $p_{\text{meas}}$.
    \item The preparation of a quantum basis state, e.g. $\ket{0}$ or $\ket{1}$, might prepare the orthogonal basis state. 
    This is initialization error, $p_{\text{init}}$.
    \item A qubit that is not experiencing any of the above might experience a random Pauli. This is idling noise denoted by $p_\text{idle}$.
    \item Some noise models, such as \texttt{SI1000}~\cite{gidney2022benchmarking}, include \textit{resonator idling noise}, where idling qubits experience more depolarizing noise---with probability $p_\text{res}$---if any other qubits are being measured or reset in that timestep.
\end{itemize}

We can refer to the \textit{circuit distance} of a code---defined now through a particular choice of syndrome extraction circuit---as the smallest number of elementary noise events (those listed above) that must occur to apply an undetectable logical operation.
Here \textit{undetectable} means that no extra measurements, such as those corresponding to stabilizers or additional flags in the syndrome extraction circuit (see below), report a non-zero syndrome.
Even if a code has large distance, its circuit distance may be small.
This is because of error propagation: a few elementary faults in the syndrome extraction circuit may be propagated via multi-qubit gates to a much larger number of total errors.
To avoid this issue, we can add \textit{flags} to our syndrome extraction circuit, which can detect various errors on the ancilla qubits during syndrome extraction rounds.
When a syndrome extraction circuit is fully fault-tolerant, the circuit distance is the same as the distance.
Note that for a sufficiently small physical error rate $p$ and a sufficiently good syndrome extraction circuit, it is useful to repeatedly run the syndrome extraction circuit multiple times, so that faults in the stabilizer measurement process can be detected and corrected in the decoding process.
We typically run about $d$ rounds of syndrome extraction, where $d$ is the distance of the code, in order to correct extraction-based errors.
This technique trades efficiency for accuracy, especially when $d$ is large.

We say that a family of codes indexed by physical dimension $n$ has a \textit{threshold} $p_{\text{th}}$ if for $p<p_{\text{th}}$, larger codes in the family will have exponentially (in $n$) lower logical error rates, $p_L$.
Intuitively, when the physical error rate is below $p_{\text{th}}$, the code family corrects errors more frequently than it introduces them.
The exact value of $p_{\text{th}}$ depends on the code family, syndrome extraction circuit, noise model, and decoder.
Thus, the numerical value of a threshold is not a property purely of the code family.

If we instead wish to benchmark a single code, which need not be part of any infinite family, we instead use the \textit{pseudothreshold} $p_{\text{pth}}$---the physical error rate parameter $p$ at which the logical error rate is equal to the physical error rate $p_L=p$.
When the physical error rate is below $p_{\text{pth}}$, the logical error rate is less than the physical error rate.
To give normalized comparisons between physical error rates $p$ and logical error rates $p_L$, we compute logical error rates by measuring the fraction of logical errors out of all experiments, per logical qubit, per round of syndrome extraction.

\section{Mirror Codes} \label{sec:mirror_codes}

In \Cref{sec:introduction}, we defined the simplest case of mirror codes---when the group $G$ is abelian.
We here give the more general definition for an arbitrary, not necessarily abelian group.
In this case, the order of multiplication plays a significant role, to the extent to which the exact choice of multiplication order defines distinct families of codes.

\begin{construction}[Mirror codes]
    Let $G$ be a finite group with $|G| = n$.
    Let $A, B \subseteq G$, and define $w_Z := |A|$, $w_X := |B|$, and $w := w_Z + w_X$.
    Associate to each group element $g \in G$ a physical qubit.
    \begin{itemize}
        \item A $(G, A, B)$ \emph{symmetric} mirror code, when well-defined, is given by the $n$ stabilizers \begin{align}
            \mathbf{S}(g) := \mathbf{Z}(A g) \mathbf{X}(B g^{-1}) .
        \end{align}

        \item A $(G, A, B)$ \emph{asymmetric} mirror code, when well-defined, is given by the $n$ stabilizers \begin{align}
            \mathbf{S}(g) := \mathbf{Z}(A g) \mathbf{X}(g^{-1} B) .
        \end{align}
    \end{itemize}
\end{construction}
Here, we recall that $Ag := \set{ag \,|\, a \in A}$, and $\mathbf{Z}(T) := \prod_{q \in T} \mathbf{Z}_q$.
Note that if the $\mathbf{Z}$ and $\mathbf{X}$ overlap on a qubit, we say that the stabilizers acts with a $\mathbf{Y}$ on that qubit.
This convention can lead to some stabilizers being able the generate minus the identity, $-\mathbf{I}$, which by convention is not a stabilizer subgroup.
Therefore, as we generate each new stabilizer $\mathbf{S}(g)$, we adopt the convention that if adding $\mathbf{S}(g)$ would result in the stabilizers generating $-\mathbf{I}$, we add $-\mathbf{S}(g)$ instead.
We henceforth ignore this stabilizer phase issue.

We first show that this construction defines a proper Low-Density Parity Check (LDPC) code.
\begin{proposition}[Mirror codes are LDPC]
    For any $w$, $(G, A, B)$ mirror codes with $|A|+|B|\leq w$ are $w$-LDPC, meaning every stabilizer has weight $\leq w$ and every qubit is in the support of at most $w$ stabilizers.
\end{proposition}
\begin{proof}
    By the definition of the stabilizers, $\mathbf{S}(g)$ will contain at most $|A|$ $\mathbf{Z}$ terms and at most $|B| \mathbf{X}$ terms.
    The weight of $\mathbf{S}(g)$ is thus at most $w$, and might be lower if the $\mathbf{Z}$ and $\mathbf{X}$ terms overlap to combine into a $\mathbf{Y}$.
    Since, for any $a$, $aG=Ga=G$, any term inside $A$ or $B$ will map to every qubit in $G$ exactly once in terms such as $Ag$, $Bg^{-1}$, and $g^{-1}B$.
    Thus, the number of different stabilizers that a qubit can be in the support of is at most $|A|+|B|=w$, proving that mirror codes are LDPC.
\end{proof}

We emphasize that although we define a mirror code with $n$ stabilizers, some stabilizers may be linearly dependent on others, and hence the rate of the corresponding code is still positive for well-chosen $(G, A, B)$.
Unlike the abelian case, however, not all choices of $A, B$ even yield well-defined mirror codes.
We next give an exact characterization of well-defined mirror codes.

\begin{proposition}[Well-defined mirror code characterization] \label{prop:mirror_code_well_defined}
For a group $G$ and subset $A, B \subseteq G$, define the mod-2 kernel functions $K_S, K_A \,:\, G \times G \to \F_2$ by \begin{align}
    K_S(g, h) := |Agh \cap B| ,\quad K_A(g, h) := |g A h \cap B| .
\end{align}
We say that a kernel $K$ is symmetric if $K(g, h) = K(h, g)$ for all $g, h \in G$.
Then a $(G, A, B)$ symmetric mirror code is well-defined if and only if $K_S$ is symmetric.
Likewise, a $(G, A, B)$ asymmetric mirror code is well-defined if and only if $K_A$ is symmetric.
\end{proposition}

\begin{proof}
Consider first the symmetric case.
Two stabilizers $\mathbf{S}(g)$ and $\mathbf{S}(h)$ commute if and only if there are an even number of anticommutations.
The total number of anticommutations is given by $|Ag \cap Bh^{-1}| + |Ah \cap B g^{-1}|$.
Thus, they commute iff $\forall g, h \in G$, \begin{align}
    |Ag \cap Bh^{-1}| = |Ah \cap B g^{-1}| \pmod{2} .
\end{align}
Note that $|Ag \cap Bh^{-1}| = |Agh \cap B|$.
Hence, the above condition is precisely equivalent to $K_S(g, h) = K_S(h, g)$.
In the asymmetric case, the number of anticommutations is given by $|Ag \cap h^{-1} B| + |Ah \cap g^{-1} B|$.
\end{proof}

A priori, there are in addition two other conceivable constructions of mirror codes which differ from the above constructions by order of operation.
We next show, however, that these alternative constructions are in reality equivalent to symmetric or asymmetric mirror codes.

\begin{lemma}[Equivalence of alternative constructions] \label{lemma:equivalence_alt_constructions}
Let $G$ be a group and $A, B \subseteq G$.
Define the stabilizer tableaux \begin{align}
    \set{\mathbf{S}'(g) := \mathbf{Z}(gA) \mathbf{X}(g^{-1} B)}_{g \in G} ,\quad \set{\mathbf{S}''(g) := \mathbf{Z}(gA) \mathbf{X}(B g^{-1})} .
\end{align}
The former tableau corresponds to a code equivalent to a $(G, A^{-1}, B^{-1})$ symmetric mirror code by qubit permutation, and the latter tableau corresponds to a code equivalent to a $(G, B, A)$ asymmetric mirror code by a combination of Hadamard transform and phases.
In each case, one code is well-defined if and only if the corresponding code is well-defined.
\end{lemma}

\begin{proof}
Consider a $(G, A^{-1}, B^{-1})$ symmetric mirror code where each qubit $q$ is labeled by a group element $g$; that is, $q(g) = g$.
We define a permutation in which each qubit is re-labeled by the inverse element, i.e. $q(g) = g^{-1}$.
The stabilizers under this permutation become \begin{align}
    \mathbf{S}_\pi(g) = \mathbf{Z}((A^{-1} g)^{-1}) \mathbf{X}((B^{-1} g^{-1})^{-1}) = \mathbf{Z}(g^{-1} A) \mathbf{X}(g B) = \mathbf{S}'(g^{-1}) .
\end{align}

Thus, the set of stabilizers of the $(G, A^{-1}, B^{-1})$ symmetric mirror code permuted by $\pi$ are precisely the set of stabilizers given by $\mathbf{S}'$, and hence the codes are equivalent.
Next, note that a $(G, B, A)$ asymmetric mirror code has stabilizers $\mathbf{S}(g) = \mathbf{Z}(B g) \mathbf{X}(g^{-1} A)$, which under conjugation by a Hadamard transform gives \begin{align}
    \mathbf{S}_{\mathbf{H}}(g) & := \mathbf{H}^{\otimes n} \mathbf{S}(g) \mathbf{H}^{\otimes n} = \mathbf{X}(B g) \mathbf{Z}(g^{-1} A) \\
    & = (-1)^{|Bg \,\cap\, g^{-1} A|} \mathbf{Z}(g^{-1} A) \mathbf{X}(B g) = (-1)^{|Bg \,\cap\, g^{-1} A|} \mathbf{S}''(g^{-1}) .
\end{align}
The phase difference is not necessarily 0.
However, the choice of phase on each stabilizer does not affect its error-correcting properties.
In both cases, the equivalence map is executed by operations which leave commutators invariant, and hence one code is well-defined if and only if the corresponding code is well-defined.
\end{proof}
On the other hand, symmetric and asymmetric mirror codes are genuinely different code families, which we prove in \Cref{app:counterexamples_equivalences}.

There is a simple equivalent characterization for commutation in symmetric mirror codes in terms of a certain element of the group algebra $\F_2[G]$ being central.
Similar language is often the convention within the context of 2BGA codes; we give this equivalent characterization in \Cref{app:group_algebra_symmetric_mirror}.
Interestingly, this group algebra formulation does not appear to extend to asymmetric mirror codes.
We make a few initial remarks about these two definitions of mirror codes.
First, note that if $G$ is abelian, then the symmetric and asymmetric constructions are equivalent, and every choice of $A, B$ gives a well-defined mirror code.
In fact, even if $G$ is not abelian, it suffices for $A, B \subseteq Z(G)$ (i.e. the subsets are in the center of $G$) for the two definitions to be equal, and for every choice of $A, B \subseteq Z(G)$ to yield a well-defined mirror code.

Although it is not obvious as to how to produce $A, B \subseteq G$ in general such that $(G, A, B)$ is a well-defined (a)symmetric mirror code based on the previous characterization, there are several simpler conditions which are \emph{sufficient} to ensure stabilizer commutativity.

\begin{lemma}[Center-based sufficient conditions for valid mirror codes] \label{lemma:center_sufficient_AB_mirror_codes}
    Let $G$ be a group and $A, B \subseteq G$. 
    Denote by $A^{-1} := \set{a^{-1} \,|\, a \in A}$ and $AB := \set{ab \,|\, a \in A, b \in B}$.
    \begin{itemize}
        \item If $A^{-1} B \subseteq Z(G)$, then $(G, A, B)$ forms a symmetric mirror code.
        \item If $A B \subseteq Z(G)$, then $(G, A, B)$ forms an asymmetric mirror code.
    \end{itemize}
\end{lemma}

\begin{proof}
Fix two stabilizers $\mathbf{S}(g)$ and $\mathbf{S}(h)$ such that $g \neq h$.
Consider first the symmetric case.
We will show that the given condition implies that the two anticommuting overlap sets $Ag \cap Bh^{-1}$ and $Ah \cap Bg^{-1}$ are in one-to-one correspondence.
That is, the two sets have the same size, which is stronger than having the same parity.
Let $v \in Ag \cap Bh^{-1}$, so that for some $a \in A, b \in B$, $v = ag = bh^{-1}$.
Then $h = g^{-1} a^{-1} b$.
Since $A^{-1} B \in Z(G)$, $h = a^{-1} b g^{-1}$.
Equivalently, $ah = b g^{-1}$, so there is a unique corresponding element $v' \in Ah \cap B g^{-1}$.
This correspondence gives the desired bijection.
In the asymmetric case, the proof proceeds similarly, where we begin with $v = ag = h^{-1} b$, so that by centrality of $AB$,
\begin{align}
    h = b g^{-1} a^{-1} = a^{-1} (ab) g^{-1} a^{-1} = a^{-1} g^{-1} a^{-1} (ab) = a^{-1} g^{-1} b
\end{align}
Rearranging, $ah = g^{-1} b$ as desired.
\end{proof}
We remark that for any $A, B \subseteq G$, if $AB \subseteq Z(G)$, then $[a, b] = 0$ for all $a \in A$ and $b \in B$.
This is because centrality of $AB$ implies that $gab = abg$ for all $g \in G$.
Choosing $g = a^{-1}$, $b = aba^{-1}$, so $ba = ab$.

\begin{lemma}[Normality-based sufficient conditions for well-defined mirror codes] \label{lemma:normality_sufficient_AB_mirror_codes}
    Let $G$ be a group and $A, B \subseteq G$.
    If $A$ and $B$ are both normal subsets of $G$, then $(G, A, B)$ forms a valid mirror code (and the symmetric and asymmetric formulations coincide).
\end{lemma}

\begin{proof}
    If $B$ is a normal subset then $g^{-1} B = B g^{-1}$ and therefore the symmetric and asymmetric formulations are identical, and we henceforth consider the symmetric formulation only.
    We note that $A$ is a normal subset if and only if $A^{-1}$ is: by normality, there is a permutation $\pi_g \,:\, A \to A$ such that
    $g a = \pi_g(a) g$.
    Hence, $g = \pi_g(a) g a^{-1}$ and $\pi_g(a)^{-1} g = g a^{-1}$.
    That is, the permutation $\cdot^{-1} \circ \pi_g \circ \cdot^{-1}$ witnesses the normality of $A^{-1}$.
    Fix stabilizers $\mathbf{S}(g)$ and $\mathbf{S}(h)$; we again give a bijection between $Ag \cap Bh^{-1}$ and $Ah \cap Bg^{-1}$.
    Let $v \in Ag \cap Bh^{-1}$, so $v = ag = b h^{-1}$ for some $a \in A, b \in B$.
    Then $h = g^{-1} a^{-1} b$.
    By normality, there exists permutations $\pi_g \,:\, A \to A$ and $\s_g \,:\, B \to B$ such that \begin{align}
        h = \pi_g(a^{-1}) \sigma_g(b) g^{-1} 
    \end{align}
    and therefore $\pi_g(a^{-1})^{-1} h = \sigma_g(b) g^{-1}$.
    Thus, the maps $(\pi_g, \s_g)$ give an injection from $Ag \cap Bh^{-1}$ to $Ah \cap Bg^{-1}$.
    An analogous injection may be derived in the reverse direction to complete the proof.
\end{proof}

\subsection{Gauge symmetries}

While any choice of $(G, A, B)$ yields a valid mirror code (so long as the stabilizers commute), there are many choices of $(G, A, B)$ which yield essentially the same code.
We say that two codes are equivalent if their stabilizer tableaux can be transformed to each other via very simple operations.
Typically, such operations are restricted to permutations of physical qubits and local Clifford operations, as such maps preserve all of the memory properties of the code.
(We do \emph{not} consider operations that pick a new basis for the stabilizer subgroup, as choosing a good basis plays a significant role in the code's practical performance and is often computationally hard, e.g. finding a short basis.)
Informally, a gauge symmetry of a mirror code is a map from $(G, A, B)$ to some $(G', A', B')$, such that the two parameterizations yield the same mirror code stabilizers up to the aforementioned simple operations.
We here give a classification of these gauge symmetries of mirror codes.

The set of gauge symmetries differ between symmetric and asymmetric mirror codes.
For example, consider the qubit permutation $q(h) = gh$ for some fixed $g \in G$.
The symmetric and asymmetric stabilizers $\mathbf{S}(h)$ respectively map to \begin{align}
    \mathbf{Z}(gAh) \mathbf{X}(g B h^{-1}) ,\quad \mathbf{Z}(g Ah) \mathbf{X}(g h^{-1} B) .
\end{align}
In the symmetric case, this permutation corresponds to $(G, A, B) \to (G, gA, gB)$, whereas there is no analogous correspondence in the asymmetric case because $g$ may not commute with $h^{-1}$.
Likewise, we may always re-label the stabilizers $\mathbf{S}(h) \mapsto \mathbf{S}(gh)$.
This transformation maps the stabilizers in each case to \begin{align}
    \mathbf{Z}(Agh) \mathbf{X}(Bh^{-1} g^{-1}) ,\quad \mathbf{Z}(Agh) \mathbf{X}(h^{-1} g^{-1} B) ,
\end{align}
respectively for symmetric and asymmetric codes.
In the asymmetric case, this transformation proves that $(G, A, B)$ and $(G, Ag, g^{-1} B)$ produce equivalent codes for any $g \in G$.
However, the symmetric case has no corresponding equivalence, since $g$ may not commute with $h$. 

\begin{definition}[Permutation gauge symmetries] \label{def:permutation_gauge_symmetries}
    A map $f \,:\, (G, A, B) \mapsto (G', A', B')$ is a permutation gauge symmetry if for all $G$ there exists a permutation $\pi$ of physical qubits such that for all valid $A, B \subseteq G$ (a)symmetric mirror codes, the $f(G, A, B) = (G', A', B')$ (a)symmetric mirror code is equivalent to the $(G, A, B)$ (a)symmetric mirror code with qubits permuted by $\pi$.
    That is, the two stabilizer tableaux are identical as sets.
\end{definition}

Before we classify permutation gauge symmetries, we prove an instrumental lemma about group actions, which may be of independent interest.
In what follows, let $\operatorname{Orb}_Y(A) := \set{y(A) \,|\, y \in Y}$ be the \emph{orbit} of $A$ under the action of $Y$.
Roughly speaking, this lemma shows that the only group action $Y \,:\, \widetilde{G} \times G \to G$ which has the same orbit as the action of right-multiplication $R \,:\, G \times G \to G$, $R(g, x) = xg$ on every \emph{subset} of a group $G$ is $R$ itself.
In this sense the orbit set exactly pins down the underlying group action.

\begin{lemma}[Orbit pinning with right multipliers] \label{lemma:two_set_three_set}
    Let $G, \widetilde{G}$ be finite groups.
    Define the right-multiplier $R_g \,:\, G \to G$ by $R_g(x) = xg$ and the corresponding group action $R \,:\, G \times G \to G$ by $R(g, x) = R_g(x)$.
    Let $Y \,:\, \widetilde{G} \times G \to G$ be any other group action with the property that for all $C \subseteq G$, \begin{align} \label{eq:two_set_three_set_assumption}
        \operatorname{Orb}_Y(C) = \operatorname{Orb}_R(C) .
    \end{align}
    Then $Y = R$.
    The same claim holds if instead $R$ is replaced by the group action of left multiplication $L(g, x) = gx$.
\end{lemma}

\begin{proof}
Let $y_g(x) := Y(g, x)$.
We first study the action of $Y$ on two-sets in $G$ of the form $\set{x, sx}$.
(Such two-sets may be interpreted as undirected edges in the Cayley graph of $G$.)
Namely, let \begin{align}
    E_{s} := \operatorname{Orb}_{R}(\set{e, s}) = \set{\set{x, sx} \,|\, x \in G} .
\end{align}
We claim that for any $y_g \in Y$, \begin{align} \label{eq:preserving_two-sets}
    y_g(E_s) = E_s .
\end{align}
The proof follows by the transitivity property of operators in $Y$.
In particular, since $E_{s} = \operatorname{Orb}_{R}(\set{e, s})$ and $\operatorname{Orb}_{R}(\set{e, s}) = \operatorname{Orb}_{Y}(\set{e, s})$, for any $\mathcal{E} \in E_s$ there exists some $y_{g^*} \in Y$ such that $y_{g^*}(\set{e, s}) = \CE$.
Hence, \begin{align}
    y_g(\CE) = y_g(y_{g^*}(\set{e, s})) \in \operatorname{Orb}_{Y}(\set{e, s}) = E_s .
\end{align}
Thus, $y_g(E_s) \subseteq E_s$.
The same holds if we replace $y_g$ with $y_g^{-1} = y_{g^{-1}}$ in the above equation, so that $y_g^{-1}(E_s) \subseteq E_s$.
These two results imply Eqn.~\eqref{eq:preserving_two-sets}.

We next study three-sets of the form $\set{x, rx, sx}$, i.e. \begin{align}
    T_{rs} := \operatorname{Orb}_{R'}(\set{e, r, s}) = \set{\set{x, rx, sx} \,|\, x \in G'} .
\end{align}
Each such three-set is an undirected triangle in the Cayley graph of $G$.
For the moment, fix $r, s \in G$ and $t_x := \set{x, rx, sx} \in T_{rs}$.
Assume that $r, s$ are such that $E_r, E_s, E_{s r^{-1}}$ are all distinct.
Since each $E_z$ is an orbit, this implies that $E_r, E_s, E_{s r^{-1}}$ are pairwise disjoint.
Under these circumstances, we may specify a triangle $t_x$ uniquely with three edge two-sets, i.e. $\set{x, rx}, \set{x, sx}, \set{rx, sx}$.
Each point is contained in exactly two edges.
Moreover, $\set{x, rx} \in E_r$, $\set{x, sx} \in E_s$, and $\set{rx, sx} \in E_{s r^{-1}}$.
In this undirected graph specification, $x$ is the unique point contained in an edge in $E_r$ and an edge in $E_s$.
By Eqn.~\eqref{eq:two_set_three_set_assumption}, \begin{align}
    \set{y_g(t_x) \,|\, y_g \in Y} = \operatorname{Orb}_Y(t_x) = \operatorname{Orb}_{R}(t_x) = \set{t_z \,|\, z \in G'}.
\end{align}
Hence, for any $x \in G$ and $y_g \in Y$, $y_g(t_x) \in \set{t_z \,|\, z \in G}$, so $y_g(t_x) = t_z$ for some $z$.
We claim in fact that $y_g(t_x) = t_{y_g(x)}$, i.e. $z = y_g(x)$.
This follows from the uniqueness and preservation properties derived above.
More precisely, within $t_x$, $x$ is the unique point contained in an edge in $E_r$ as well as $E_s$; the same holds true for $z$ in $t_z$.
Moreover, $y_g(E_w) = E_w$ for all $w$ by Eqn.~\eqref{eq:preserving_two-sets}.
Note that $y_g(t_x) = \set{y_g(x), y_g(rx), y_g(sx)}$ and $t_z = \set{z, rz, sz}$.
By Eqn.~\eqref{eq:preserving_two-sets}, $\set{y_g(x), y_g(rx)} = \set{z, rz} \in E_r$ and $\set{y_g(x), y_g(sx)} = \set{z, sz} \in E_s$.
Since $z$ is the only point within $t_z$ in both $E_r$ and $E_s$, we conclude that $z = y_g(x)$.
Similarly, using the fact that $rx$ ($sx$) is the unique point contained in both the $E_r$ ($E_s$) and $E_{s r^{-1}}$ edges of $t_x$, we may conclude that $y_g(rx) = rz$ ($y_g(sx) = sz$).
In sum then, $y_g(rx) = r y_g(x)$ and $y_g(sx) = s y_g(x)$.

We now address the condition we assumed in order to derive these relations, namely that $E_r, E_s, E_{s r^{-1}}$ are distinct two-sets.
This condition is equivalent to the requirements $r \notin \set{e, s, s^{-1}, s^{2}}$ and $s \notin \set{e, r, r^{-1}, r^2}$.
For any non-cyclic group, there are at least two generators $r, s$ which by definition satisfy these relations.
More generally, for any non-cyclic group $G$ we may always take a set of independent generators $r_1, \dots, r_m$ and pair each $r_i$ with some $s_i$ such that $r \notin \langle s_i \rangle$ and $s_i \notin \langle r_i \rangle$ which implies the desired requirement.
Decomposing any $r = u_{i_1} \cdots u_{i_\ell} \in G$ in this manner, we have by induction $y_g(rx) = r y_g(x)$ for all $r \in G$.
We then have $y_g(r) = r y_g(e)$ for all $r \in G$.
That is, $y_g$ is simply a right-multiplication by some element $y_g(e) \in G$ determined by $g$, i.e. $Y = R$.

All of these arguments hold for non-cyclic groups.
We next give a modified argument, requiring only two-sets, which hold for cyclic groups.
Let $G = \langle r \rangle$.
Since $Y$ is a group action, there is a unique $g^* \in G$ such that $y_{g^*}(e) = r$.
Then $y_{g^*}(\set{e, r}) = \set{r, y_g(r)}$, and by Eqn.~\eqref{eq:orbit_equivalence}, $y_g(\set{e, r}) = \set{r^k, r^{k+1}}$ for some integer $k \in [0, n-1]$.
Since $\set{r^k, r^{k+1}}$ must contain $r$, the only choices for this set are $\set{e, r}$ and $\set{r, r^2}$.
Note that for $n \geq 3$, these sets are distinct.
If $\set{r^k, r^{k+1}} = \set{e, r}$, then $y_{g^*}(r) = e$.
This is, however, impossible, because this would imply that $y_{g^*}(\set{e, r}) = \set{e, r}$.
The construction of $Y$ implies that any $y_g$ fixes $\set{e, r}$ if and only if $g = e$ (for $n \geq 3$ where the orbit has $n$ distinct elements), and if $g^* = e$ here then $y_{g^*}(e) \neq r$.
The only remaining possibility is that $\set{r^k, r^{k+1}} = \set{r, r^2}$, i.e. $y_{g^*}(r) = r^2$.
Repeating this argument inductively yields $y_{g^*}(r^k) = r^{k+1}$, i.e. $y_{g^*} = R_{r}$.
Since $|Y| = n$ and $y_{g^*}$ has order $n$, it follows that $Y = \langle R_r \rangle = R'$.
Finally, for $n = 1, 2$ there is no distinction between $e$ and $r^2$ so the same claim holds.
\end{proof}

\begin{theorem}[Permutation gauge symmetries are affine maps] \label{thm:permutation_gauge_symmetry}
A map $f \,:\, (G, A, B) \mapsto (G', A', B')$ is a permutation gauge symmetry if and only if: $G' = \varphi(G)$ where $\varphi \,:\, G \to G'$ is an isomorphism, $A' = u \varphi(A) v$, and \begin{itemize}
    \item for a symmetric mirror code, $B' = u \varphi(B) v^{-1}$ where $u \in G'$ is arbitrary and $v \in Z(G')$, and
    \item for an asymmetric mirror code, $B' = v^{-1} \varphi(B) u$ where $u \in Z(G')$ and $v \in G'$ is arbitrary.
\end{itemize}
In both cases, the permutation is of the form $\pi(g) = u \varphi(g)$.
\end{theorem}

\begin{proof}
We begin with symmetric mirror codes.
For sufficiency, note that 
\begin{align}
    \set{\vec Z(u \varphi(A) v h) \vec X(u \varphi(B) v^{-1} h^{-1})}_{h \in G'} & = \set{\vec Z(u \varphi(A) h) \vec X(u \varphi(B) h^{-1})}_{h \in G'} \\
    & = \set{\vec Z(u \varphi(A) \varphi(g)) \vec X(u \varphi(B) \varphi(g)^{-1})}_{g \in G} \\
    & = \set{\vec Z(\pi(Ag)) \vec X(\pi(Bg^{-1}))}_{g \in G}
\end{align}
The first equality relabels the stabilizer indices $h \mapsto vh$ (using the fact that $v \in Z(G)$.
The second equality defines $g := \varphi^{-1}(h)$, and re-indexes the set by elements in $G$.
The final equality implements the assumed form of $\pi$.

For necessity, note that the permutation must hold for all valid $A, B \subseteq G$, so we may consider them separately, e.g. by first setting $B = \emptyset$.
By assumption, \begin{align}
    \set{\pi(Ag) \,|\, g \in G} = \set{A' g' \,|\, g' \in G'} .
\end{align}
We aim to show that $\pi$ is essentially multiplication by some group element, up to an isomorphism.
Intuitively, this proof is relatively straightforward if across all $A, B$, the permutation of \emph{stabilizers} in the defining set of a mirror code were fixed.
However, even though we are fixing the permutation of physical qubits across $A, B$, we require only that the stabilizers are equal as sets, so from instance to instance the permutation of stabilizers could in principle change.
We show that in reality, the stabilizers must also permute the same way by comparing the symmetry to that of a simple multiplication operation.
To that end, define the right multiplier map \begin{align}
    R_g \,: G \to G ,\quad R_g(x) = xg .
\end{align}
We then define an induced map on $G'$, given by \begin{align}
    y_g \,:\, G' \to G' ,\quad y_g = \pi \circ R_g \circ \pi^{-1} .
\end{align}
We also define a right multiplier map on $G'$, given by $R'_{g'}(x) = x g'$ for $g' \in G'$.
Note that $R^{(')}_h \circ R^{(')}_g = R^{(')}_{gh}$, and thus $y_g$ satisfies an analogous property given by \begin{align}
    \label{eq:yg_composition_property}
    y_{h} \circ y_{g} = \pi \circ R_h \circ R_g \circ \pi^{-1} = \pi \circ R_{gh} \circ \pi^{-1} = y_{gh} .
\end{align}
Let $Y := \set{y_g \,|\, g \in G}$ be the set of such induced operators.
Then $Y$ is a group action of $G$ on $G'$ and is therefore itself a group.
For any subset $C \subseteq G'$, we define the orbit of $C$ under $Y$ as \begin{align}
    \operatorname{Orb}_{Y}(C) := \set{y_g(C) \,|\, g \in G} .
\end{align}
Importantly, for $C = \pi(A)$, $\operatorname{Orb}_{Y}(C) = \set{\pi(Ag) \,|\, g \in G}$.
We can compare this orbit to that under direction right multiplication $R'$ in $G'$, given by $\operatorname{Orb}_{R'}(C) := \set{R'_{g'}(C) \,|\, g' \in G'}$.
Specifically, we claim that for any $C' \subseteq G'$, \begin{align} \label{eq:orbit_equivalence}
    \operatorname{Orb}_{Y}(C) = \operatorname{Orb}_{R'}(C) .
\end{align}
To show this, let $A := \pi^{-1}(C)$ so that $\operatorname{Orb}_{Y}(C) = \set{\pi(Ag) \,|\, g \in G}$.
At the same time, by assumption \begin{align}
    \operatorname{Orb}_{Y}(C) = \set{\pi(Ag) \,|\, g \in G} = \set{A' g' \,|\, g' \in G'} = \operatorname{Orb}_{R'}(A') 
\end{align}
for some $A' \subseteq G$.
Now, since $C \in \operatorname{Orb}_{Y}(C)$, the above implies that $C \in \operatorname{Orb}_{R'}(A')$.
But orbits partition the space on which they act, which proves Eqn.~\eqref{eq:orbit_equivalence}.

Now, by \Cref{lemma:two_set_three_set}, $Y = R'$.
That is, each $y_g = \pi \circ R_g \circ \pi^{-1}$ is a right-multiplication in $G'$.
Let $\varphi \,:\, G \to G'$ encode this correspondence, namely \begin{align}
    y_g = R'_{\varphi(g)} .
\end{align}
By construction, $\varphi$ is a bijection.
Moreover, \begin{align}
    y_h \circ y_g = R'_{\varphi(h)} \circ R'_{\varphi(g)} = R'_{\varphi(g) \psi(h)}
\end{align}
by the composition property discussed above.
At the same time, by Eqn.~\eqref{eq:yg_composition_property} \begin{align}
    y_h \circ y_g & = y_{gh} = R'_{\varphi(gh)} .
\end{align}
$R'_{g} = R'_{h}$ if and only if $g = h$, and therefore the above equations imply that $\varphi(gh) = \varphi(g) \varphi(h)$, i.e. $\varphi$ is an isomorphism.
We may write $\pi \circ R_g \circ \pi^{-1} = R'_{\varphi(g)}$, so that $\pi = R'_{\varphi(g)} \circ \pi \circ R_g^{-1}$ for all $g \in G$.
Pointwise, for all $g \in G$ and $x \in G'$, \begin{align}
    \pi(x) = \pi(x g^{-1}) \varphi(g) .
\end{align}
Taking $g \mapsto g^{-1}$ and rearranging, $\pi(xg) = \pi(x) \varphi(g)$.
Finally, setting $x = e$ and defining $u := \pi(e)$,
\begin{align}
    \pi(g) = \pi(e) \varphi(g) = u \varphi(g) .
\end{align}
The only other degree of freedom in the permutation gauge symmetry beyond qubit permutation is a permutation of the stabilizer generators.
That is, a choice of bijection $\psi \,:\, G \to G$, so that \begin{align}
    \set{\vec Z(Ag) \vec X(B g^{-1})}_{g \in G} = \set{\vec Z(A \psi(g)) \vec X(B \psi(g)^{-1})}_{g \in G} .
\end{align}
To be a symmetry on $(G, A, B)$, this re-labeling must be expressible as a map on $A, B$ for any $A, B$.
That is, $A \psi(g) = A' g$ for some $A'$, and similarly for $B$.
This is only possible generically if $\psi(g) = hg$ for some $h \in G$, so that $A' = Ah$.
However, it then follows that $B \psi(g) = B g^{-1} h^{-1} = B' g^{-1}$ for some $B'$ which does not depend on $g$.
This holds only if $g^{-1} h^{-1} = f g^{-1}$ for some fixed $f$ for all $g$.
But for $g_1, g_2 \in G$, we would have $f  = g_1^{-1} h^{-1} g_1 = g_2^{-1} h^{-1} g_2$, implying that \begin{align}
    (g_2 g_1^{-1}) h^{-1} (g_2 g_1^{-1})^{-1} = h^{-1} .
\end{align}
Hence, $h^{-1}$ commutes with all elements in $G$, so $h^{-1} \in Z(G)$ (and thus $h \in Z(G)$).
In this case, $A' = Ah$ and $B' = B h^{-1}$.
Composing a stabilizer relabeling with the permutation, we obtain a map \begin{align}
    \set{\vec Z(\pi(Ahg)) \vec X(\pi(B h^{-1} g^{-1}))}_{g \in G} & = \set{\vec Z(u \varphi(A) \varphi(h) \varphi(g)) \vec X(u \varphi(B) \varphi(h)^{-1} \varphi(g)^{-1}) }_{g \in G} \\
    & = \set{\vec Z(u \varphi(A) v g') \vec X(u \varphi(B) v^{-1} (g')^{-1}) }_{g' \in G'} 
\end{align}
where $v := \varphi(h) \in Z(G')$.
We can relabel the stabilizers after the permutation too, but this just modifies the choice of $v$.
Consequently, the most general map is $A' = u \varphi(A) v$ and $B' = u \varphi(B) v^{-1}$, such that $\varphi$ is an isomorphism from $G$ to $G'$, $u \in G'$ is arbitrary, and $v \in Z(G')$ is in the center of $G'$.
This completes the proof in the symmetric case.

We next prove the asymmetric case.
Now, $u \in Z(G)$ and $v$ is arbitrary.
For sufficiency, \begin{align}
    \set{\vec Z(u \varphi(A) v h) \vec X(h^{-1} v^{-1} \varphi(B) u)}_{h \in G'} & = \set{\vec Z(u \varphi(A) v h) \vec X((vh)^{-1} \varphi(B) u)}_{h \in G'} \\
    & = \set{\vec Z(u \varphi(A) h) \vec X(h^{-1} \varphi(B) u)}_{h \in G'} \\
    & = \set{\vec Z(u \varphi(A) h) \vec X(uh^{-1}  \varphi(B))}_{h \in G'} \\
    & = \set{\vec Z(u \varphi(A) \varphi(g)) \vec X(u\varphi(g)^{-1}  \varphi(B))}_{g \in G} \\
    & = \set{\vec Z(\pi(Ag)) \vec X(\pi(g^{-1} B))}_{g \in G} .
\end{align}
The first line is algebra; the next a relabeling of indices $h \mapsto vh$.
The third line uses the centrality of $u$, and the fourth line defines $g := \varphi^{-1}(h)$.
Finally, the last line uses the assumed form of $\pi$.

For necessity, the same proof in the symmetric case for $B = \emptyset$ implies that $\pi(x) = u \varphi(x)$ where $u = \pi(e)$.
Running the same proof on $A = \emptyset$ and using left-multiplication instead of right gives $\pi(x) = \varphi(x) u$.
Hence, $u \varphi(x) = \varphi(x) u$ for all $x \in G$.
Since $\varphi$ is a bijection, this implies that $u \in Z(G')$.
On the other hand, the stabilizer relabeling is a choice of bijection $\psi$ such that \begin{align}
    \set{\vec Z(Ag) \vec X(B g^{-1})}_{g \in G} = \set{\vec Z(A \psi(g)) \vec X(\psi(g)^{-1}) B}_{g \in G} .
\end{align}
As before, this constrains $\psi(g) = hg$ for some $h \in G$, but this time on the $B$ side we obtain $\psi(g)^{-1} B = g^{-1} h^{-1} B$, so we set $B' = h^{-1} B$ and there is no requirement for $h \in Z(G)$.
That is, $h$ may be arbitrary.
This completes the proof in the asymmetric case.
\end{proof}

One may also consider gauge symmetries corresponding to a local Clifford operation rather than a permutation.
Here we will define such maps individually for each $G$, as the structure depends on the properties of $G$ itself.

\begin{definition}[LC gauge symmetries] \label{def:LC_gauge_symmetries}
    A map $f_G \,:\, (G, A, B) \mapsto (G', A', B')$ is a local-Clifford (LC) gauge symmetry if there exists a local Clifford $\vec C = \bigotimes_{i=1}^n \vec C_i \in \Cliff_n$ such that for all valid $A, B \subseteq G$ (a)symmetric mirror codes, the $f_G(G, A, B) = (G', A', B')$ (a)symmetric mirror code stabilizers are equivalent to the $(G, A, B)$ (a)symmetric mirror code stabilizers conjugated by $\vec C$, up to a phase.
\end{definition}

In what follows, let $\D$ denote the symmetric difference of two sets.
For $a \in \Z_2$, we define \begin{align}
    aA := \begin{cases}
        \emptyset & a = 0 , \\
        A & a = 1 .
    \end{cases}
\end{align}

\begin{theorem}[Classification of LC gauge symmetries] \label{thm:LC_gauge_symmetries}
A map $f_G \,:\, (G, A, B) \mapsto (G', A', B')$ is a LC gauge symmetry on a symmetric mirror code, for $G$ having at least one element of order $\geq 3$, if and only if either $\vec C = \vec I^{\otimes n}$ or $\vec C = \vec H^{\otimes n}$.
The maps can respectively be implemented as $(A', B') = (A, B)$ and $(A', B') = (B, A)$, though these choices are not unique due to the freedom of re-labeling stabilizers in a manner discussed in \Cref{thm:permutation_gauge_symmetry}.
(Here we define a Clifford only up to global phase.)
If every element of $G$ has order at most $2$, then $\vec C = \vec C_{\operatorname{loc}}^{\otimes n}$ for any $\vec C_{\operatorname{loc}} \in \Cliff_1$, and we may implement the corresponding transformation via \begin{align} \label{eq:LC_gauge_symmetric_exponent2}
    A' = a A \,\D\, cB ,\quad B' = bA \,\D\, dB ,
\end{align}
for \begin{align} \label{eq:LC_matrix}
    M := \begin{pmatrix}
        a & b \\ c & d
    \end{pmatrix} \in \operatorname{GL}(\F_2, 2) .
\end{align}
On an asymmetric mirror code with group $G$, if $G$ is abelian, then the code's LC gauge symmetries are characterized identically as in the symmetric mirror code case.
If $G$ is non-abelian, then $\vec C = \vec I^{\otimes n}$.
\end{theorem}

\begin{proof}
Consider first symmetric mirror codes with $G$ having at least one element of order $\geq 3$.
Here sufficiency can be checked by inspection.
For necessity, consider the singleton case $A = \set{e}, B = \emptyset$.
Then the stabilizers are $\set{\vec Z(g)}_{g \in G}$.
Local Cliffords cannot modify the support, and mirror codes with check weight $1$ necessarily use the same Pauli on every stabilizer.
Thus, $\vec C = \vec C_{\text{loc}}^{\otimes n}$ for some $C_{\text{loc}} \in \Cliff_1$, and this statement holds \emph{for any} $G$, regardless of element orders.
In fact, this statement holds for asymmetric mirror codes as well, since there is no distinction between symmetric and asymmetric codes when one of the subsets is empty.
Further, if $G$ has an element of order at least 3, the transformed Pauli cannot be $\vec Y$  as then the map would have to place an element in $B$ and the check weight of the stabilizers would increase since $g^{-1} \neq g$ for some $g \in G$.
Hence, $\vec C$, if not the identity, must map $\vec Z$ to $\vec X$ on every qubit, and using $A = \emptyset, B = \set{e}$ instead, we find that $\vec C$ must also map $\vec X$ to $\vec Z$.
Therefore, $\vec C = \vec H^{\otimes n}$.

If all elements of $G$ have order at most $2$, then $g^{-1} = g$ and the above argument fails.
In fact, in this case we claim that any local Clifford of the form $\vec C = \vec C_{\text{loc}}^{\otimes n}$ executes a LC gauge symmetry.
This structural form is the most general possible by the above, so we need only show sufficiency.
The Hadamard $\vec H$ and phase $\vec S$ gates generate $\Cliff_1$, and we have already shown sufficiency for $\vec H$, so we need only show sufficiency for $\vec S$.
For any $A, B \subseteq G$, \begin{align}
    \vec S^{\otimes n} \vec Z(Ag) \vec X(B g) \vec S^{\dag \otimes n} & = i^{\a(g)} \vec Z (Ag) \vec Y(B g) \\
    & = i^{\a(g)} \vec Z((A \,\D\, B) g) \vec X(B g) ,
\end{align}
where $\a \,:\, G \to \Z_4$ encodes some induced phase.
Thus, $\vec S^{\otimes n}$ maps $A' = A \,\D\, B$ and $B' = B$, where $\D$ denotes the symmetric difference of sets.
In general, the corresponding transformation matrices acting on $A, B$ in the manner of Eqn.~\eqref{eq:LC_gauge_symmetric_exponent2}, \begin{align}
    M_{\vec H}  := \begin{pmatrix}
        0 & 1 \\ 1 & 0
    \end{pmatrix} , \quad M_{\vec S} := \begin{pmatrix}
        1 & 1 \\ 1 & 0
    \end{pmatrix} ,
\end{align}
generate all of $\operatorname{GL}(\F_2, 2)$ by multiplication.
Thus, in the case $G$ has only elements of order $2$, the most general LC gauge symmetry takes the form given in Eqn.~\eqref{eq:LC_gauge_symmetric_exponent2}.

In the asymmetric case, we will prove that if there is a $f_G$ which corresponds to any non-trivial LC gauge symmetry, then $G$ must be abelian.
This will complete the proof, as when $G$ is abelian there is no distinction between the symmetric and asymmetric formulations of mirror codes.
Every single-qubit Clifford is some sequence of $\vec H$ and $\vec S$.
$\vec H$ maps $\vec Z(Ag) \vec X(g^{-1} B)$ to $\vec{Z} (g^{-1} B) \vec X(Ag)$ and $\vec S$ maps $\vec Z(Ag) \vec X(g^{-1} B)$ to $\vec Z(Ag \,\D\, g^{-1} B) \vec X(g^{-1} B)$ (all up to phase).
Recall from above that if $\vec C$ corresponds to a LC gauge symmetry then $\vec C = \vec C_{\text{loc}}^{\otimes n}$. Hence the most general transformation that a LC symmetry could execute is \begin{align}
    \vec Z(Ag) \vec X(g^{-1} B) \longrightarrow \vec Z(aAg \,\D\, cg^{-1} B) \vec X(bAg \,\D\, d g^{-1} B) ,
\end{align}
for $a, b, c, d$ as in Eqn.~\eqref{eq:LC_matrix}.
We claim that if $b = 1$ or $c = 1$, then $G$ is abelian.
(If $b = c = 0$, then $a = d = 1$ because the matrix in Eqn.~\eqref{eq:LC_matrix} is invertible; hence the transformation is the identity map.)
Suppose first that $c = 1$.
Then for $A = \emptyset$, the resulting $\vec Z$ component is $g^{-1} B$ for each $g$.
To be a valid transformation, there must be some $A'$ such that $\set{A'g}_{g \in G} = \set{g^{-1} B}_{g \in G} = \set{g B}_{g \in G}$.
Let $L_g \,:\, G \to G$ be the left-multiplier map $x \mapsto gx$, and $R_g \,:\, G \to G$ be the right-multiplier map $x \mapsto xg$.
Let $L := \set{L_g}$ and $R := \set{R_g}$.
Then the above statement is equivalent to the claim that for all $B \subseteq G$, there exists $A' \subseteq G$ such that \begin{align} \label{eq:LC_proof_orbit_firststep}
    \operatorname{Orb}_{L}(B) = \operatorname{Orb}_{R}(A') .
\end{align}
We claim also that $\operatorname{Orb}_R(A') = \operatorname{Orb}_R(B)$.
This is because $B \in \operatorname{Orb}_L(B)$, so by Eqn.~\eqref{eq:LC_proof_orbit_firststep}, $B \in \operatorname{Orb}_{R}(A')$.
But $B \in \operatorname{Orb}_R(B)$ and since orbits partition the set on which it acts---here the power set of $G$---we must have that $\operatorname{Orb}_R(A') = \operatorname{Orb}_R(B)$.
In combination with Eqn.~\eqref{eq:LC_proof_orbit_firststep}, \begin{align}
    \operatorname{Orb}_{L}(B) = \operatorname{Orb}_{R}(B) 
\end{align}
for all $B \subseteq G$.
Then by \Cref{lemma:two_set_three_set}, $L = R$.
Equivalently, there exists a permutation $\kappa \,:\, G \to G$ such that $L_g = R_{\k(g)}$.
Hence, \begin{align}
    g = ge = L_g(e) = R_{\k(g)}(e) = e \k(g) = \k(g) ,
\end{align}
so in fact $\k(g) = g$ and therefore $L_g = R_g$ for all $g \in G$.
Thus, for all $g, h \in G$, \begin{align}
    gh = L_g(h) = R_g(h) = hg ,
\end{align}
so $G$ is abelian as claimed.
Similarly, if instead $b = 1$, then we set $B = \emptyset$ and repeat an analogous argument.
\end{proof}

In this section we have shown that there are large families of symmetries and transformations that preserve the set of stabilizers of a mirror code up to ordering or up to qubit permutations.
This means that in the typical case, mirror codes will have very large groups of automorphisms.
This is very useful for fault-tolerant quantum computing
for two central reasons.
First, automorphisms of a quantum error-correcting code will preserve the set of stabilizers, but might completely change the set of logical operators of the code, meaning that automorphisms can be used to perform logical computation.
Second, many of the operations discussed can be performed \textit{tranversally}, without performing any entangling gates, meaning that they will not propagate errors and will be fault-tolerant.
For example, applying a single-qubit gate, such as a Hadamard gate tranversally to every qubit can be a fault-tolerant logical operation.
Additionally, depending on the architecture being used, permuting the qubits is as simple as classically permuting a table of addresses that indicates where the qubits are stored, such as in an ion trap.

We can use these gauge symmetries to define a canonical form for abelian mirror codes.
We do not attempt to define a canonical form for non-abelian mirror codes to the inability to compare and order elements of arbitrary groups.

\begin{definition}[Canonical form for abelian mirror codes]
\label{def:canonical_form}
    We say that an abelian mirror code $(G, A, B)$ is in canonical form if $G$ is decomposed into prime powers $p_i^{\alpha_i}$ for lexicographically ordered $(p_i,\alpha_i)$ and for lexicographically sorted and ordered $A$ and $B$, across all gauge symmetries of the code.
    Note that we sort $A$ and $B$ as lists of elements of $G=\mathbb{Z}_{p_1^{\alpha_1}}\times\mathbb{Z}_{p_2^{\alpha_2}}\times\dots\mathbb{Z}_{p_k^{\alpha_k}}$, but also swap $A$ and $B$ if necessary to sort the pair $(A, B)$, with shorter strings considered lexicographically earlier, forcing $|A|\leq|B|$.
\end{definition}

Using this canonical form makes it easier to search for mirror codes.
For example, if both $A$ and $B$ contain at least one element (which they must for the code to have a distance above 1), then $A$ must contain the element $\vec0$.
Various other simplifications can further speed up the search, such as breaking up the search by dimensions of $G$ with a common $p_i$, as all automorphisms of $G$ must be products of automorphisms of the parts with a single prime.

\subsection{CSS Mirror codes}
\label{sec:CSS_mirror_codes}

While no mirror code is a CSS code in the strict sense of having every stabilizer in its tableau being either all $\mathbf{X}$ or all $\mathbf{Z}$, there are some notions of being ``essentially CSS'' that a subset of mirror codes do satisfy.
As the majority of quantum code constructions have focused exclusively on CSS codes, it is useful to delineate precisely what subset of mirror codes are CSS, and how we may test for this. 
We here give some precise formulations of the notion of being essentially CSS, and classify mirror codes which satisfy such conditions.

\begin{definition}[Equivalently CSS] \label{def:equivalently_CSS}
A stabilizer tableau is \emph{equivalently CSS} if there exists a local Clifford operator $\mathbf{C} = \bigotimes_{i=1}^n \mathbf{C}_i$ such that after conjugation by $\mathbf{C}$, the transformed stabilizer tableau is CSS.
If there exists such a transformation wherein $\mathbf{C}$ consists entirely of Hadamard and identity operators, i.e. $\mathbf{C} = \mathbf{H}^{\mathbf{v}}$ for some $\mathbf{v} \in \Z_2^n$, then the code is \emph{equivalently CSS via Hadamards}.
\end{definition}

We first give a characterization of equivalently CSS mirror codes in terms of a 2-colouring condition on $(G, A, B)$, and then give a series of equivalent simpler characterizations of mirror codes which are equivalently CSS via Hadamards.
Recall that each stabilizer is labeled by a group element $g$, as is each qubit.
In what follows, let $\mathbf{P}_g(q) \in \set{\vec I, \vec X, \vec Y, \vec Z}$ be the Pauli applied on qubit $q \in G$ in stabilizer $\mathbf{S}(g)$.
For a general mirror code, it is difficult to give a simple characterization of being equivalently CSS.
Our characterization is written in terms of our group $G$ but does not use the actual mirror structure meaningfully, e.g. the choice of $A, B$.
In fact, the characterization will apply to any $n$-qubit stabilizer code with $n$ stabilizers specified.

\begin{theorem}[Equivalently CSS mirror codes as a 2-colouring condition] \label{thm:equivalently_CSS_2_colouring}
    Let $(G, A, B)$ form a valid (a)symmetric mirror code $C$.
    Then $C$ is equivalently CSS if and only if there exists a 2-colour function $\t \,:\, G \to \Z_2$ and injective labeling functions $\ell_q \,:\, \Z_2 \to \set{\vec X, \vec Y, \vec Z}$ such that for all $g, q \in G$, $\vec P_g(q) \in \set{\vec I, \ell_q(\t(g))}$.
\end{theorem}

\begin{proof}
We begin with sufficiency.
The criterion is stated such that the colouring function $\t$ indicates for each stabilizer whether it will be a $\vec Z$ or $\vec X$ type stabilizer once transformed into a CSS code.
The condition implies that the colouring is consistent in the sense that a single local Clifford can correctly map all stabilizers into either pure $\vec Z$ or pure $\vec X$.
We construct our local Clifford by choosing $\mathbf{C}_q$ to map $\vec \ell_q(0) \mapsto \vec Z$ and $\vec \ell_q(1) \mapsto \vec X$.
A local Clifford is uniquely specified (up to global phases) by its action on two independent Paulis, so this specification uniquely defines $\mathbf{C}_q$.
Let $\vec C := \bigotimes_{g \in G} \vec C_q$.
Suppose that $\t(g) = 0$, so that $\vec P_g(q) \in \set{\vec I, \ell_q(0)}$ for all $q \in G$.
After conjugation by $\vec C$, the new stabilizer $\vec C \vec S(g) \vec C^\dag$ is pure $\vec Z$ by construction.
Similarly, if $\t(g) = 1$, then the post-conjugation stabilizer is pure $\vec X$, giving a CSS code.

Next, for necessity, suppose that there exists $\vec C_q \in \Cliff_1$ for each $q \in G$ such that after conjugation by $\vec C := \bigotimes_{g \in G} \vec C_q$ the stabilizers are either pure $\vec Z$ or pure $\vec X$.
Let $\t \,:\, G \to \Z_2$ record which of these occurred for each stabilizer.
That is, if $\vec S(g)$ became purely $\vec Z$, then set $\t(g) = 0$, and otherwise set $\t(g) = 1$.
Next, define $\ell_q(0) = \vec C_q^\dag \vec Z \vec C_q$ and $\ell_q(1) = \vec C_q^\dag \vec X \vec C_q$.
By construction, $\ell_q$ is injective and $\vec P_g(Q) \in \set{\vec I, \ell_q(\t(g))}$.
\end{proof}

While this characterization is intuitive, by no means is it clear that it is efficiently checkable for a given code.
We claim, however, that there is a simple $\widetilde{O}(n (|A| + |B|)^2)$-time algorithm to check this characterization (assuming group operations are $\widetilde{O}(1)$ time).
As we are typically interested in codes wherein $|A|, |B|$ are small constants, this runtime is essentially linear.
Our algorithm proceeds as follows.
\begin{enumerate}
    \item[(1) ] For each $q \in G$, compute $X_q, Y_q, Z_q$, where $P_q$ is the set of stabilizers, indexed by group elements $g$, which place Pauli $\vec P$ on $q$.
    If for some $q$ all three sets are non-empty, reject.
    \item[(2) ] We will build a graph with black and red edges.
    Each node is associated with a group element $g$.
    For each $q \in G$ and for each $P \in \set{X, Y, Z}$, add a black edge between all pairs of nodes in $P_q$.
    For each $q \in G$, if exactly two of $X_q, Y_q, Z_q$ are non-empty, say $X_q$ and $Y_q$, then for each $u \in X_q$ and $v \in Y_q$ add a red edge between $u$ and $v$.
    \item[(3) ] ``Merge'' all nodes connected by a black edge.
    Each red edge containing a node in a merged set now contains the single merged node (red self-edges are also kept).
    The graph now has only red edges.
    \item[(4) ] Accept if the resulting graph is bipartite, reject otherwise.
\end{enumerate}

To compute $X_q, Y_q, Z_q$ (on, e.g., symmetric mirror codes), note that the stabilizers $\vec S(g)$ which place $\vec X$ on $q$ are precisely those which have $q \in B g^{-1}$; this is the set $q^{-1} B$.
Likewise the stabilizers which place $\vec Z$ on $q$ are those $\vec S(g)$ for $g \in A^{-1} q$.
Hence, \begin{align}
    Y_q = A^{-1} q \cap q^{-1} B ,\; X_q = q^{-1} B \setminus Y_q ,\; Z_q = A^{-1} q \setminus Y_q .
\end{align}
These each take time $\widetilde{O}(|A| \cdot |B|)$, where we use $\widetilde{O}$ to mask log factors that might arise in, e.g., the length needed to describe a group element in $G$ with $|G| = n$.
If a qubit has three distinct Paulis placed on it across stabilizers, no local Clifford can transform the stabilizers in such a way that the qubit only has $\vec X$ or $\vec Z$ placed on it across all stabilizers.
Thus, a mirror code cannot be equivalently CSS if it fails the first step.
If it succeeds, however, then for each $q$ what the image of $\ell_q$ is precisely the Paulis associated with the two non-empty constructed sets.
Next, the our two edge types encode two constraints.
The first is that if two stabilizers $\vec S(g_1)$ and $\vec S(g_2)$ place the \emph{same} Pauli $P$ on qubit $q$, then their corresponding 2-colour function values must be equal, i.e. $\t(g_1) = \t(g_2)$ (we refer to $\t(g)$ as the \emph{type} of $\vec S(g)$).
This is because the 2-colour function encodes whether the stabilizer will become $\vec X$-type or $\vec Z$-type after the local Clifford, and since the two stabilizers both place the same Pauli on some qubit $q$, they will continue to have the same Pauli on $q$ after the local Clifford conjugation.
Here, black edges connect nodes whose stabilizers must be of the same type.
At the same time, if two stabilizers place \emph{different} Paulis $P, Q$ on qubit $q$, then $\t(g_1) \neq \t(g_2)$, because a single-qubit Clifford is invertible and hence cannot map distinct Paulis to the same Pauli.
Red edges thus connect nodes whose stabilizers must be of opposite types.
Note that this edge construction takes time $\widetilde{O}(|A|^2 + |B|^2 + |A| \cdot |B|) = \widetilde{O}((|A| + |B|)^2)$ per node.
Now, we quotient out by nodes connected via black edges, which effectively enforces the equality of their types.
This final graph has nodes corresponding to sets of stabilizers, such that connected nodes must assign different types to their corresponding sets of stabilizers.
Such an assignment exists if and only if this graph is bipartite.
This argument can be readily formalized to prove that the above algorithm accepts if and only if the code is equivalently CSS; we omit this complete formalization for brevity.
The merging step can be directly in $\widetilde{O}(n (|A| + |B|)^2)$ time.
The quotiented graph has $\leq n$ nodes and $\leq n (|A| + |B|)^2$ edges.
A graph $(V, E)$ can be checked for bipartiteness via a breadth-first search in $O(|V| + |E|) = O(n (|A| + |B|)^2)$ steps using linked-list representations.
Thus, the runtime in total is $\widetilde{O}(n (|A| + |B|)^2)$ as desired.

Often, however, the notion of being equivalently CSS via an arbitrarily local Clifford is too general, and we may instead wish to restrict the local Clifford to Hadamards.
This is particularly relevant in the case of mirror codes, because it is tempting to assume from the form of the stabilizers  $\vec S(g) = \vec Z(Ag) \vec X(Bg^{-1})$ that if the code is equivalently CSS, it is equivalently CSS via Hadamards.
(We will, however, later prove a result that lends some formal credence to the intuition that the only good equivalently CSS mirror codes are equivalently CSS via Hadamards.)
This assumption is not generally true, however, and can be verified explicitly by combining a test for being equivalently CSS and a test for being equivalently CSS via Hadamards.
We next discuss simple tests for being equivalently CSS via Hadamards, beginning with a useful technical lemma.

\begin{lemma}[Hadamard-CSS applies to a union of cosets] \label{lemma:Hadamard-CSS_union_of_cosets}
    Suppose that a valid $(G, A, B)$ symmetric mirror code is equivalently CSS via Hadamards.
    Then the Hadamard transform which makes the code CSS is applied on a union of right cosets of $K_{AB} := \langle AA^{-1} \cup BB^{-1} \rangle$.
    If $(G, A, B)$ forms instead a valid asymmetric mirror code, then the transform is applied on a union of double cosets of $K^{(L)}_{A} := \langle AA^{-1} \rangle$ and $K^{(R)}_B := \langle B^{-1} B \rangle$, i.e. sets of the form $K^{(L)}_{A} g K^{(R)}_{B}$.
\end{lemma}

\begin{proof}
Consider the symmetric case first.
By assumption, there is a subset $T \subset G$ such that after Hadamarding qubits in $T$, the resultant stabilizers are either purely $\mathbf{X}$ or purely $\mathbf{Z}$.
Thus, for all $g \in G$, either $Ag \subseteq T, Bg^{-1} \subseteq G\setminus T$ or $Ag \subseteq G \setminus T, Bg^{-1} \subseteq T$.
Equivalently, for $\mathbbm{1}_{T} \,:\, G \to \set{0, 1}$ the indicator function of membership in $T$, $\mathbbm{1}_T(a_1 g) = \mathbbm{1}_T(a_2 g)$ for all $a_1, a_2 \in A$ and $g \in G$.
Letting $g \mapsto a_2^{-1} g$, $\mathbbm{1}_T(g) = \mathbbm{1}_T(a_1 a_2^{-1} g)$, i.e. \begin{align}
    \mathbbm{1}_T(g) = \mathbbm{1}_T(hg) 
\end{align}
for all $h \in AA^{-1}$ and $g \in G$.
Likewise, $\mathbbm{1}_T$ is invariant under left-multiplications by elements of $BB^{-1}$.
Consequently, $\mathbbm{1}_T$ is constant over left multiplications by elements of the subgroup $K_{AB} := \langle AA^{-1} \cup BB^{-1} \rangle$.
In other words, $\mathbbm{1}_T(g)$ depends only on which right coset of $K_{AB}$ contains $g$.
Thus, $T$ is itself a union of right cosets of $K_{AB}$.

The proof is similar in the asymmetric case, but while $\mathbbm{1}_T(a_1 g) = \mathbbm{1}_T(a_2 g)$ for all $a_1, a_2 \in A, g \in G$, now $\mathbbm{1}_T(g b_1) = \mathbbm{1}_T(g b_2)$ for all $b_1, b_2 \in B, g \in G$.
With $g \mapsto g b_1^{-1}$, we have $\mathbbm{1}_T(g) = \mathbbm{1}_T(g b_1^{-1} b_2)$.
Hence, $\mathbbm{1}_T$ is invariant under right-multiplications by elements of $B^{-1} B$.
Hence, $\mathbbm{1}_T$ is constant over left-multiplications by elements of $K^{(L)}_A$ and right-multiplications by elements of $K^{(R)}_B$.
\end{proof}

Note that if $G$ is abelian then $K_{AB} = K^{(L)}_A K^{(R)}_B$, so the symmetric and asymmetric conditions agree.
This lemma already implies a simple characterization of codes equivalently CSS via Hadamards if the group is abelian.

\begin{theorem}[Characterization of Hadamard-CSS abelian mirror codes] \label{thm:Hadamard-CSS_abelian_mirror}
Suppose that $(G, A, B)$ forms a valid mirror code $C$, and that $G$ is abelian.
The following are equivalent.
\begin{enumerate}
    \item[(1) ] $C$ is equivalently CSS via Hadamards.
    
    \item[(2) ] There exists a homomorphism \begin{align}
        \phi \,:\, G \to \Z_2
    \end{align}
    such that $\phi(a)$ is constant for all $a \in A$, $\phi(b)$ is constant for all $b \in B$, and $\phi(a) = 1 - \phi(b)$ for all $a \in A, b \in B$.
    
    \item[(3) ] There exists a subgroup $H \leq G$ of index 2, such that one coset contains $A$ and the other contains $B$.
\end{enumerate}
If $G$ is not abelian, then (2) and (3) are sufficient conditions for the $(G, A, B)$ (a)symmetric mirror code to be equivalently CSS via Hadamards.
\end{theorem}

\begin{proof}
We first prove that (1) and (2) are equivalent.
For sufficiency, suppose such a $\phi$ exists.
Let $H = \ker \phi$.
By assumption, $\phi$ is constant on $A$ and $B$, so let $\phi(A) = \set{\m}$ and $\phi(B) = \set{1-\m}$, where $\m \in \Z_2$.
Beginning with the stabilizers \begin{align}
    \mathbf{S}(g) := \mathbf{Z}(Ag) \mathbf{X}(Bg^{-1}) ,
\end{align}
We Hadamard all qubits in $H$.
Note that for all $a \in A$ and $b \in B$, \begin{align}
    \phi(ag) = \m + \phi(g) ,\; \phi(bg^{-1}) = 1 - \m - \phi(g) = 1 + \m + \phi(g) .
\end{align}
Thus, for all $g \in G$, $Ag$ and $Bg^{-1}$ are contained in distinct cosets of $H$.
That is, either $Ag \subseteq H, Bg^{-1} \subseteq G\setminus H$ or $Ag \subseteq G \setminus H, Bg^{-1} \subseteq H$.
Thus, the Hadamard on $H$ transforms each $\mathbf{S}(g)$ to be either purely $\mathbf{X}$ or purely $\mathbf{Z}$.
We observe that this argument applies similarly when $Bg^{-1} \mapsto g^{-1} B$, and does not depend on whether $G$ is abelian.

For necessity, let $T \subseteq G$ be the subset of qubits which we Hadamard to produce a CSS code.
\Cref{lemma:Hadamard-CSS_union_of_cosets} implies that $T$ is a union of right cosets of $K_{AB} = \langle AA^{-1} \cup BB^{-1} \rangle$.
Since $K_{AB} \vartriangleleft G$, we may define the quotient group $Q := G/K_{AB}$.
Since $\mathbbm{1}_T$ is constant on cosets of $K_{AB}$, there is a well-defined map $\widetilde{\mathbbm{1}}_T \,:\, Q \to \Z_2$, given by $\mathbbm{1}_T = \widetilde{\mathbbm{1}}_T \circ \mathcal{Q}$, where $\mathcal{Q} \,:\, G \to Q$ is the quotient homomorphism.
Now, note that $\mathcal{Q}(a_1) = \mathcal{Q}(a_2)$ for all $a_1, a_2 \in A$, since $\mathcal{Q}(a_1 a_2^{-1}) = e$ as $a_1 a_2^{-1} \in AA^{-1} \subseteq K_{AB}$.
Likewise, $\mathcal{Q}(b_1) = \mathcal{Q}(b_2)$ for all $b_1, b_2 \in B$.
Denote $\a := \mathcal{Q}(a)$ and $\b := \mathcal{Q}(b)$ for $a \in A, b \in B$.
Since Hadamarding all qubits in $T$ produces pure $\vec Z$ or pure $\vec X$ stabilizers, $\mathbbm{1}_T(ag) = 1 - \mathbbm{1}_T(bg^{-1})$ for all $a \in A, b \in B, g \in G$.
Over $Q$, this condition is equivalently \begin{align}
    \widetilde{\mathbbm{1}}_T(\a h) = 1 - \widetilde{\mathbbm{1}}_T(\b h^{-1})
\end{align}
for all $h \in Q$.
Taking $h \mapsto \a^{-1} h$, \begin{align}
    \widetilde{\mathbbm{1}}_T(h) = 1 - \widetilde{\mathbbm{1}}_T(\g h^{-1})
\end{align}
where $\gamma := \a \b$.
Consequently, $\g \neq x^2$ for any $x \in Q$.
If $\exists x \in Q \,:\, x^2 = c$, then with $h = x$ above, $\widetilde{\mathbbm{1}}_T(x) = 1 - \widetilde{\mathbbm{1}}_T(x)$, a contradiction.
To conclude, define the squared subgroup $Q^2 := \set{x^2 \,|\, x \in Q}$, noting that $Q^2 \leq Q$.
Note that $Q/Q^2 = \Z_2$.
Let $\mathcal{R} \,:\, Q \to Q/Q^2$ be the quotient homomorphism.
Since $\gamma \notin Q^2$, $\mathcal{R}(\g) = 1$.
Hence, $\mathcal{R}(\a \b) = \mathcal{R}(\a) + \mathcal{R}(\b) = 1$, so $\mathcal{R}(\a) = 1 - \mathcal{R}(\b)$.
Pulling back, define $\phi \,:\, G \to \Z_2$ by $\phi = \mathcal{R} \circ \mathcal{Q}$.
By the above, $\phi(a)$ is constant for $a \in A$, $\phi(b)$ is constant for $b \in B$, and $\phi(a) = \mathcal{R}(\a) = 1 - \mathcal{R}(\b) = 1 - \phi(b)$.

Next we show that (2) and (3) are equivalent.
Let $C, D$ be the cosets of the assumed subgroup $H \leq G$, such that $A \subseteq C$ and $B \subseteq D$.
Since $H$ has index 2, $G/H = \Z_2$.
Let $\phi \,:\, G \to \Z_2$ be the quotient homomorphism.
Then $\phi$ is constant over $A$ and constant over $B$, and $\phi(a) = 1 - \phi(b)$ $\forall a \in A, b \in B$ as desired.
Conversely, let $\phi \,:\, G \to \Z_2$ be a homomorphism constant on $A$, constant on $B$, and $\phi(a) = 1 - \phi(b) \,\forall a \in A, b \in B$.
Then $\ker \phi$ is a index-2 subgroup of $G$ such that $A, B$ are contained in distinct cosets as claimed.
\end{proof}

Beyond abelian groups, mirror codes equivalently CSS via Hadamards have a more intricate characterization.
This is due in part to the fact that $K_{AB}$ need not be a normal subgroup, so that a quotient group $G/K_{AB}$ is not even well-defined.
However, even when $K_{AB}$ is normal, there are mirror codes which are equivalently CSS via Hadamards but have no non-trivial homomorphism to $\Z_2$.
We give an explicit example using the special linear group in \Cref{app:counterexamples_equivalences}.

Fortunately, there is a relatively simple characterization in both the symmetric and asymmetric cases as the bipartiteness of a certain graph constructed from the cosets of the group.
The relation between CSS codes and various bipartite graphs is abundant across the study of stabilizer codes~\cite{khesin2025universal,kissinger2022phase}.

\begin{definition}[Coset constraint graphs] \label{def:coset_constraint_graph}
Let $G$ be a group and $A, B \subseteq G$.
\begin{itemize}
    \item The right coset constraint graph $\G_{GAB}$ is an undirected graph with $m := |G|/|K_{AB}|$ vertices, where $K_{AB} = \langle AA^{-1} \cup BB^{-1} \rangle$.
    Each vertex is associated with a right coset of $K_{AB}$ in $G$.
    There are $n = |G|$ edges in this graph, one for each $g \in G$, given by \begin{align}
        e_g := (K_{AB}(ag), K_{AB}(b g^{-1}))
    \end{align}
    where $a, b$ are fixed elements of $A, B$, respectively; the exact choice does not affect $e_g$.

    \item The double coset constraint graph $\G^{\mathrm{D}}_{GAB}$ is defined similarly, except that vertices are associated with a double coset $K^{(L)}_{A} g K^{(R)}_B$, where $K^{(L)}_{A} := \langle AA^{-1} \rangle$ and $K^{(R)}_B := \langle B^{-1} B \rangle$.
    Each edge is given by \begin{align}
        e^{\mathrm{D}}_g := (K^{(L)}_{A} (a g) K^{(R)}_B, K^{(L)}_{A} (g^{-1} b) K^{(R)}_B) ,
    \end{align}
    where again $a, b$ are respectively fixed elements of $A, B$ with the exact choice irrelevant.
\end{itemize}
\end{definition}

We next show that the bipartiteness of these constraint graphs precisely characterize whether a corresponding mirror code is equivalently CSS via Hadamards.

\begin{theorem}[Characterization of Hadamard-CSS mirror codes] \label{thm:Hadamard-CSS_general_mirror}
A valid $(G, A, B)$ (a)symmetric mirror code is equivalently CSS via Hadamards if and only if the right (double) coset constraint graph $\G_{GAB}$ is bipartite.
\end{theorem}

\begin{proof}
In the symmetric case, the code is equivalently CSS via Hadamards if and only if there is a subset $T \subseteq G$ such that applying $\vec H$ on each qubit in $T$ transforms every stabilizer into either pure $\vec Z$ or pure $\vec X$.
By \Cref{lemma:Hadamard-CSS_union_of_cosets}, $T$ is a union of right cosets of $K_{AB}$.
For a given stabilizer $\mathbf{S}(g) = \mathbf{Z}(Ag) \mathbf{X}(Bg^{-1})$, we must Hadamard exactly one of $Ag$ or $Bg^{-1}$.
We are then forced to Hadamard the entire coset containing either $Ag$ or $Bg^{-1}$.
If we select, e.g., the coset containing $Ag$, then we cannot select the coset containing $Bg^{-1}$, and vice versa.
This rule applies for any $g$, and corresponds exactly to a bipartition of the vertices in $\G_{GAB}$ with no edges crossing the cut.
Conversely, if the graph is bipartite, Hadamarding one of the two sides of the cut results in a CSS code by construction.
The asymmetric case proof is analogous.
\end{proof}

\subsection{Choosing the group and subsets}

It is challenging to give generic recipes for constructing mirror codes that are guaranteed to have excellent properties.
Our codes are therefore found primarily by direct search, as discussed in \Cref{sec:code_search_benchmarking}.
However, there are certain fairly general choices of $(G, A, B)$ which provably yield mirror codes with poor properties.
We here discuss some of these choices.

As a starting point, we characterize the properties of $(G, A, B)$ which give distance at most 2.
This characterization relies on an analysis of when weight-2 Paulis commute with the stabilizers of the $(G, A, B)$ mirror code.
As before, $\D$ denotes the symmetric difference of sets in what follows.

\begin{lemma}[Centralizer elements of mirror codes] \label{lemma:centralizer_mirror_stabilizers}
Let $G$ be a finite group of order $n$ and $A, B \subseteq G$.
Let $\vec P = \vec Z(S) \vec X(T)$ be a Pauli, where $S, T \subseteq G$.
Then $\vec P$ commutes with all stabilizers of the $(G, A, B)$ mirror code if and only if \begin{align} \label{eq:commuting_Pauli_symmetric}
    \Delta_{u \in S} \; u^{-1} B = \D_{u \in T} \; A^{-1} u 
\end{align}
in the symmetric formulation, and \begin{align} \label{eq:commuting_Pauli_asymmetric}
    \Delta_{u \in S} \; B u^{-1} = \D_{u \in T} \; A^{-1} u 
\end{align}
in the asymmetric formulation.
\end{lemma}

\begin{proof}
We first work with the symmetric case.
Then $\vec P$ commutes with the $(G, A, B)$ symmetric mirror code if and only if $\forall g \in G$ \begin{align}
    |S \cap Bg^{-1}| + |T \cap Ag| \equiv 0 \pmod{2}  .
\end{align}
Note that $u \in Bg^{-1}$ is equivalent to $g \in u^{-1} B$.
Likewise, $u \in Ag$ is equivalent to $g \in A^{-1} u$.
Hence, the above condition can be re-written as \begin{align}
    \sum_{u \in S} \mathbbm{1}[g \in u^{-1} B] + \sum_{u \in T} \mathbbm{1}[g \in A^{-1} u] & \equiv 0 \pmod{2} \quad \forall g \in G.
\end{align}
This expression is equivalent to \begin{align}
    \mathbbm{1}\left[g \in (\Delta_{u \in S} \; u^{-1} B) \;\D\; (\D_{u \in T} \; A^{-1} u) \right] = 0 \quad \forall g \in G.
\end{align}
This condition holds exactly when the set in question, $(\Delta_{u \in S} \; u^{-1} B) \;\D\; (\D_{u \in T} \; A^{-1} u)$ is empty.
Thus, we obtain precisely Eqn.~\eqref{eq:commuting_Pauli_symmetric}.
The asymmetric proof proceeds identically with $u^{-1} B \leftrightarrow B u^{-1}$, giving Eqn.~\eqref{eq:commuting_Pauli_asymmetric}.
\end{proof}

For example, $\vec Z(\set{u, v})$ commutes with a symmetric $(G, A, B)$ mirror code's stabilizers if and only if $u^{-1} B = v^{-1} B$; this typically would not occur unless, say, $B = \langle v u^{-1} \rangle$.
A more interesting case is a Pauli $\vec Z_u \vec X_v$, which commutes with a symmetric $(G, A, B)$ mirror code's stabilizers if and only if $u^{-1} B = A^{-1} v$.

Ideally, we could apply \Cref{lemma:centralizer_mirror_stabilizers} to the study of which choices of $A, B$ produce mirror codes with very low weight (e.g. 2) logical operators.
However, this lemma alone is insufficient, because any given Pauli in \Cref{lemma:centralizer_mirror_stabilizers} shown to commute with all stabilizers could itself be a stabilizer.
Hence, we require a further result which ensures that some low-weight Pauli not only commutes with all stabilizers, but is not itself a stabilizer.
We specialize henceforth to the weight 2 case, and when $B$ takes the form $u A^{-1} v$.

\begin{lemma}[$A$ should not be an inverted translate of $B$] \label{lemma:A_B_inverted_translate}
Let $G$ be a finite group and $A, B \subseteq G$.
If $B = u A^{-1} v$ for some $u, v \in G$, then the $(G, A, B)$ symmetric mirror code (if valid) has either logical dimension $k = 0$ or distance $d \leq 2$.
If $B = A^{-1} w$ for some $w \in G$, then the $(G, A, B)$ asymmetric mirror code (if valid) has either $k = 0$ or $d \leq 2$.
\end{lemma}

\begin{proof}
Write $w = vu$ so that in the asymmetric case, $B u^{-1} = A^{-1} v$.
By \Cref{lemma:centralizer_mirror_stabilizers}, the assumption implies that $\vec Z_u \vec X_v$ commutes with all of the stabilizers of the $(G, A, B)$ (a)symmetric mirror code.
Hence, $\vec Z_u \vec X_v$ is either a stabilizer or a logical operator.
If it is a logical operator, we are done, as $d \leq 2$.
We claim that if $\vec Z_u \vec X_v$ is a stabilizer, then the code has logical dimension $k = 0$.
If $\vec Z_u \vec X_v$ were a stabilizer, then there exists $C \subseteq G$ such that \begin{align}
    \vec Z_u \vec X_v = \prod_{g \in C} \vec Z(Ag) \vec X (B g^{-1}) = \vec Z(\D_{g \in C} \; Ag) \vec X(\D_{g \in C} \; B g^{-1})
\end{align}
in the symmetric case, with $Bg^{-1} \rightarrow g^{-1} B$ in the asymmetric case.
Thus, in either case, \begin{align}
    \set{u} = \D_{g \in C} \; Ag .
\end{align}
For any $h \in G$ then, $\set{uh} = \D_{g \in C} \; Agh = \D_{g \in Ch} Ag$.
Thus, every singleton can be expressed as the $\vec Z$ component of some stabilizer.
But $\set{\vec Z_{uh}}_{h \in G}$ span the space of $\vec Z$ operators, which implies that the stabilizer subgroup has dimension $n$.
Since the logical dimension $k$ is $n$ minus the stabilizer subgroup dimension, we have $k = 0$.
\end{proof}

A loose interpretation of this result is that when constructing a mirror code, one should avoid choosing $A, B$ which are ``reflected translations'' of each other, i.e. where $A^{-1}$ and $B$ are related by some left and right translation.
While this relation may not seem generic, it appears quite naturally in a particularly interesting class of mirror codes, namely those which are equivalently CSS yet not via Hadamards.
Such codes can be mapped to a CSS code by some local Clifford, but have at least one stabilizer with a $\vec Y$ on some qubit.

\begin{proposition}[Balanced, equivalently CSS mirror codes with a $\vec Y$ are bad] \label{prop:codes_with_Y_suck}
Let $G$ be a finite group and $A, B \subseteq G$ with $|A| = |B|$.
If $(G, A, B)$ forms a valid (a)symmetric mirror code $C$ which is equivalently CSS (see \Cref{def:equivalently_CSS}) and which has a stabilizer containing a $\vec Y$ acting on some qubit $q \in G$, then $C$ either has logical dimension $k = 0$ or has distance $\leq 2$.
\end{proposition}

\begin{proof}
We first prove the symmetric case.
Suppose on some stabilizer $\vec S(g)$ there is a Pauli $\vec Y$ acting on qubit $q \in G$.
Since the code is equivalently CSS, there are at most 2 distinct Paulis acting on $q$ across all stabilizers, and thus there cannot be both a stabilizer which places only a $\vec X$ (i.e. without a $\vec Z$) on $q$ and a stabilizer which places only a $\vec Z$ on $q$.
If no stabilizer places only $\vec X$ operators on $q$, then every time an $\vec X$ is placed on $q$ a $\vec Z$ must also be placed on $q$.
That is, if $q \in Bg^{-1}$ then $q \in Ag$.
Now, $q \in Bg^{-1}$ if and only if $g \in q^{-1} B$, and $q \in Ag$ if and only if $g \in A^{-1} q$.
So equivalently, if $g \in q^{-1} B$ then $g \in A^{-1} q$, i.e. $q^{-1} B \subseteq A^{-1} q$.
If instead no stabilizer places only $\vec Z$ on $q$, then $A^{-1} q \subseteq q^{-1} B$.
In sum, either \begin{align}
    q^{-1} B \subseteq A^{-1} q \text{ or } A^{-1} q \subseteq q^{-1} B .
\end{align}
Now, further assuming that $|A| = |B|$, we instead have an exact equivalence: \begin{align}
    A^{-1} q = q^{-1} B .
\end{align}
Applying \Cref{lemma:A_B_inverted_translate} with $u = v = q$ then completes the proof.
In the asymmetric case, we instead have $A^{-1} q = B q^{-1}$, i.e. $B = A^{-1} q^2$, applying \Cref{lemma:A_B_inverted_translate} with $w = q^2$ again completes the proof.
\end{proof}

Our no-go theorem does not apply if $|A| \neq |B|$ because the poor distance is caused by an over-abundance of symmetry between $A$ and $B$.
The choice of $A, B$ such that $|A| = |B|$ is very natural however, especially for noise models that are unbiased between $\vec Z$ and $\vec X$ errors, because equal $\vec Z$ and $\vec X$ weight intuitively ``protects'' equally well against both error types.


\subsection{Relation with other quantum LDPC codes}

A fine-grained comparison of mirror codes with its immediate relatives---two-block group algebra (2BGA) codes and its notable children---is shown in \Cref{fig:mirror_classification}.
We here justify the containments and separations given in the figure.

\begin{proposition}[Normal 2BGA codes are mirror codes]
Let $G$ be a group and $A, B \subseteq G$ form a $(G, A, B)$ 2BGA code $C_{G, A, B}$.
If $A, B$ are normal subsets, then with \begin{align}
    G' := \Z_2 \times G ,\; A' := \set{0} \times A ,\; B' := \set{1} \times B^{-1} ,
\end{align}
$(G', A', B')$ forms a valid mirror code which is equivalent to $C_{G, A, B}$ up to a Hadamard transform on a subset of qubits in $C_{G, A, B}$ and a qubit permutation.
Note that $A', B'$ are normal subsets of $G'$, and thus the $(G', A', B')$ symmetric and asymmetric mirror code constructions are identical.
If $G$ is abelian, then so is $G'$.
\end{proposition}

\begin{proof}
The stabilizers of the $(G, A, B)$ 2BGA code are given by \begin{align}
    \mathbf{S}_{\mathbf{Z}}(g) = \mathbf{Z}_L(Ag) \mathbf{Z}_R(gB) ,\; \mathbf{S}_{\mathbf{X}}(g) = \mathbf{X}_L(g B^{-1}) \mathbf{X}_R(A^{-1} g) 
\end{align}
for each $g \in G$.
Since $A, B$ are normal subsets, \Cref{lemma:normality_sufficient_AB_mirror_codes} implies that $A^{-1}, B^{-1}$ are also normal and that the $(G', A', B')$ mirror code is identical whether symmetric or asymmetric; we use the symmetric construction.
This code has stabilizers \begin{align}
    \mathbf{S}((0, g)) = \mathbf{Z}(\set{0} \times Ag) \mathbf{X}(\set{1} \times B^{-1} g^{-1}) ,\; \mathbf{S}((1, g)) = \mathbf{Z}(\set{1} \times Ag) \mathbf{X}(\set{0} \times B^{-1} g^{-1}) 
\end{align}
for each $g \in G$.
We give a sequence of transformations from the former set of stabilizers into the latter.
First, we lift $G$ to $G' = \Z_2 \times G$ and re-label qubits such that the left qubit associated with group element $g$ becomes the qubit associated with $(0, g)$; the right qubit associate with $g$ becomes $(1, g)$.
Next, we apply a Hadamard $\mathbf{H}$ on all qubits associated with group elements of the form $(1, g)$.
The 2BGA stabilizers thus become \begin{align}
    \mathbf{S}_{\mathbf{Z}}'(g) = \mathbf{Z}(\set{0} \times Ag) \mathbf{X}(\set{1} \times gB) ,\; \mathbf{S}_{\mathbf{X}}'(g) = \mathbf{X}(\set{0} \times g B^{-1}) \mathbf{Z}(\set{1} \times A^{-1} g) .
\end{align}
Finally, we apply a permutation on the physical qubits.
Presently, each qubit is labeled by the function $q(b, g) = (b, g)$.
We permute such that $q(0, g) = (0, g)$ and $q(1, g) = (1, g^{-1})$.
This permutation yields stabilizers \begin{align}
    \mathbf{S}_{\mathbf{Z}}''(g) = \mathbf{Z}(\set{0} \times Ag) \mathbf{X}(\set{1} \times B^{-1} g^{-1}) ,\; \mathbf{S}_{\mathbf{X}}''(g) = \mathbf{X}(\set{0} \times g B^{-1}) \mathbf{Z}(\set{1} \times g^{-1} A) .
\end{align}
We observe now that $\mathbf{S}_{\mathbf{Z}}''(g) = \mathbf{S}((0, g))$, and that, by normality of $A, B^{-1}$, $\mathbf{S}_{\mathbf{X}}''(g) = \mathbf{S}((1, g^{-1}))$.
Therefore the sets of stabilizers are identical.
\end{proof}

Note that a bivariate bicycle (BB) code is an abelian 2BGA code with group $G = \Z_\ell \times \Z_m$, while a generalized bicycle (GB) code is an abelian 2BGA code with group $G = \Z_\ell$.
This completes the containments shown in \Cref{fig:mirror_classification}.

Generally, a 2GBA code can be defined from a group $G$ and any two subsets $A, B$.
However, there also exists 2BGA codes which are not mirror codes of either type; we prove this result in \Cref{app:counterexamples_equivalences}.

\begin{proposition}[Not all 2BGA codes are mirror codes] \label{prop:2BGA_not_in_mirror}
    There exists $(G, A, B)$, wherein either $A$ or $B$ is not a normal subset of $G$, such that the no $(G', A', B')$ (a)symmetric mirror code yields the same set of stabilizers.
\end{proposition}

At the same time, since there exists mirror codes of either type which are not equivalently CSS, mirror codes are also not contained in 2BGA codes.

\section{Fault-Tolerant Syndrome Extraction} \label{sec:FT_SEC}

When doing syndrome extraction in practice, it is desirable to do this fault-tolerantly, where the circuit distance of the whole syndrome extraction circuit is still equal to $d$, the distance of the code.
However, fully fault-tolerant circuits might require many ancillary qubits and might lower the pseudothreshold of a code.
For sufficiently small physical error rate $p$, a fully fault-tolerant circuit will always perform better and achieve a lower logical error rate $p_L$ than a non-fault tolerant one.
However, for many practical ranges of $p$, sometimes a less fault-tolerant syndrome extraction circuit will have a lower $p_L$, due to the added errors introduced by the operations used in making the circuit fault-tolerant.
These things also come as a trade-off: one can slightly increase their circuit distance by adding some flags without making the circuits fully fault-tolerant.
This would have the advantage of being less expensive than full-fault tolerance in terms of ancilla overhead, and might be desirable depending on the size of one's architecture or their current physical error rates.

One of the most efficient ways to implement a fault-tolerant syndrome extraction circuit is to carefully order the controlled operations in a stabilizer measurement that would normally be less fault tolerant, making sure that any potential hook error that propagates to multiple qubits does so in a way that does not reduce the circuit distance.
Proving the fault tolerance of this is difficult in general, but can be done for the surface and toric codes.
A carefully-scheduled bare ancilla syndrome extraction circuit is fault tolerant for the surface code~\cite{surface2, surface3, surface4, surface1}.
For larger or higher weight codes, it can be difficult to prove that the syndrome extraction circuit has a particular circuit distance, so for a given schedule, the circuit distance can be estimated~\cite{BB}.
Even more efficient syndrome extraction circuits can exist, but they start to depend very sensitively on the exact stabilizers of the code being used.

The ultimate comparison is the tradeoff between \textit{qubit overhead per logical qubit}, and the \textit{logical error rate per logical qubit per round}.
Specifically, if one incorporates their choice of code, syndrome extraction circuit, and decoder, all for a given noise model, one can compare this tradeoff across various codes and circuits for different physical error rates.
The qubit overhead is the total number of qubits used in the entire computation, including both data qubits and check qubits.
The overhead per logical qubit is the total overhead divided by the number of logical qubits.
This quantity is relevant as it allows us to determine the total number of physical qubits needed to perform a given logical computation.
The logical error rate per logical qubit per round is a normalized quantity which computes the probability that the outcome of a given logical operator would change during a round of syndrome extraction.
This is not computed strictly by division, but is instead derived from simulated net outcomes after several rounds.

\begin{figure}[h]
    \centering
    \includegraphics[width=\linewidth]{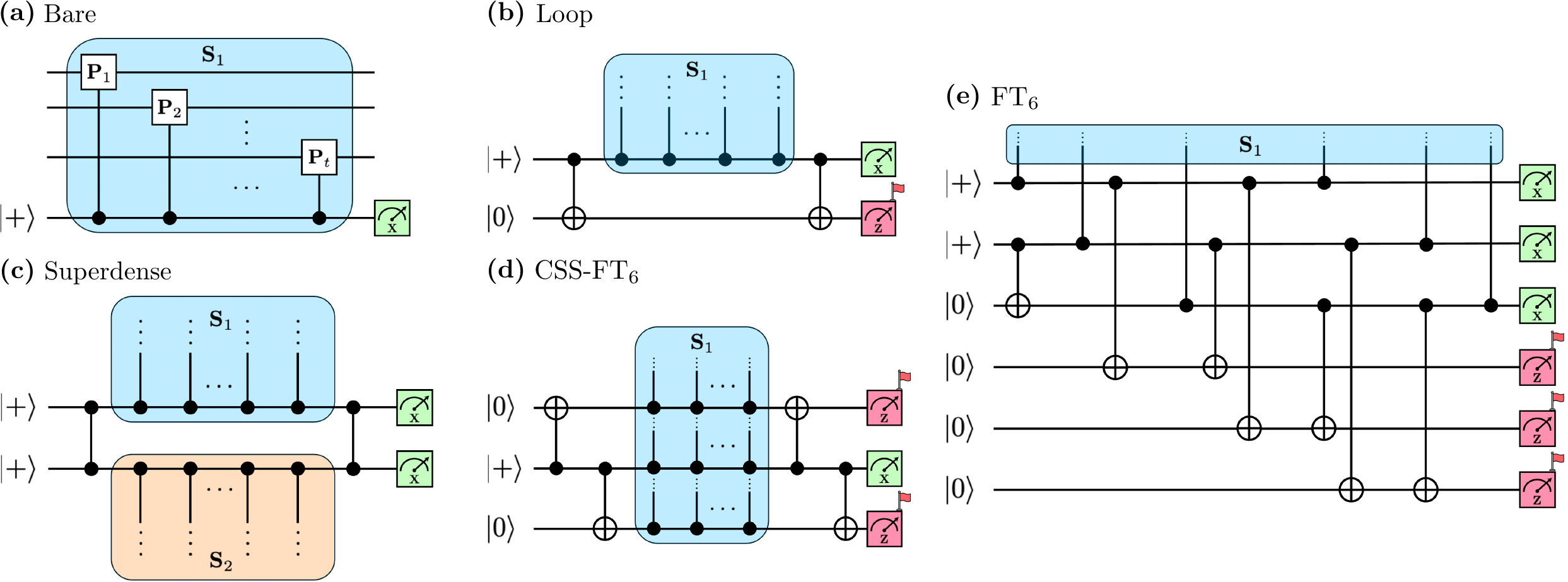}
    \caption{Five syndrome extraction circuits offering various levels of fault tolerance. Colours indicate the basis of the measurement. Measurements in the $\mathbf{X}$ basis are used to compute syndromes, while $\mathbf{Z}$ basis measurements are flags, which should return the $\bra{0}$ outcome if no faults occurred. These are also denoted with flags. \textbf{(a)} The bare ancilla syndrome extraction circuit is the most efficient but least fault tolerant circuit. \textbf{(b)} The loop circuit uses a single flag qubit to detect if any errors occur on the ancilla qubit that might propagate to multiple data qubits. In this circuit and later ones, the dots with dotted lines indicate controlled Pauli operations just like in the bare ancilla circuit. This is more fault tolerant than the bare ancilla circuit. \textbf{(c)} The superdense circuit pairs up stabilizers, using them to flag each other. If an error happens that might propagate to the qubits of one stabilizer, this will flip the other stabilizer's syndrome. The top (bottom) measurement reports the syndrome outcome for $\mathbf{S}_1$ ($\mathbf{S}_2$). \textbf{(d)} The $\text{CSS-FT}_6$ circuit uses two flag qubits and distributes the controlled operations among its three ancilla qubits. This is fault tolerant when specifically measuring an all-$\mathbf{X}$ stabilizer of a CSS code with weight $\leq6$. To measure an all-$\mathbf{Z}$ stabilizer of such a code, everything must be done in the other basis, which means swapping the controls and targets of the $\leq6$ external and 4 internal CNOT gates, and changing the state preparation and measurement bases from $\mathbf{X}$ to $\mathbf{Z}$ and vice versa. In circuits (a)-(d), the $\mathbf{X}$ measurements compute the stabilizer's syndrome. \textbf{(e)} This circuit is fully fault-tolerant for any stabilizer code of weight $\leq6$.
    The circuit has been laid out so each operation that involves a data qubit is at a different timestep.
    If it is desirable, the three flags could be measured by a single reused qubit to bring the ancilla count down from 6 to 4. The stabilizer's syndrome is computed by adding the parities of the three $\mathbf{X}$ measurements.}
    \label{fig:FTcirc}
\end{figure}

In this section, we present several circuits for syndrome extraction and discuss their overheads and fault tolerance properties.
We start with well-established circuits.
All of the circuits discussed here are shown in \cref{fig:FTcirc}.

A baseline starting point is the \textit{bare} ancilla syndrome extraction circuit.
This uses no extra two-qubit gates and 1 qubit per stabilizer.
It offers no fault tolerance properties aside from those that can be gained by a carefully chosen schedule, but such improvements are common to all the circuits we present.
Additionally, all the circuits we will discuss require us to use one controlled operation for each Pauli in the stabilizer, so we discuss only their additional two-qubit gate count, as they will all use at least $w$ such gates, where $w$ is the stabilizer weight.
The bare ancilla circuit is vulnerable to a Pauli $X$ error flipping the ancilla qubit somewhere between the $2^\text{nd}$ and $(w-1)^\text{th}$ controlled operation, and thus propagating to 2 or more data qubits.
If $w\leq3$, this is not a concern.

The next step up is adding a single flag to this circuit, resulting in a single \textit{loop}, a well-known technique for efficiently increasing the circuit distance~\cite{chamberland2018flag}.
This ancillary qubit acts as a detector that expects the measurement outcome $\bra{+}$, and will thus tell the decoder that an error has occurred if this outcome is not observed.
Thus, we would require two errors for this to go undetected and to propagate errors to multiple outputs.
This uses 2 ancilla qubits and 2 CNOT gates per stabilizer.
This is fault-tolerant for CSS codes where $w\leq5$.

An interesting variation of the loop circuit is the \textit{superdense} syndrome extraction circuit, inspired by one of the same name for colour codes~\cite{superdense}.
This circuit combines two stabilizers to measure their respective Paulis on the data qubits but also to flag any faults that occur during the computation of the paired stabilizer.
This is particularly helpful if the stabilizers overlap heavily, as they do in the colour code.
As in the loop case, this requires multiple errors to cause a fault that propagates to more data qubits.
This uses 1 ancilla qubit and 1 CZ gate per stabilizer.

We now move on to the fully fault-tolerant circuits.\footnote{To our knowledge, these constructions are original. These were developed stemming from initial discussions with Mackenzie Shaw and using the formalism of \textit{Fault Tolerance by Construction}~\cite{FTbC}.}
The fault tolerance of these circuits was proved using the methods of this formalism.
The first of these involves adding two flags to the ``main'' ancillary qubit, the one whose outcome determines the syndrome.
Here, the controlled operations are spread across the three qubit lines, and the four combinations of the two flags' outcomes will tell us which of the three qubit lines experienced a fault, if any.
If $w\leq6$ and no more than two controlled operations are placed on each qubit line, the circuit \textit{CSS-$\text{FT}_6$} will be fault-tolerant for CSS codes when measuring an all-$\mathbf{X}$ stabilizer.
This follows from the fact that the ZX-calculus version of a CNOT gate under edge flip noise is fault-equivalent to a circuit CNOT in a CSS noise model, but not under circuit-level noise~\cite{FTbC}.
To fault-tolerantly measure the all-$\mathbf{Z}$ stabilizers of a weight-6 CSS code, we simply apply a change of basis to every part of the circuit.
The CNOT gates in the middle now have the controls on the data qubits and the targets on the ancilla qubits, the 4 CNOT gates on the ancilla qubits also switch control and target qubits, the prepared states switch from $\ket0$ to $\ket+$ and vice versa, and we change the basis of the measurements to the other of $\mathbf{X}$ and $\mathbf{Z}$.
This circuit uses 4 CNOT gates and 3 ancilla qubits per stabilizer.
By exhaustive analysis, we have verified that this circuit is the optimal fault-tolerant circuit for weight 6 CSS codes in terms of both CNOT count and the number of ancilla qubits.
Note, more efficient code-specific circuits are possible, but they will not be fault-tolerant for all possible weight 6 CSS codes.

Our last circuit is a fully fault-tolerant circuit for measuring stabilizers of weight $\leq6$ for arbitrary stabilizer codes, including non-CSS ones.
This \textit{$\text{FT}_6$} circuit is laid out to avoid multiple simultaneous controlled operations on the data, as $n$ stabilizers of weight $w$ cannot be measured faster than in $w$ rounds regardless.
The syndrome of the stabilizer is determined by adding the parities of the three $\mathbf{X}$ measurements.
This circuit uses 7 CNOT gates and 6 ancilla qubits per stabilizer.
Depending on the available architecture, the ancilla count can be reduced to 4 per stabilizer by reusing the flag qubit for its three measurements.
If measurements are very slow or resonator idling noise is a concern, it might be worthwhile to separate the flags onto separate qubits.

To actually schedule all of the syndrome extraction circuits to run in parallel, we use a SAT-solver to find a low-depth solution to make all the circuits measuring stabilizers commute.
We now have all of the pieces we need to perform a detailed, end-to-end search and analysis of mirror codes.

\section{Code Search and Benchmarking}
\label{sec:code_search_benchmarking}

A naive algorithm for listing all possible abelian mirror codes is very straightforward.
For a given $|A|$ and $|B|$, we can simply iterate over $n$, find all possible ways to factor $n$ into powers of primes to create $G$, and examine all possible subsets $A,B\subseteq G$ of the appropriate size.
Even for the non-abelian case, we merely have to list all possible non-abelian groups $G$, which we do using \texttt{GAP}, and then enumerate all possible $A$ and $B$ as before.
The only other detail in the non-abelian case is to make sure to generate both the symmetric and asymmetric mirror code.

Such an algorithm will produce many equivalent mirror codes as detailed in~\cref{thm:permutation_gauge_symmetry}, to increase the maximum tractable value of $n$, we wish to skip generating as many codes as possible that would be equivalent to those we have seen before.
We used various such optimizations to perform an (almost) exhaustive search of codes with $|A|=|B|=3$ up to 300 qubits for abelian codes and 150 qubits for non-abelian ones.
For other values of $|A|$ and $|B|$, we searched up to various other values of $n$.

The full details of these optimizations are technical, but we outline them briefly here using an example in the abelian case.
Suppose we are searching for codes for the group $G=\mathbb{Z}_{2}\times\mathbb{Z}_{2}\times\mathbb{Z}_{4}\times \mathbb{Z}_{3}\times\mathbb{Z}_{3}$
First, we split $G$ into subgroups each of is a product of cyclic groups of sizes given by powers of a single prime.
In this case, these would be $\mathbb{Z}_{2}\times\mathbb{Z}_{2}\times\mathbb{Z}_{4}$ and $\mathbb{Z}_{3}\times\mathbb{Z}_{3}$.
This is done as automorphisms of abelian groups act independently on each such piece of $G$, allowing us to consider the automorphisms separately and combine the results together.
Specifically, by storing all possible candidate substrings of $A$ and $B$ for each such piece of $G$, we can reuse these candidates again for the next group which has the same factor.
For example, for the subgroup $\mathbb{Z}_{2}\times\mathbb{Z}_{2}\times\mathbb{Z}_{4}$, we would store candidate ``parts'' of $A$ and $B$, where the entries in these parts would only have 3 terms instead of 5 as expected for $G$.
If we later wish to investigate code with the group $\mathbb{Z}_{2}\times\mathbb{Z}_{2}\times\mathbb{Z}_{4}\times\mathbb{Z}_5$, the same candidate parts can be reused.

For generating each such part, we remember that we only keep a code in canonical form, that is one where the terms are lexicographically smallest.
Thus we ensure $A$ begins with $(0,0,...,0)$.
For all later terms, we explore them using a nested depth-first search.
At each step, we consider what automorphisms fix the set of terms we have in our current list, and we recurse by considering what entries are lexicographically that our current ones but also fixed or made larger by these automorphisms.
This ensures that no automorphism can make our list lexicographically smaller.
We also do this across all permutations of our final candidate list of vectors to make sure that our ordering is lexicographically minimal even when treating $A$ and $B$ as sets rather than lists.
When it comes time to add the first entry of $B$, we recognize that since $(A,B)\to(A+g,B+g)$ is a relabeling of the qubits and $(A,B)\to(A+g,B-g)$ is a relabeling of the stabilizers, we are free to adjust $B$ by any $2g$, fixing its first entry to have only 0s and 1s if the prime base of the subgroup is 2 and only 0s otherwise.
To combine these candidate parts across subgroups, we try all possible combinations, stopping early if we find ourselves building a string that is not lexicographically sorted across some corresponding permutation.

We showcase some abelian mirror codes in~\cref{tab:all_codes}, with their $(G, A, B)$ constructors and $\llbracket n,k,d\rrbracket$ code parameters, as well as whether or not they are equivalently CSS via Hadamards.

\medskip

\begingroup
\renewcommand{\tabularxcolumn}[1]{>{\centering\arraybackslash}m{#1}}
\setlength{\tabcolsep}{3pt}
\renewcommand{\arraystretch}{1.08}
\setlength{\LTpre}{0pt}
\setlength{\LTpost}{0pt}
\setlength{\LTleft}{0pt}
\setlength{\LTright}{0pt}
\footnotesize
\begin{tabularx}{\textwidth}{|>{\centering\arraybackslash}m{0.05\textwidth}|>{\centering\arraybackslash}m{0.05\textwidth}|>{\centering\arraybackslash}m{0.05\textwidth}|>{\centering\arraybackslash}m{0.08\textwidth}|>{\centering\arraybackslash}m{0.05\textwidth}|>{\centering\arraybackslash}m{0.05\textwidth}|X|X|X|}
\hline
$n$ & $k$ & $d$ & CSS? & $w_X$ & $w_Z$ & $G$ & $A$ & $B$ \\
\hline\hline
\endfirsthead
\hline
$n$ & $k$ & $d$ & CSS? & $w_X$ & $w_Z$ & $G$ & $A$ & $B$ \\
\hline\hline
\endhead
\hline
\multicolumn{9}{r}{\emph{Continued on next page}} \\
\endfoot
\endlastfoot
30 & 4 & $5$ & \texttt{Yes} & 2 & 3 & $\mathbb{Z}_{2}\times\allowbreak \mathbb{Z}_{3}\times\allowbreak \mathbb{Z}_{5}$ & $(0,0,0)$,\allowbreak\ $(0,0,1)$ & $(1,0,0)$,\allowbreak\ $(1,1,0)$,\allowbreak\ $(1,2,2)$ \\
\hline
36 & 4 & $6$ & \texttt{Yes} & 2 & 3 & $\mathbb{Z}_{4}\times\allowbreak \mathbb{Z}_{3}\times\allowbreak \mathbb{Z}_{3}$ & $(0,0,0)$,\allowbreak\ $(2,0,1)$ & $(1,0,0)$,\allowbreak\ $(1,1,0)$,\allowbreak\ $(1,2,1)$ \\
\hline\hline
16 & 4 & $4$ & \texttt{Yes} & 2 & 4 & $\mathbb{Z}_{16}$ & $(0)$,\allowbreak\ $(4)$ & $(1)$,\allowbreak\ $(3)$,\allowbreak\ $(5)$,\allowbreak\ $(11)$ \\
\hline
24 & 6 & $4$ & \texttt{Yes} & 2 & 4 & $\mathbb{Z}_{8}\times\allowbreak \mathbb{Z}_{3}$ & $(0,0)$,\allowbreak\ $(2,0)$ & $(1,0)$,\allowbreak\ $(1,1)$,\allowbreak\ $(3,0)$,\allowbreak\ $(5,1)$ \\
\hline
36 & 3 & $7$ & \texttt{No} & 2 & 4 & $\mathbb{Z}_{4}\times\allowbreak \mathbb{Z}_{9}$ & $(0,0)$,\allowbreak\ $(0,1)$ & $(0,0)$,\allowbreak\ $(1,1)$,\allowbreak\ $(2,3)$,\allowbreak\ $(3,2)$ \\
\hline
42 & 6 & $6$ & \texttt{Yes} & 2 & 4 & $\mathbb{Z}_{2}\times\allowbreak \mathbb{Z}_{3}\times\allowbreak \mathbb{Z}_{7}$ & $(0,0,0)$,\allowbreak\ $(0,0,1)$ & $(1,0,0)$,\allowbreak\ $(1,0,2)$,\allowbreak\ $(1,1,0)$,\allowbreak\ $(1,1,3)$ \\
\hline
44 & 4 & $7$ & \texttt{Yes} & 2 & 4 & $\mathbb{Z}_{4}\times\allowbreak \mathbb{Z}_{11}$ & $(0,0)$,\allowbreak\ $(0,1)$ & $(1,0)$,\allowbreak\ $(1,2)$,\allowbreak\ $(3,0)$,\allowbreak\ $(3,6)$ \\
\hline
48 & 4 & $8$ & \texttt{Yes} & 2 & 4 & $\mathbb{Z}_{16}\times\allowbreak \mathbb{Z}_{3}$ & $(0,0)$,\allowbreak\ $(4,1)$ & $(1,0)$,\allowbreak\ $(3,0)$,\allowbreak\ $(7,2)$,\allowbreak\ $(9,1)$ \\
\hline
54 & 6 & $7$ & \texttt{Yes} & 2 & 4 & $\mathbb{Z}_{2}\times\allowbreak \mathbb{Z}_{27}$ & $(0,0)$,\allowbreak\ $(0,3)$ & $(1,0)$,\allowbreak\ $(1,1)$,\allowbreak\ $(1,6)$,\allowbreak\ $(1,16)$ \\
\hline
56 & 8 & $6$ & \texttt{Yes} & 2 & 4 & $\mathbb{Z}_{8}\times\allowbreak \mathbb{Z}_{7}$ & $(0,0)$,\allowbreak\ $(0,1)$ & $(1,0)$,\allowbreak\ $(1,2)$,\allowbreak\ $(3,0)$,\allowbreak\ $(3,3)$ \\
\hline
70 & 10 & $6$ & \texttt{Yes} & 2 & 4 & $\mathbb{Z}_{2}\times\allowbreak \mathbb{Z}_{5}\times\allowbreak \mathbb{Z}_{7}$ & $(0,0,0)$,\allowbreak\ $(0,0,1)$ & $(1,0,0)$,\allowbreak\ $(1,0,2)$,\allowbreak\ $(1,1,0)$,\allowbreak\ $(1,1,3)$ \\
\hline
72 & 6 & $9$ & \texttt{Yes} & 2 & 4 & $\mathbb{Z}_{8}\times\allowbreak \mathbb{Z}_{9}$ & $(0,0)$,\allowbreak\ $(0,1)$ & $(1,0)$,\allowbreak\ $(3,0)$,\allowbreak\ $(5,1)$,\allowbreak\ $(7,6)$ \\
\hline
72 & 8 & $7$ & \texttt{Yes} & 2 & 4 & $\mathbb{Z}_{8}\times\allowbreak \mathbb{Z}_{9}$ & $(0,0)$,\allowbreak\ $(0,1)$ & $(1,0)$,\allowbreak\ $(1,2)$,\allowbreak\ $(3,0)$,\allowbreak\ $(3,4)$ \\
\hline
80 & 6 & $\leq 10$ & \texttt{Yes} & 2 & 4 & $\mathbb{Z}_{2}\times\allowbreak \mathbb{Z}_{8}\times\allowbreak \mathbb{Z}_{5}$ & $(0,0,0)$,\allowbreak\ $(1,0,1)$ & $(0,1,0)$,\allowbreak\ $(0,3,0)$,\allowbreak\ $(1,5,1)$,\allowbreak\ $(1,7,2)$ \\
\hline
90 & 10 & $7$ & \texttt{Yes} & 2 & 4 & $\mathbb{Z}_{2}\times\allowbreak \mathbb{Z}_{9}\times\allowbreak \mathbb{Z}_{5}$ & $(0,0,0)$,\allowbreak\ $(0,1,0)$ & $(1,0,0)$,\allowbreak\ $(1,0,1)$,\allowbreak\ $(1,2,0)$,\allowbreak\ $(1,4,1)$ \\
\hline
91 & 4 & $\leq 13$ & \texttt{No} & 2 & 4 & $\mathbb{Z}_{7}\times\allowbreak \mathbb{Z}_{13}$ & $(0,0)$,\allowbreak\ $(0,1)$ & $(0,0)$,\allowbreak\ $(1,0)$,\allowbreak\ $(2,3)$,\allowbreak\ $(5,4)$ \\
\hline
126 & 14 & $7$ & \texttt{Yes} & 2 & 4 & $\mathbb{Z}_{2}\times\allowbreak \mathbb{Z}_{9}\times\allowbreak \mathbb{Z}_{7}$ & $(0,0,0)$,\allowbreak\ $(0,1,0)$ & $(1,0,0)$,\allowbreak\ $(1,0,1)$,\allowbreak\ $(1,2,0)$,\allowbreak\ $(1,4,1)$ \\
\hline
132 & 8 & $\leq 11$ & \texttt{Yes} & 2 & 4 & $\mathbb{Z}_{4}\times\allowbreak \mathbb{Z}_{3}\times\allowbreak \mathbb{Z}_{11}$ & $(0,0,0)$,\allowbreak\ $(0,0,1)$ & $(1,0,0)$,\allowbreak\ $(1,1,1)$,\allowbreak\ $(3,0,2)$,\allowbreak\ $(3,1,5)$ \\
\hline
144 & 6 & $\leq 14$ & \texttt{Yes} & 2 & 4 & $\mathbb{Z}_{16}\times\allowbreak \mathbb{Z}_{9}$ & $(0,0)$,\allowbreak\ $(8,1)$ & $(1,0)$,\allowbreak\ $(3,0)$,\allowbreak\ $(5,1)$,\allowbreak\ $(15,7)$ \\
\hline
162 & 6 & $\leq 15$ & \texttt{Yes} & 2 & 4 & $\mathbb{Z}_{2}\times\allowbreak \mathbb{Z}_{81}$ & $(0,0)$,\allowbreak\ $(0,3)$ & $(1,0)$,\allowbreak\ $(1,4)$,\allowbreak\ $(1,16)$,\allowbreak\ $(1,51)$ \\
\hline
174 & 6 & $\leq 16$ & \texttt{Yes} & 2 & 4 & $\mathbb{Z}_{2}\times\allowbreak \mathbb{Z}_{3}\times\allowbreak \mathbb{Z}_{29}$ & $(0,0,0)$,\allowbreak\ $(0,0,1)$ & $(1,0,0)$,\allowbreak\ $(1,0,4)$,\allowbreak\ $(1,1,7)$,\allowbreak\ $(1,1,18)$ \\
\hline\hline
18 & 4 & $4$ & \texttt{Yes} & 3 & 3 & $\mathbb{Z}_{2}\times\allowbreak \mathbb{Z}_{9}$ & $(0,0)$,\allowbreak\ $(0,1)$,\allowbreak\ $(0,2)$ & $(1,0)$,\allowbreak\ $(1,1)$,\allowbreak\ $(1,5)$ \\
\hline
30 & 4 & $6$ & \texttt{Yes} & 3 & 3 & $\mathbb{Z}_{2}\times\allowbreak \mathbb{Z}_{3}\times\allowbreak \mathbb{Z}_{5}$ & $(0,0,0)$,\allowbreak\ $(0,1,0)$,\allowbreak\ $(0,2,1)$ & $(1,0,0)$,\allowbreak\ $(1,1,1)$,\allowbreak\ $(1,2,3)$ \\
\hline
30 & 8 & $4$ & \texttt{Yes} & 3 & 3 & $\mathbb{Z}_{2}\times\allowbreak \mathbb{Z}_{3}\times\allowbreak \mathbb{Z}_{5}$ & $(0,0,0)$,\allowbreak\ $(0,0,1)$,\allowbreak\ $(0,1,3)$ & $(1,0,0)$,\allowbreak\ $(1,0,2)$,\allowbreak\ $(1,1,1)$ \\
\hline
36 & 6 & $6$ & \texttt{No} & 3 & 3 & $\mathbb{Z}_{2}\times\allowbreak \mathbb{Z}_{2}\times\allowbreak \mathbb{Z}_{3}\times\allowbreak \mathbb{Z}_{3}$ & $(0,0,0,0)$,\allowbreak\ $(0,1,0,1)$,\allowbreak\ $(1,0,0,2)$ & $(0,0,0,0)$,\allowbreak\ $(0,1,1,0)$,\allowbreak\ $(1,1,2,0)$ \\
\hline
48 & 8 & $6$ & \texttt{No} & 3 & 3 & $\mathbb{Z}_{2}\times\allowbreak \mathbb{Z}_{2}\times\allowbreak \mathbb{Z}_{2}\times\allowbreak \mathbb{Z}_{2}\times\allowbreak \mathbb{Z}_{3}$ & $(0,0,0,0,0)$,\allowbreak\ $(0,0,0,1,1)$,\allowbreak\ $(0,0,1,0,2)$ & $(0,0,1,1,0)$,\allowbreak\ $(0,1,0,0,1)$,\allowbreak\ $(1,0,0,0,2)$ \\
\hline
56 & 6 & $8$ & \texttt{Yes} & 3 & 3 & $\mathbb{Z}_{8}\times\allowbreak \mathbb{Z}_{7}$ & $(0,0)$,\allowbreak\ $(0,1)$,\allowbreak\ $(2,3)$ & $(1,0)$,\allowbreak\ $(1,3)$,\allowbreak\ $(5,2)$ \\
\hline
60 & 4 & $10$ & \texttt{No} & 3 & 3 & $\mathbb{Z}_{2}\times\allowbreak \mathbb{Z}_{2}\times\allowbreak \mathbb{Z}_{3}\times\allowbreak \mathbb{Z}_{5}$ & $(0,0,0,0)$,\allowbreak\ $(0,0,1,0)$,\allowbreak\ $(0,0,2,1)$ & $(0,1,0,0)$,\allowbreak\ $(1,0,1,1)$,\allowbreak\ $(1,1,2,2)$ \\
\hline
72 & 8 & $8$ & \texttt{Yes} & 3 & 3 & $\mathbb{Z}_{8}\times\allowbreak \mathbb{Z}_{3}\times\allowbreak \mathbb{Z}_{3}$ & $(0,0,0)$,\allowbreak\ $(0,0,1)$,\allowbreak\ $(2,0,2)$ & $(1,0,0)$,\allowbreak\ $(1,1,0)$,\allowbreak\ $(3,2,0)$ \\
\hline
72 & 12 & $6$ & \texttt{Yes} & 3 & 3 & $\mathbb{Z}_{2}\times\allowbreak \mathbb{Z}_{4}\times\allowbreak \mathbb{Z}_{3}\times\allowbreak \mathbb{Z}_{3}$ & $(0,0,0,0)$,\allowbreak\ $(0,2,0,1)$,\allowbreak\ $(1,0,0,2)$ & $(0,1,0,0)$,\allowbreak\ $(0,3,1,0)$,\allowbreak\ $(1,1,2,0)$ \\
\hline
84 & 6 & $10$ & \texttt{Yes} & 3 & 3 & $\mathbb{Z}_{4}\times\allowbreak \mathbb{Z}_{3}\times\allowbreak \mathbb{Z}_{7}$ & $(0,0,0)$,\allowbreak\ $(0,0,1)$,\allowbreak\ $(0,1,3)$ & $(1,0,0)$,\allowbreak\ $(1,1,5)$,\allowbreak\ $(3,0,1)$ \\
\hline
85 & 8 & $9$ & \texttt{No} & 3 & 3 & $\mathbb{Z}_{5}\times\allowbreak \mathbb{Z}_{17}$ & $(0,0)$,\allowbreak\ $(0,1)$,\allowbreak\ $(1,9)$ & $(0,0)$,\allowbreak\ $(0,4)$,\allowbreak\ $(1,2)$ \\
\hline
90 & 8 & $10$ & \texttt{Yes} & 3 & 3 & $\mathbb{Z}_{2}\times\allowbreak \mathbb{Z}_{3}\times\allowbreak \mathbb{Z}_{3}\times\allowbreak \mathbb{Z}_{5}$ & $(0,0,0,0)$,\allowbreak\ $(0,0,1,0)$,\allowbreak\ $(0,0,2,1)$ & $(1,0,0,0)$,\allowbreak\ $(1,1,0,0)$,\allowbreak\ $(1,2,0,2)$ \\
\hline
96 & 16 & $6$ & \texttt{Yes} & 3 & 3 & $\mathbb{Z}_{2}\times\allowbreak \mathbb{Z}_{2}\times\allowbreak \mathbb{Z}_{2}\times\allowbreak \mathbb{Z}_{4}\times\allowbreak \mathbb{Z}_{3}$ & $(0,0,0,0,0)$,\allowbreak\ $(0,0,0,2,1)$,\allowbreak\ $(0,0,1,0,2)$ & $(0,0,0,1,0)$,\allowbreak\ $(0,1,0,1,1)$,\allowbreak\ $(1,0,0,1,2)$ \\
\hline
120 & 8 & $\leq 12$ & \texttt{Yes} & 3 & 3 & $\mathbb{Z}_{8}\times\allowbreak \mathbb{Z}_{3}\times\allowbreak \mathbb{Z}_{5}$ & $(0,0,0)$,\allowbreak\ $(0,0,1)$,\allowbreak\ $(2,1,3)$ & $(1,0,0)$,\allowbreak\ $(1,1,1)$,\allowbreak\ $(7,0,2)$ \\
\hline
126 & 12 & $10$ & \texttt{Yes} & 3 & 3 & $\mathbb{Z}_{2}\times\allowbreak \mathbb{Z}_{9}\times\allowbreak \mathbb{Z}_{7}$ & $(0,0,0)$,\allowbreak\ $(0,1,0)$,\allowbreak\ $(0,5,1)$ & $(1,0,0)$,\allowbreak\ $(1,1,1)$,\allowbreak\ $(1,6,6)$ \\
\hline
144 & 12 & $12$ & \texttt{Yes} & 3 & 3 & $\mathbb{Z}_{2}\times\allowbreak \mathbb{Z}_{8}\times\allowbreak \mathbb{Z}_{3}\times\allowbreak \mathbb{Z}_{3}$ & $(0,0,0,0)$,\allowbreak\ $(0,2,0,1)$,\allowbreak\ $(1,0,0,2)$ & $(0,1,0,0)$,\allowbreak\ $(0,3,1,0)$,\allowbreak\ $(1,5,2,0)$ \\
\hline
216 & 12 & $14$ & \texttt{Yes} & 3 & 3 & $\mathbb{Z}_{2}\times\allowbreak \mathbb{Z}_{4}\times\allowbreak \mathbb{Z}_{3}\times\allowbreak \mathbb{Z}_{3}\times\allowbreak \mathbb{Z}_{3}$ & $(0,0,0,0,0)$,\allowbreak\ $(0,2,0,0,1)$,\allowbreak\ $(1,0,0,1,0)$ & $(0,1,0,0,0)$,\allowbreak\ $(0,3,1,0,0)$,\allowbreak\ $(1,1,2,1,1)$ \\
\hline
240 & 8 & $\leq 20$ & \texttt{Yes} & 3 & 3 & $\mathbb{Z}_{16}\times\allowbreak \mathbb{Z}_{3}\times\allowbreak \mathbb{Z}_{5}$ & $(0,0,0)$,\allowbreak\ $(0,0,1)$,\allowbreak\ $(2,1,3)$ & $(1,0,0)$,\allowbreak\ $(5,0,2)$,\allowbreak\ $(11,1,1)$ \\
\hline
288 & 12 & $\leq 18$ & \texttt{Yes} & 3 & 3 & $\mathbb{Z}_{2}\times\allowbreak \mathbb{Z}_{16}\times\allowbreak \mathbb{Z}_{3}\times\allowbreak \mathbb{Z}_{3}$ & $(0,0,0,0)$,\allowbreak\ $(0,2,0,1)$,\allowbreak\ $(1,0,0,2)$ & $(0,1,0,0)$,\allowbreak\ $(0,7,1,0)$,\allowbreak\ $(1,3,2,0)$ \\
\hline\hline
24 & 6 & $4$ & \texttt{No} & 3 & 4 & $\mathbb{Z}_{2}\times\allowbreak \mathbb{Z}_{4}\times\allowbreak \mathbb{Z}_{3}$ & $(0,0,0)$,\allowbreak\ $(0,0,1)$,\allowbreak\ $(0,2,2)$ & $(0,0,0)$,\allowbreak\ $(0,1,0)$,\allowbreak\ $(1,0,0)$,\allowbreak\ $(1,3,0)$ \\
\hline
36 & 4 & $\leq 8$ & \texttt{No} & 3 & 4 & $\mathbb{Z}_{2}\times\allowbreak \mathbb{Z}_{2}\times\allowbreak \mathbb{Z}_{9}$ & $(0,0,0)$,\allowbreak\ $(0,1,1)$,\allowbreak\ $(1,0,2)$ & $(0,0,0)$,\allowbreak\ $(0,1,4)$,\allowbreak\ $(1,1,1)$,\allowbreak\ $(1,1,3)$ \\
\hline
42 & 6 & $7$ & \texttt{Yes} & 3 & 4 & $\mathbb{Z}_{2}\times\allowbreak \mathbb{Z}_{3}\times\allowbreak \mathbb{Z}_{7}$ & $(0,0,0)$,\allowbreak\ $(0,0,1)$,\allowbreak\ $(0,1,3)$ & $(1,0,0)$,\allowbreak\ $(1,0,2)$,\allowbreak\ $(1,1,1)$,\allowbreak\ $(1,2,4)$ \\
\hline
42 & 10 & $5$ & \texttt{Yes} & 3 & 4 & $\mathbb{Z}_{2}\times\allowbreak \mathbb{Z}_{3}\times\allowbreak \mathbb{Z}_{7}$ & $(0,0,0)$,\allowbreak\ $(0,1,1)$,\allowbreak\ $(0,2,3)$ & $(1,0,0)$,\allowbreak\ $(1,0,1)$,\allowbreak\ $(1,1,3)$,\allowbreak\ $(1,1,6)$ \\
\hline
42 & 12 & $4$ & \texttt{Yes} & 3 & 4 & $\mathbb{Z}_{2}\times\allowbreak \mathbb{Z}_{3}\times\allowbreak \mathbb{Z}_{7}$ & $(0,0,0)$,\allowbreak\ $(0,0,1)$,\allowbreak\ $(0,0,3)$ & $(1,0,0)$,\allowbreak\ $(1,0,1)$,\allowbreak\ $(1,1,5)$,\allowbreak\ $(1,2,5)$ \\
\hline
45 & 8 & $6$ & \texttt{No} & 3 & 4 & $\mathbb{Z}_{3}\times\allowbreak \mathbb{Z}_{3}\times\allowbreak \mathbb{Z}_{5}$ & $(0,0,0)$,\allowbreak\ $(0,0,1)$,\allowbreak\ $(0,1,3)$ & $(0,0,0)$,\allowbreak\ $(0,0,2)$,\allowbreak\ $(1,0,1)$,\allowbreak\ $(2,2,1)$ \\
\hline
48 & 10 & $6$ & \texttt{No} & 3 & 4 & $\mathbb{Z}_{4}\times\allowbreak \mathbb{Z}_{4}\times\allowbreak \mathbb{Z}_{3}$ & $(0,0,0)$,\allowbreak\ $(0,2,1)$,\allowbreak\ $(2,0,2)$ & $(0,0,0)$,\allowbreak\ $(0,1,0)$,\allowbreak\ $(1,0,0)$,\allowbreak\ $(3,3,0)$ \\
\hline
54 & 4 & $\leq 12$ & \texttt{No} & 3 & 4 & $\mathbb{Z}_{2}\times\allowbreak \mathbb{Z}_{3}\times\allowbreak \mathbb{Z}_{9}$ & $(0,0,0)$,\allowbreak\ $(0,0,1)$,\allowbreak\ $(0,1,0)$ & $(0,0,0)$,\allowbreak\ $(0,1,4)$,\allowbreak\ $(1,0,2)$,\allowbreak\ $(1,0,8)$ \\
\hline
60 & 4 & $\leq 13$ & \texttt{No} & 3 & 4 & $\mathbb{Z}_{4}\times\allowbreak \mathbb{Z}_{3}\times\allowbreak \mathbb{Z}_{5}$ & $(0,0,0)$,\allowbreak\ $(0,1,1)$,\allowbreak\ $(2,0,2)$ & $(0,0,0)$,\allowbreak\ $(1,1,1)$,\allowbreak\ $(1,2,4)$,\allowbreak\ $(2,0,4)$ \\
\hline
60 & 6 & $\leq 10$ & \texttt{No} & 3 & 4 & $\mathbb{Z}_{4}\times\allowbreak \mathbb{Z}_{3}\times\allowbreak \mathbb{Z}_{5}$ & $(0,0,0)$,\allowbreak\ $(0,1,0)$,\allowbreak\ $(0,2,1)$ & $(0,0,0)$,\allowbreak\ $(1,0,1)$,\allowbreak\ $(2,0,4)$,\allowbreak\ $(3,0,3)$ \\
\hline
60 & 8 & $\leq 9$ & \texttt{No} & 3 & 4 & $\mathbb{Z}_{2}\times\allowbreak \mathbb{Z}_{2}\times\allowbreak \mathbb{Z}_{3}\times\allowbreak \mathbb{Z}_{5}$ & $(0,0,0,0)$,\allowbreak\ $(0,1,0,1)$,\allowbreak\ $(1,0,1,3)$ & $(0,0,0,0)$,\allowbreak\ $(0,1,1,1)$,\allowbreak\ $(1,0,0,3)$,\allowbreak\ $(1,1,2,0)$ \\
\hline
62 & 5 & $\leq 13$ & \texttt{No} & 3 & 4 & $\mathbb{Z}_{2}\times\allowbreak \mathbb{Z}_{31}$ & $(0,0)$,\allowbreak\ $(0,1)$,\allowbreak\ $(1,14)$ & $(0,0)$,\allowbreak\ $(0,6)$,\allowbreak\ $(0,20)$,\allowbreak\ $(1,11)$ \\
\hline
62 & 10 & $7$ & \texttt{Yes} & 3 & 4 & $\mathbb{Z}_{2}\times\allowbreak \mathbb{Z}_{31}$ & $(0,0)$,\allowbreak\ $(0,1)$,\allowbreak\ $(0,12)$ & $(1,0)$,\allowbreak\ $(1,1)$,\allowbreak\ $(1,2)$,\allowbreak\ $(1,6)$ \\
\hline
63 & 6 & $\leq 11$ & \texttt{No} & 3 & 4 & $\mathbb{Z}_{9}\times\allowbreak \mathbb{Z}_{7}$ & $(0,0)$,\allowbreak\ $(0,1)$,\allowbreak\ $(3,3)$ & $(0,0)$,\allowbreak\ $(0,3)$,\allowbreak\ $(1,2)$,\allowbreak\ $(8,2)$ \\
\hline
66 & 4 & $\leq 14$ & \texttt{No} & 3 & 4 & $\mathbb{Z}_{2}\times\allowbreak \mathbb{Z}_{3}\times\allowbreak \mathbb{Z}_{11}$ & $(0,0,0)$,\allowbreak\ $(0,1,0)$,\allowbreak\ $(0,2,1)$ & $(0,0,0)$,\allowbreak\ $(0,0,4)$,\allowbreak\ $(1,1,1)$,\allowbreak\ $(1,1,7)$ \\
\hline
75 & 4 & $\leq 17$ & \texttt{No} & 3 & 4 & $\mathbb{Z}_{3}\times\allowbreak \mathbb{Z}_{25}$ & $(0,0)$,\allowbreak\ $(0,1)$,\allowbreak\ $(1,8)$ & $(0,0)$,\allowbreak\ $(0,2)$,\allowbreak\ $(0,16)$,\allowbreak\ $(2,11)$ \\
\hline
85 & 8 & $\leq 12$ & \texttt{No} & 3 & 4 & $\mathbb{Z}_{5}\times\allowbreak \mathbb{Z}_{17}$ & $(0,0)$,\allowbreak\ $(0,1)$,\allowbreak\ $(1,9)$ & $(0,0)$,\allowbreak\ $(0,3)$,\allowbreak\ $(0,8)$,\allowbreak\ $(2,2)$ \\
\hline
93 & 5 & $\leq 21$ & \texttt{No} & 3 & 4 & $\mathbb{Z}_{3}\times\allowbreak \mathbb{Z}_{31}$ & $(0,0)$,\allowbreak\ $(0,1)$,\allowbreak\ $(1,12)$ & $(0,0)$,\allowbreak\ $(0,4)$,\allowbreak\ $(0,23)$,\allowbreak\ $(1,26)$ \\
\hline
93 & 7 & $\leq 16$ & \texttt{No} & 3 & 4 & $\mathbb{Z}_{3}\times\allowbreak \mathbb{Z}_{31}$ & $(0,0)$,\allowbreak\ $(1,1)$,\allowbreak\ $(2,12)$ & $(0,0)$,\allowbreak\ $(0,4)$,\allowbreak\ $(2,10)$,\allowbreak\ $(2,15)$ \\
\hline
93 & 10 & $\leq 11$ & \texttt{No} & 3 & 4 & $\mathbb{Z}_{3}\times\allowbreak \mathbb{Z}_{31}$ & $(0,0)$,\allowbreak\ $(0,1)$,\allowbreak\ $(0,12)$ & $(0,0)$,\allowbreak\ $(0,4)$,\allowbreak\ $(1,18)$,\allowbreak\ $(2,18)$ \\
\hline
99 & 4 & $\leq 23$ & \texttt{No} & 3 & 4 & $\mathbb{Z}_{3}\times\allowbreak \mathbb{Z}_{3}\times\allowbreak \mathbb{Z}_{11}$ & $(0,0,0)$,\allowbreak\ $(0,1,1)$,\allowbreak\ $(0,2,3)$ & $(0,0,0)$,\allowbreak\ $(0,1,5)$,\allowbreak\ $(1,0,4)$,\allowbreak\ $(2,1,8)$ \\
\hline
99 & 6 & $\leq 19$ & \texttt{No} & 3 & 4 & $\mathbb{Z}_{9}\times\allowbreak \mathbb{Z}_{11}$ & $(0,0)$,\allowbreak\ $(3,1)$,\allowbreak\ $(6,3)$ & $(0,0)$,\allowbreak\ $(0,2)$,\allowbreak\ $(1,1)$,\allowbreak\ $(1,4)$ \\
\hline
\caption{Highlighted abelian $(G, A, B)$ mirror codes found from an exhaustive search for almost all $n\leq300$ qubits (for $|A|=|B|=3$ and smaller $n$ for other weights).
The parameters of each code are $\llbracket n, k, d \rrbracket$; entries of $d$ of the form $\leq c$ were estimated rather than computed.
$w_Z$ and $w_X$ denote the weight of the $\vec Z$ and $\vec X$ part of each stabilizer, respectively, so that the total weight of each stabilizer is at most $w_Z + w_X$.}\label{tab:all_codes}\\
\end{tabularx}
\endgroup

\begin{figure}[h!]
\makebox[\textwidth][c]{
    \centering
    \includegraphics[width=0.8\linewidth]{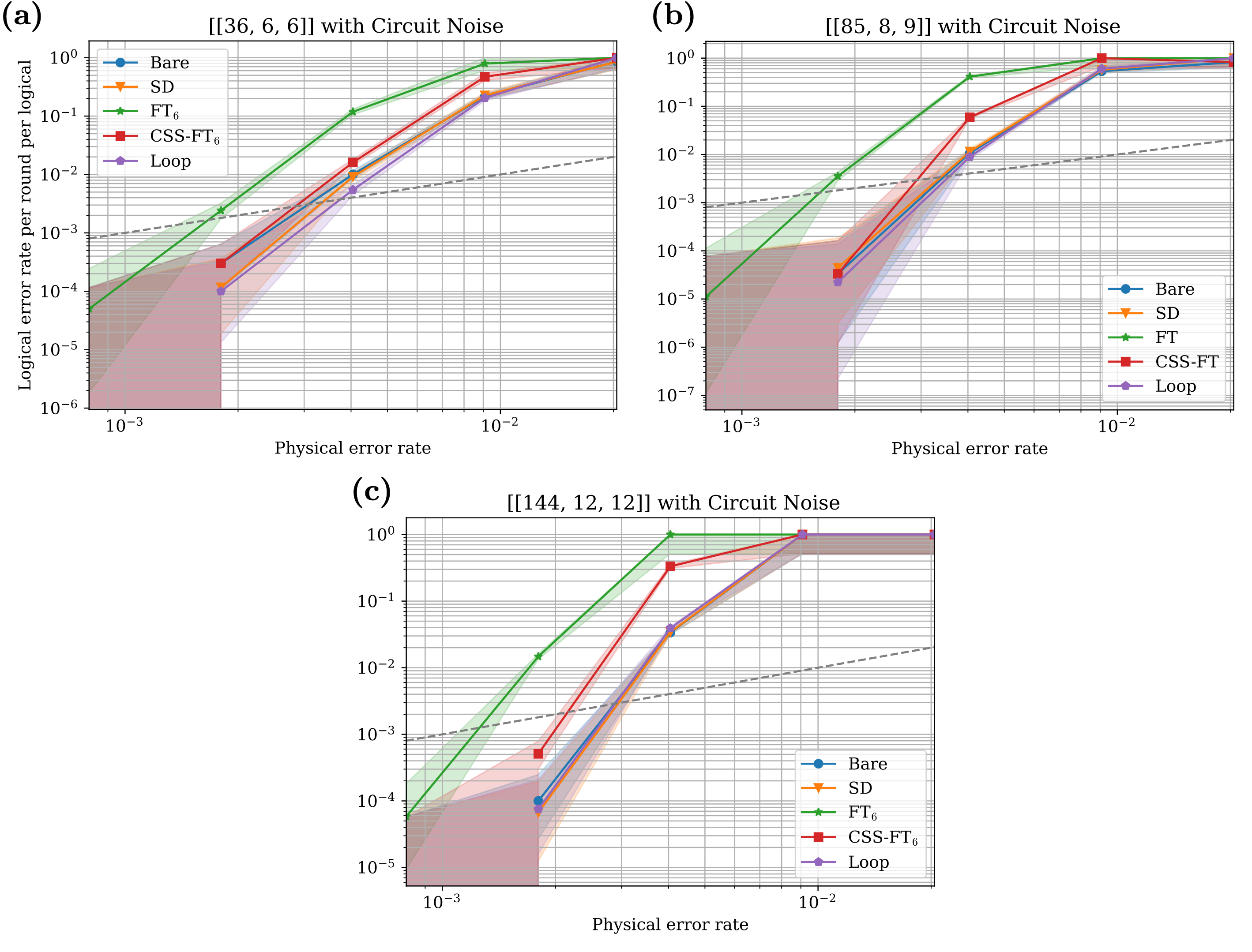}
    }
    \caption{A comparison of a several mirror codes under various syndrome extraction circuits using the \texttt{SI1000} circuit-level noise model.}
    \label{fig:plot_circuit}
\end{figure}

\begin{figure}[h!]
\makebox[\textwidth][c]{
    \centering
    \includegraphics[width=0.8\linewidth]{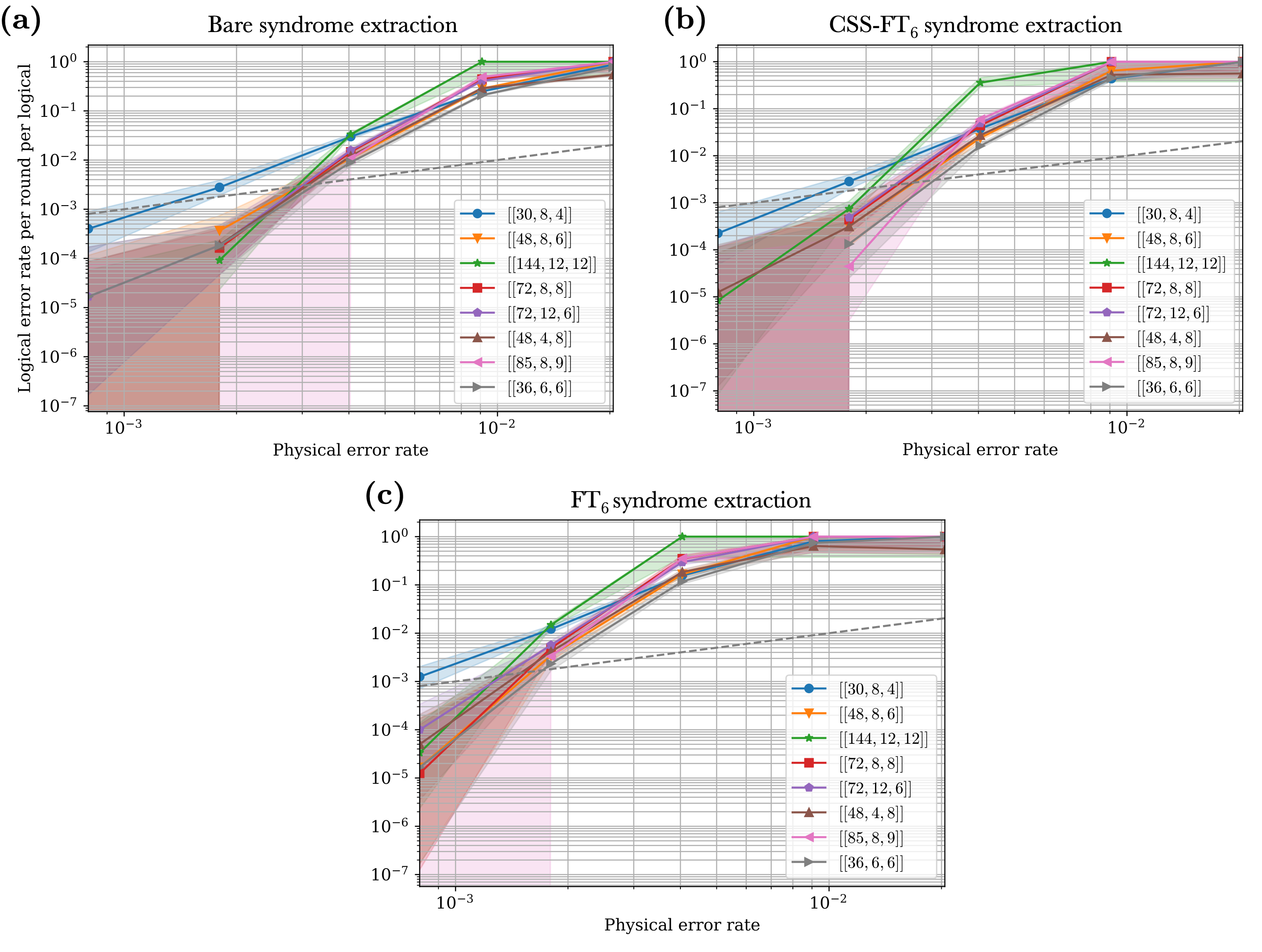}}
    \caption{A comparison of several syndrome extraction circuits for various mirror codes using the \texttt{SI1000} circuit-level noise model.}
    \label{fig:plot_codes}
\end{figure}

We benchmark the logical error rate of a few select codes under the various syndrome extraction circuits discussed in \cref{sec:FT_SEC} in \cref{fig:plot_circuit}, using the \texttt{SI1000} circuit-level noise model.
We see that most of these mirror codes have pseudothresholds of around 0.2\%, for most syndrome extraction methods.
We also see that the more fault-tolerant syndrome extraction circuits have curves that are shifted left and are most sloped, indicating that at lower physical error rates, the logical error rate might be improved with more fault-tolerant circuits.
Additionally, we compare the performance of several mirror codes with a few different syndrome extraction circuits in \cref{fig:plot_codes}.
Interestingly, we find that nearly all these codes have similar pseudothresholds, indicating that device size might be a dominant factor in code choice in the near future.

\section{Discussion and Outlook} \label{sec:outlook}

In this work, we constructed a simple family of quantum LDPC stabilizer codes which are not generically CSS.
We also devised a series of syndrome extraction circuits with provable fault tolerance guarantees that increase in generality as the number of ancillary qubits per stabilizer measurement increases; these circuits are general to quantum LDPC codes and have no particular dependence on the mirror construction.
We then performed an exhaustive search to find several novel high-rate, high-distance codes for almost all $n \leq 300$ (for $|A| = |B| = 3$, and smaller bounds for other weights), as shown in \Cref{tab:all_codes}.
Combining some of these constructions with a general-purpose decoder, we performed an end-to-end error-corrected quantum memory experiment with a standard circuit-level noise model.
These experiments demonstrated our codes to have a pseudothreshold comparable to that of state-of-the-art quantum LDPC constructions such as the $\llbracket 144, 12, 12 \rrbracket$ bivariate bicycle code, on the order of $p = 10^{-3}$.
They also numerically illustrate the comparative performance of our fault-tolerant circuit constructions.
As the fault tolerance increases, so too does the overhead.
This overhead broadens the set of possible faults, but decreases their individual probabilities of causing a logical error.
Hence, as fault tolerance and overhead increase, the logical error rates have a smaller pseudothreshold, but drop more steeply below the pseudothreshold.
Consequently, using near-term devices with sharp overhead constraints and larger physical error rates, one may wish to start by using less fault-tolerant circuits, and then adopt increasingly fault-tolerant circuits---which have higher overhead and require lower physical error rate to gain an advantage---as devices scale in size and fall in noise rates.

In sum, our findings suggest that mirror codes, in combination with fault-tolerant syndrome extraction and standard decoders, serve as an excellent candidate for near-term fault-tolerant memory realizations, especially on devices which have the connectivity to implement quantum LDPC codes with geometrically non-local checks, but which have limited physical qubit count, e.g. $n \leq 250$, and thus may require codes smaller than, for example, the $\llbracket 144, 12, 12 \rrbracket$ code, which requires 288 qubits to implement.
Our results also corroborate a more general intuition that, at small scales, stabilizer codes can dominate in performance over CSS codes, and therefore that many practically useful quantum codes lie outside of the CSS regime.

There are several interesting directions stemming from this work.
The first direction concerns questions about analyzing the code properties of mirror codes.
Although mirror codes can be well-defined for a general non-abelian group, we found through an exhaustive search of up to $n \leq 150$ (excluding $n = 128$) that every mirror code for non-abelian $G$ was dominated in parameters by an abelian mirror code.
We conjecture that this pattern extends indefinitely, a substantial departure from 2BGA code-based intuition, wherein several interesting non-abelian 2BGA codes are known~\cite{2BGA}.

\begin{conjecture}
For every $(G, A, B)$ (a)symmetric mirror code with parameters $\llbracket n, k, d \rrbracket$, where $G$ is non-abelian, there exists a $(G', A', B')$ mirror code with parameters $\llbracket n', k', d' \rrbracket$, where $G'$ is abelian, such that $n' \leq n$, $k' \geq k$, and $d' \geq d$, and at least one of the inequalities is strict.
\end{conjecture}

We have also observed that for every non-CSS mirror code we found with parameters $\llbracket n,k,d \rrbracket$ and check weight $w$, there exists also a CSS mirror code with parameters $\llbracket 4n, 2k, 2d \rrbracket$ with the same $|A|$ and $|B|$, and thus also of check weight $w$.
This correspondence is reminiscent of a known map from $\llbracket n, k, d \rrbracket$ stabilizer codes to $\llbracket 4n, 2k, 2d \rrbracket$ self-dual CSS codes~\cite{bravyi2010majorana}, except that this known map does not preserve check weight.
With a few notable exceptions, such as the $\llbracket 85, 8, 9 \rrbracket$ code, we also observed a similar correspondence wherein a non-CSS $\llbracket n, k, d \rrbracket$ mirror code with weight $w$ often implies the existence of a CSS $\llbracket 2n, 2k, d \rrbracket$ mirror code with weight $w$.
(Note that this latter correspondence has one code being non-CSS and the other being CSS; the correspondence would be trivial without this condition, as the map can be achieved by adjoining two copies of the code.)
It is an interesting open question as to under what conditions these correspondences are guaranteed.

More generally, it is an open question as to how to give sufficient conditions on $G, A, B$ which guarantee properties of a mirror code, such as rate and distance.
Such conditions may or may not depend on the fine-grained choice of $G$.
In other words, for a fixed $n$, does the exact choice of $G$ make a significant difference on the guarantees one can prove on the code?
One possible way to approach this question is by developing techniques to study a randomly chosen mirror code, either for a fixed $n$ or a fixed $G$.
While randomized analyses are most on uniformly random codes, they have also been applied to LDPC codes in classical coding theory~\cite{mosheiff2020ldpc}.
Conversely, we have here shown some simple ways in which the choice of $(G, A, B)$ can yield a mirror code with very poor properties, such as $k = 0$ or $d \leq 2$.
Adding to this list of poor $(G, A, B)$ remains a useful open question.
As a concrete example, a large fraction of groups at relevant $n$ have orders which are powers of two.
Yet almost all codes we searched where $n$ is a power of two have poor parameters: can one prove that this is always the case?
If so, the cost of exhaustive code searches would decrease by a significant constant factor solely by skipping $n$ where $n = 2^t$ for some $t \geq 1$.
Finally, we have observed that for a fixed $|A|, |B|$, there appears to be a ceiling value $g_{\max}$, such that for all $n$ sufficiently large, every $(G, A, B)$ mirror code has parameters $\llbracket n, k, d \rrbracket$ satisfying $kd/n \leq g_{\max}$.
For example, in our code search, for $|A| = 2, |B| = 3$, this ceiling is $2/3$, and for $|A| = 2, |B| = 4$, this ceiling is $6/7$, and for $|A| = |B| = 3$, this ceiling is 1.
(By increasing $|A|, |B|$, the ceiling can be made substantially larger than $1$.)
We conjecture that this behavior is universal.

\begin{conjecture}
Fix integers $s, t \geq 0$.
There exists a $g_{\max}^{(s,t)} \in \R_{\geq 0}$ and $n_0 > 1$ such that for every $n \geq n_0$, $G$ with $|G| = n$, $A, B \subseteq G$ with $|A| = s, |B| = t$, if $(G, A, B)$ forms a valid $\llbracket n, k, d \rrbracket$ (a)symmetric mirror code, then $kd/n \leq g_{\max}^{(s,t)}$.
\end{conjecture}

A second direction concerns logical computation with mirror codes.
We have here taken a first step in analytically understanding logical operations on mirror codes by classifying the effect of all permutations and local Clifford operations on $(G, A, B)$ mirror codes which map the code to another $(G', A', B')$ mirror code, including the special case of an automorphism map wherein $(G', A', B') = (G, A, B)$.
To analytically study logical operations on mirror codes, a natural next step is to characterize how such maps affect the logical basis of a mirror code.
Given the rich set of operations on mirror codes which leave the stabilizer subgroup invariant, if one can show that a subset of these operations have non-trivial logical effect, then one can utilize simple operations such as permutations to execute logical gates with essentially zero quantum cost.
Such permutation-induced logical operations have recently been studied for their potential utility in certain algorithmic applications~\cite{koh2026entangling}.
More generally, analytically obtaining the logical Pauli, Clifford, or more general basis of a given mirror code remains an important open question.

A final key direction, which is essentially independent of mirror codes, concerns our fault-tolerant syndrome extraction circuit constructions.
In principle, the superdense circuit may have an advantage over the bare circuit for some stabilizer codes, but in our numerical experiments we observe no substantial distinction in their performances against circuit-level noise.
To understand rigorously for which codes superdense circuits offer an improvement over bare circuits remains an open problem.
Moreover, our circuits in practice gain protection from errors from two distinct sources.
One is the design of the circuit itself, and the other is the \emph{scheduling} of the circuit operations.
This is because the operation scheduling impacts both how long qubits idle and which hook errors may occur.
Hence, the more operations that can be scheduled in parallel in a single time step, the less noise qubits accrue from idling, and the more clever the schedule, the less hook errors can cause decoding failures.
In this work, we scheduled our circuits by optimizing for parallelization and not hook error reduction, and we did so by reducing the scheduling problem into a SAT problem, and then applying a state-of-the-art SAT solver.
It remains open to optimize over both parallelization and hook error reduction.
From a computation standpoint, can we instead find generic or analytical scheduling solutions without a SAT solver reduction, or otherwise obtain a more efficient algorithm to schedule better?
As a concrete example of these questions, given a code, is it possible to schedule our loop syndrome extraction circuit carefully enough to be fault tolerant (i.e. where the circuit distance equals the code distance)?
In general, simply by improving scheduling, one can likely increase the error threshold of a code substantially.

\section*{Acknowledgments}
We are grateful to Victor Albert, Craig Gidney, Zhiyang (Sunny) He, Joschka Roffe, Noah Shutty, and Katherine Van Kirk for insightful discussions.
We thank Mackenzie Shaw in particular for fruitful discussions that led to the construction of the CSS-fault-tolerant syndrome extraction circuit.
ABK is supported by the Engineering and Physical Sciences Research Council grant number EP/Z002230/1: (De)constructing quantum software (DeQS).
JZL is supported in part by a National Defense Science and Engineering Graduate (NDSEG) fellowship.

\printbibliography

@article{2BGA,
  title={Quantum two-block group algebra codes},
  author={Lin, Hsiang-Ku and Pryadko, Leonid P},
  journal={Physical Review A},
  volume={109},
  number={2},
  pages={022407},
  year={2024},
  publisher={APS}
}

@article{BB,
  title={High-threshold and low-overhead fault-tolerant quantum memory},
  author={Bravyi, Sergey and Cross, Andrew W and Gambetta, Jay M and Maslov, Dmitri and Rall, Patrick and Yoder, Theodore J},
  journal={Nature},
  volume={627},
  number={8005},
  pages={778--782},
  year={2024},
  publisher={Nature Publishing Group UK London}
}

@article{koh2026entangling,
  title={Entangling logical qubits without physical operations},
  author={Koh, Jin Ming and Gong, Anqi and Diaconu, Andrei C and Tan, Daniel Bochen and Geim, Alexandra A and Gullans, Michael J and Yao, Norman Y and Lukin, Mikhail D and Majidy, Shayan},
  journal={arXiv preprint arXiv:2601.20927},
  year={2026}
}

@article{bravyi2010majorana,
  title={Majorana fermion codes},
  author={Bravyi, Sergey and Terhal, Barbara M and Leemhuis, Bernhard},
  journal={New Journal of Physics},
  volume={12},
  number={8},
  pages={083039},
  year={2010}
}

@inproceedings{mosheiff2020ldpc,
  title={LDPC codes achieve list decoding capacity},
  author={Mosheiff, Jonathan and Resch, Nicolas and Ron-Zewi, Noga and Silas, Shashwat and Wootters, Mary},
  booktitle={2020 IEEE 61st Annual Symposium on Foundations of Computer Science (FOCS)},
  pages={458--469},
  year={2020},
  organization={IEEE}
}

@article{FTbC,
  title={Fault Tolerance by Construction},
  author={Rodatz, Benjamin and Po{\'o}r, Boldizs{\'a}r and Kissinger, Aleks},
  journal={arXiv preprint arXiv:2506.17181},
  year={2025}
}

@article{superdense,
  title={New circuits and an open source decoder for the color code},
  author={Gidney, Craig and Jones, Cody},
  journal={arXiv preprint arXiv:2312.08813},
  year={2023}
}

@article{surface1,
  title={Surface code off-the-hook: diagonal syndrome-extraction scheduling},
  author={Kishony, Gilad and Fowler, Austin},
  journal={arXiv preprint arXiv:2602.09099},
  year={2026}
}

@article{surface2,
  title={Topological quantum memory},
  author={Dennis, Eric and Kitaev, Alexei and Landahl, Andrew and Preskill, John},
  journal={Journal of Mathematical Physics},
  volume={43},
  number={9},
  pages={4452--4505},
  year={2002},
  publisher={American Institute of Physics}
}

@article{surface3,
  title={Surface codes: Towards practical large-scale quantum computation},
  author={Fowler, Austin G and Mariantoni, Matteo and Martinis, John M and Cleland, Andrew N},
  journal={Physical Review A—Atomic, Molecular, and Optical Physics},
  volume={86},
  number={3},
  pages={032324},
  year={2012},
  publisher={APS}
}

@article{khesin2025universal,
  title={Universal graph representation of stabilizer codes},
  author={Khesin, Andrey Boris and Lu, Jonathan Z and Shor, Peter W},
  journal={PRX Quantum},
  volume={6},
  number={4},
  pages={040325},
  year={2025},
  publisher={APS}
}

@article{surface4,
  title={Low-distance surface codes under realistic quantum noise},
  author={Tomita, Yu and Svore, Krysta M},
  journal={Physical Review A},
  volume={90},
  number={6},
  pages={062320},
  year={2014},
  publisher={APS}
}

@article{kissinger2022phase,
  title={Phase-free ZX diagrams are CSS codes (... or how to graphically grok the surface code)},
  author={Kissinger, Aleks},
  journal={arXiv preprint arXiv:2204.14038},
  year={2022}
}

@article{gidney2022benchmarking,
  title={Benchmarking the planar honeycomb code},
  author={Gidney, Craig and Newman, Michael and McEwen, Matt},
  journal={Quantum},
  volume={6},
  pages={813},
  year={2022},
  publisher={Verein zur F{\"o}rderung des Open Access Publizierens in den Quantenwissenschaften}
}

@inproceedings{algo1,
  title={A fast quantum mechanical algorithm for database search},
  author={Grover, Lov K},
  booktitle={Proceedings of the twenty-eighth annual ACM symposium on Theory of computing},
  pages={212--219},
  year={1996}
}

@article{algo2,
  title={Polynomial-time algorithms for prime factorization and discrete logarithms on a quantum computer},
  author={Shor, Peter W},
  journal={SIAM review},
  volume={41},
  number={2},
  pages={303--332},
  year={1999},
  publisher={SIAM}
}

@article{algo3,
  title={Quantum computer systems for scientific discovery},
  author={Alexeev, Yuri and Bacon, Dave and Brown, Kenneth R and Calderbank, Robert and Carr, Lincoln D and Chong, Frederic T and DeMarco, Brian and Englund, Dirk and Farhi, Edward and Fefferman, Bill and others},
  journal={PRX quantum},
  volume={2},
  number={1},
  pages={017001},
  year={2021},
  publisher={APS}
}

@article{algo4,
  title={Quantum algorithm for linear systems of equations},
  author={Harrow, Aram W and Hassidim, Avinatan and Lloyd, Seth},
  journal={Physical review letters},
  volume={103},
  number={15},
  pages={150502},
  year={2009},
  publisher={APS}
}

@article{algo5,
  title={Universal quantum simulators},
  author={Lloyd, Seth},
  journal={Science},
  volume={273},
  number={5278},
  pages={1073--1078},
  year={1996},
  publisher={American Association for the Advancement of Science}
}

@article{algo6,
  title={Quantum advantage in learning from experiments},
  author={Huang, Hsin-Yuan and Broughton, Michael and Cotler, Jordan and Chen, Sitan and Li, Jerry and Mohseni, Masoud and Neven, Hartmut and Babbush, Ryan and Kueng, Richard and Preskill, John and others},
  journal={Science},
  volume={376},
  number={6598},
  pages={1182--1186},
  year={2022},
  publisher={American Association for the Advancement of Science}
}

@article{algo7,
  title={The vast world of quantum advantage},
  author={Huang, Hsin-Yuan and Choi, Soonwon and McClean, Jarrod R and Preskill, John},
  journal={arXiv preprint arXiv:2508.05720},
  year={2025}
}

@article{algo8,
  title={Quantum algorithms: an overview},
  author={Montanaro, Ashley},
  journal={npj Quantum Information},
  volume={2},
  number={1},
  pages={1--8},
  year={2016},
  publisher={Nature Publishing Group}
}

@article{algo9,
  title={Quantum computing in the NISQ era and beyond},
  author={Preskill, John},
  journal={Quantum},
  volume={2},
  pages={79},
  year={2018},
  publisher={Verein zur F{\"o}rderung des Open Access Publizierens in den Quantenwissenschaften}
}

@article{algo10,
  title={Quantum amplitude amplification and estimation},
  author={Brassard, Gilles and Hoyer, Peter and Mosca, Michele and Tapp, Alain},
  journal={arXiv preprint quant-ph/0005055},
  year={2000}
}

@article{algo11,
  title={Quantum algorithms for algebraic problems},
  author={Childs, Andrew M and Van Dam, Wim},
  journal={Reviews of Modern Physics},
  volume={82},
  number={1},
  pages={1--52},
  year={2010},
  publisher={APS}
}

@article{algo12,
  title={Quantum measurements and the Abelian stabilizer problem},
  author={Kitaev, A Yu},
  journal={arXiv preprint quant-ph/9511026},
  year={1995}
}

@inproceedings{algo13,
  title={Exponential algorithmic speedup by a quantum walk},
  author={Childs, Andrew M and Cleve, Richard and Deotto, Enrico and Farhi, Edward and Gutmann, Sam and Spielman, Daniel A},
  booktitle={Proceedings of the thirty-fifth annual ACM symposium on Theory of computing},
  pages={59--68},
  year={2003}
}

@article{algo14,
  title={Quantum computation by adiabatic evolution},
  author={Farhi, Edward and Goldstone, Jeffrey and Gutmann, Sam and Sipser, Michael},
  journal={arXiv preprint quant-ph/0001106},
  year={2000}
}

@article{algo15,
  title={Hamiltonian simulation by qubitization},
  author={Low, Guang Hao and Chuang, Isaac L},
  journal={Quantum},
  volume={3},
  pages={163},
  year={2019},
  publisher={Verein zur F{\"o}rderung des Open Access Publizierens in den Quantenwissenschaften}
}

@article{algo16,
  title={Simulating Hamiltonian dynamics with a truncated Taylor series},
  author={Berry, Dominic W and Childs, Andrew M and Cleve, Richard and Kothari, Robin and Somma, Rolando D},
  journal={Physical review letters},
  volume={114},
  number={9},
  pages={090502},
  year={2015},
  publisher={APS}
}

@article{algo17,
  title={Simulating physics with computers, International journal of theoretical physics},
  author={Feynman, Richard P},
  year={1982},
  publisher={Volume}
}

@inproceedings{ec1,
  title={Fault-tolerant quantum computation},
  author={Shor, Peter W},
  booktitle={Proceedings of 37th conference on foundations of computer science},
  pages={56--65},
  year={1996},
  organization={IEEE}
}

@article{ec2,
  title={Fault-tolerant quantum computation by anyons},
  author={Kitaev, A Yu},
  journal={Annals of physics},
  volume={303},
  number={1},
  pages={2--30},
  year={2003},
  publisher={Elsevier}
}

@article{ec3,
  title={Resilient quantum computation},
  author={Knill, Emanuel and Laflamme, Raymond and Zurek, Wojciech H},
  journal={Science},
  volume={279},
  number={5349},
  pages={342--345},
  year={1998},
  publisher={American Association for the Advancement of Science}
}

@inproceedings{ec4,
  title={Fault-tolerant quantum computation with constant error},
  author={Aharonov, Dorit and Ben-Or, Michael},
  booktitle={Proceedings of the twenty-ninth annual ACM symposium on Theory of computing},
  pages={176--188},
  year={1997}
}

@article{ec5,
  title={Simple quantum error-correcting codes},
  author={Steane, Andrew M},
  journal={Physical Review A},
  volume={54},
  number={6},
  pages={4741},
  year={1996},
  publisher={APS}
}

@article{ec6,
  title={Scheme for reducing decoherence in quantum computer memory},
  author={Shor, Peter W},
  journal={Physical review A},
  volume={52},
  number={4},
  pages={R2493},
  year={1995},
  publisher={APS}
}

@article{ec7,
  title={Good quantum error-correcting codes exist},
  author={Calderbank, A Robert and Shor, Peter W},
  journal={Physical Review A},
  volume={54},
  number={2},
  pages={1098},
  year={1996},
  publisher={APS}
}

@book{gottesman1997stabilizer,
  title={Stabilizer codes and quantum error correction},
  author={Gottesman, Daniel},
  year={1997},
  publisher={California Institute of Technology}
}

@article{breuckmann2021quantum,
  title={Quantum low-density parity-check codes},
  author={Breuckmann, Nikolas P and Eberhardt, Jens Niklas},
  journal={PRX quantum},
  volume={2},
  number={4},
  pages={040101},
  year={2021},
  publisher={APS}
}

@inproceedings{alg1,
  title={Quantum reed—solomon codes},
  author={Grassl, Markus and Geiselmann, Willi and Beth, Thomas},
  booktitle={International Symposium on Applied Algebra, Algebraic Algorithms, and Error-Correcting Codes},
  pages={231--244},
  year={1999},
  organization={Springer}
}

@article{alg2,
  title={Quantum reed-muller codes},
  author={Steane, Andrew M},
  journal={IEEE Transactions on Information Theory},
  volume={45},
  number={5},
  pages={1701--1703},
  year={2002},
  publisher={IEEE}
}

@article{alg3,
  title={Quantum codes with addressable and transversal non-Clifford gates},
  author={He, Zhiyang and Vaikuntanathan, Vinod and Wills, Adam and Zhang, Rachel Yun},
  journal={arXiv preprint arXiv:2502.01864},
  year={2025}
}

@article{panteleev2021degenerate,
  title={Degenerate quantum LDPC codes with good finite length performance},
  author={Panteleev, Pavel and Kalachev, Gleb},
  journal={Quantum},
  volume={5},
  pages={585},
  year={2021},
  publisher={Verein zur F{\"o}rderung des Open Access Publizierens in den Quantenwissenschaften}
}

@article{chamberland2018flag,
  title={Flag fault-tolerant error correction with arbitrary distance codes},
  author={Chamberland, Christopher and Beverland, Michael E},
  journal={Quantum},
  volume={2},
  pages={53},
  year={2018},
  publisher={Verein zur F{\"o}rderung des Open Access Publizierens in den Quantenwissenschaften}
}

@article{google2025quantum,
  title={Quantum error correction below the surface code threshold},
  journal={Nature},
  volume={638},
  number={8052},
  pages={920--926},
  year={2025},
  publisher={Nature Publishing Group UK London}
}

@article{laflamme1996perfect,
  title={Perfect quantum error correcting code},
  author={Laflamme, Raymond and Miquel, Cesar and Paz, Juan Pablo and Zurek, Wojciech Hubert},
  journal={Physical Review Letters},
  volume={77},
  number={1},
  pages={198},
  year={1996},
  publisher={APS}
}

\appendix

\section{Group Algebraic Formulation of Symmetric Mirror Codes}
\label{app:group_algebra_symmetric_mirror}

We here remark that there is a simple re-formulation for the validity characterization of a symmetric mirror code, given by \Cref{prop:mirror_code_well_defined}, in the language of group algebra.
As this language is conventional in the description of two-block group algebra codes~\cite{2BGA}, we here state this equivalent characterization for analogy with the 2BGA formulation.

Recall that the group algebra $K[G]$ of a group $G$ and field $K$ is the space of all formal sums of group elements over $K$.
That is, each element of $K[G]$ is of the form \begin{align}
    \vec{a} := \sum_{g \in G} a_g g ,
\end{align}
where $a_g \in K$.
$K[G]$ is a ring.
The addition operation is given by \begin{align}
    \vec{a} + \vec{b} := \sum_{g \in G} (a_g + b_g) g ,
\end{align}
where addition on the right-hand side is the addition operator of the field $K$, Moreover, the multiplication operation is given by \begin{align}
    \vec{a} \vec{b} := \sum_{g, h \in G} (a_g b_h) (gh) = \sum_{g \in G} \left( \sum_{h \in G} a_h b_{h^{-1} g} \right) g ,
\end{align}
where multiplication between group elements is the group operation while multiplication between field elements is the field product operation.
Given a group algebra element $\vec a \in K[G]$, we define \begin{align}
    \vec{a}^* := \sum_{g \in G} a_g g^{-1} = \sum_{g \in G} a_{g^{-1}} g .
\end{align}
Finally, for a subset $A \subseteq G$, we define the indicator group algebra element as \begin{align}
    \mathbbm{1}_A := \sum_{a \in A} a ,
\end{align}
where we have omitted explicitly writing the coefficient $1$ (the multiplicative identity in $K$).
We say that a group algebra element $\vec a$ is \emph{central} if it commutes with all group elements, i.e. $g \vec a = \vec a g$ for all $g \in G$.
Here, we are treating $g$ as an element of $K[G]$, where the coefficient of $g$ is $1$ and the coefficient of all $h \neq g$ is 0.
For purposes of this work, it suffices to set $K = \F_2$.

\begin{proposition}[Symmetric mirror validity characterization as a group algebra]
Let $G$ be a group and $A, B \subseteq G$.
A $(G, A, B)$ symmetric mirror code is well-defined if and only if $\mathbbm{1}_A^* \mathbbm{1}_B \in \F_2[G]$ is central.
\end{proposition}

\begin{proof}
From \Cref{prop:mirror_code_well_defined}, a mirror code is well-defined if and only if $K_S(g, h) = K_S(h, g) \,\forall g, h \in G$, where $K_S(g, h) := |A gh \cap B|$.
Note that $K_S(g, h)$ depends only on $gh$, so we may define $K_S(gh) := K_S(g, h)$, and the condition is equivalently $K_S(gh) = K_S(hg)$ for all $g, h \in G$.
Since $hg = h (gh) h^{-1}$, the condition is equivalent to $K_S(g) = K_S(hgh^{-1})$ for all $g, h \in G$.
In other words, $K_S \,:\, G \to \F_2$ is constant over conjugacy classes of $G$.
Now,  \begin{align}
    \mathbbm{1}_A^* \mathbbm{1}_B & = \sum_{g \in G} \left( \sum_{h \in G} \mathbbm{1}_{h \in A^{-1}} \mathbbm{1}_{h^{-1} g \in B} \right) g = \sum_{g \in G} \left( \sum_{h \in G} \mathbbm{1}_{h^{-1} \in A} \mathbbm{1}_{h^{-1} \in B g^{-1}} \right) g \\
    & = \sum_{g \in G} (|\set{A \cap B g^{-1}}|) g = \sum_{g \in G} (|\set{A g \cap B }|) g \\
    & = \sum_{g \in G} K_S(g) g .
\end{align}
Note that \begin{align}
    h \mathbbm{1}_A^* \mathbbm{1}_B h^{-1} = \sum_{g \in G} K_S(g) hgh^{-1} = \sum_{g \in G} K_S(h^{-1} g h) g .
\end{align}
Thus, $h \mathbbm{1}_A^* \mathbbm{1}_B h^{-1} = \mathbbm{1}_A^* \mathbbm{1}_B \,\forall h \in G$ if and only if $K_S(h^{-1} g h) = K_S(g)$ for all $g, h \in G$, as claimed.
\end{proof}

\section{Counterexamples and code family inequivalences} \label{app:counterexamples_equivalences}

We here give some counterexamples referenced in the main text, as well as details on showing that various code families are not equivalent.

\subsection{Counterexamples}

A simple characterization for a $(G, A, B)$ mirror code to be equivalently CSS via Hadamards when $G$ is abelian is the existence of a non-trivial homomorhpism $\phi \,:\, G \to \Z_2$.
This condition is sufficient regardless of whether $G$ is abelian, but is not necessary, even when one of the primary assumptions in the proof---that $K_{AB} := \langle AA^{-1} \cup BB^{-1} \rangle$ is a normal subgroup of $G$---holds.
We here illustrate this failure with an explicit example.

Let $G = \mathrm{SL}(2, \Z_5)$ the group of $2 \times 2$ matrices with unit determinant over $\Z_5$.
This group is non-abelian with center $Z(G) = \set{\pm I} = \Z_2$.
Moreover, $G$ is the double cover of the alternating group $A_5$, i.e. $G/Z(G) = A_5$.
Let $\mathcal{Q} \,:\, G \to A_5$ be the quotient map.
$A_5$ has 15 elements of order 2; choose one and call it $\beta$.
Choose $\widetilde{\beta} \in \mathcal{Q}^{-1}(\beta)$.
We define \begin{align}
    A := Z(G) ,\; B := Z(G) \widetilde{\b} ,
\end{align}
so that the stabilizers are given by \begin{align}
    \mathbf{S}(g)= \vec Z(Z(G) g) \vec X(Z(G) \widetilde{\b} g^{-1}) .
\end{align}
Note that $AA^{-1} = A$, $BB^{-1} = A$, so $K_{AB} = \langle A \rangle = A = Z(G) \vartriangleleft G$.
Every pair of such stabilizers commutes because the phase is given by \begin{align}
    |Z(G) g \cap Z(G) \widetilde{\b} h^{-1}| + |Z(G) h \cap Z(G) \widetilde{\b} g^{-1}| \pmod{2} .
\end{align}
In each term, the intersection is between two cosets of $Z(G)$.
Since cosets form a partition, either the two cosets are identical or disjoint.
In the former case, the set has size $|Z(G)| = 2$ and in the latter case the set has size 0, so the phase is always 0.

Now, note that $\mathcal{Q}(Ag) = \set{\mathcal{Q}(g)}$ and $\mathcal{Q}(Bg^{-1}) = \set{\b \mathcal{Q}(g)^{-1}}$.
Next, define $f \,:\, A_5 \to A_5$ by $f(h) = \b h^{-1}$.
Then $f^2(h) = f(f(h)) = \b h \b^{-1}$ so $f^4 = \text{id}$ because $\b$ has order 2.
We claim, however, that any element $h \in G/Z(G)$ has orbit size either 2 or 4 under $f$.
This is because if $f(h) = h$, then $h^2 = \b$; the same holds if $f^3(h) = h$, since $f^3(h) = f(\b h \b^{-1}) = h^{-1} \b$.
However, this would imply that $h$ has order 4, and $A_5$ has no elements of order 4.
We now choose the Hadamard subset $T'$ in quotient space greedily: pick an element $h$ in $A_5$ and place $h$ in $T'$ while placing $f(h)$ in $A_t\setminus T'$.
Since $f$ has even-sized orbits, such a choice is always possible while there remains unpartitioned elements.
To conclude, we lift our quotiented set $T'$ back into $G$, by letting $T := \mathcal{Q}^{-1}(T')$ be the Hadamarded set of qubits.
Hadamarding this set results in a CSS code if and only if for all $g \in G$, either all of $Ag$ is in $T$ or all of $Bg^{-1}$ is in $T$ but not both.
Equivalently in quotient space, either $\mathcal{Q}(g) \in T'$ or $\b \mathcal{Q}(g)^{-1}$.
This is precisely the manner in which we constructed $T'$.

At the same time, $G = \mathrm{SL}(2, \Z_5)$ is known to be perfect, i.e. $G = [G, G]$ where $[G, G] = \set{ghg^{-1}h^{-1} \,|\, g, h \in G}$.
Suppose that there is a homomorphism $\phi \,:\, G \to \Z_2$.
Then $[G, G] \subseteq \ker \phi$ since $\Z_2$ is abelian.
Thus, $G = \ker \phi$, so $\phi(g) = 0 \,\forall g \in G$, and $\phi$ is a trivial homomorphism.

In sum, $(G, A, B)$ forms a valid symmetric mirror code which (a) is equivalently CSS via Hadamards, (b) has $K_{AB}$ as a normal subgroup, yet (c) has no non-trivial homomorphism $G \to \Z_2$.

\subsection{Inequivalence of code families}

Next we recall a simple code property which is invariant under local Clifford operations and qubit permutations.
Using this quantity, we show that symmetric and asymmetric mirror codes are distinct families (i.e. neither one contains the other), and that 2BGA codes are not contained in the set of mirror codes.

\begin{definition}[Stabilizer weight enumerator] \label{def:weight_enum}
    Let $\CS \leq \CP_n$ be a stabilizer subgroup.
    The weight enumerator of $\CS$ is a polynomial \begin{align}
        W(\CS) := \sum_{j=0}^n w_j x^j ,
    \end{align}
    where $w_j$ is the number of elements in $\CS$ with weight $j$.
\end{definition}

The weight enumerator is a strong invariant of the stabilizer code associated with $\CS$ because it is invariant under choice of stabilizer tableau, qubit permutation, and local Clifford operations.
Thus, we can show that two code families are inequivalent (up to choice of tableau, permutations, and local Cliffords) by identifying a code in one family with a weight enumerator that does not match that of any code in the other family.
We did so numerically to prove the following theorem.

\begin{theorem}[Code family separations] \label{thm:code_family_inequivalences}
The following code family inequivalences hold.
\begin{enumerate}
    \item[(1) ] The set of all symmetric mirror codes is not contained in the set of all asymmetric mirror codes.
    \item[(2) ] The set of all asymmetric mirror codes is not contained in the set of all symmetric mirror codes.
    \item[(3) ] The set of all 2BGA codes is not contained in the set of all mirror codes (symmetric and asymmetric).
\end{enumerate}
\end{theorem}

\end{document}